\documentclass[runningheads]{llncs}

\usepackage{booktabs}   %% For formal tables:
\usepackage{xcolor}
\usepackage[ligature, inference]{semantic}
\usepackage{amsmath}
\usepackage{graphicx}
\usepackage{xspace}
\usepackage{listings}
\usepackage{color, colortbl}
\usepackage{textcomp}
\usepackage{multirow}
\usepackage{hyperref}
\usepackage[ruled, vlined, linesnumbered]{algorithm2e}
\usepackage{enumitem}
\usepackage[all]{xy}
\usepackage[firstpage]{draftwatermark}

\begin{document}

\title{Stateful Dynamic Partial Order Reduction for Model Checking Event-Driven Applications that Do Not Terminate}

\author{Rahmadi Trimananda\inst{1}\and
Weiyu Luo\inst{1}\and
Brian Demsky\inst{1}\and
Guoqing Harry Xu\inst{2}}
\authorrunning{R. Trimananda et al.}
\institute{University of California, Irvine, USA\\
\email{\{rtrimana,weiyul7,bdemsky\}@uci.edu} \and
University of California, Los Angeles, USA \\
\email{harryxu@g.ucla.edu}}
%

%%%%\author{Lars Th{\o}rv{\"a}ld}
%%%%\affiliation{%
%%%%  \institution{The Th{\o}rv{\"a}ld Group}
%%%%  \streetaddress{1 Th{\o}rv{\"a}ld Circle}
%%%%  \city{Hekla}
%%%%  \country{Iceland}}
%%%%\email{larst@affiliation.org}


\lstset{
        % backgroundcolor=\color{lbcolor},
        % rulecolor=,
        % aboveskip={1.5\baselineskip},
        % extendedchars=true,
        tabsize=3,
        language=Java,
        morekeywords={@config},
        basicstyle=\ttfamily\scriptsize,
        %upquote=true,
        columns=fixed,
        showstringspaces=false,
        breaklines=true,
        prebreak = \raisebox{0ex}[0ex][0ex]{\ensuremath{\hookleftarrow}},
        % frame=single,
        showtabs=false,
        showspaces=false,
        showstringspaces=false,
        identifierstyle=\ttfamily,
        keywordstyle=\color[rgb]{0,0,1},
        commentstyle=\color[rgb]{0.133,0.545,0.133},
        stringstyle=\color[rgb]{0.627,0.126,0.941},
        numbers=left,
        numberstyle=\tiny,
        numbersep=8pt,
        escapeinside={/*@}{@*/}
}

\hypersetup{
  colorlinks,
  linkcolor={red!50!black},
  citecolor={blue!50!black},
  urlcolor={blue!80!black}
}

\newcommand{\cf}{\hbox{\emph{cf.}}\xspace}
\newcommand{\etal}{\hbox{\emph{et al.}}\xspace}
\newcommand{\eg}{\hbox{\emph{e.g.},}\xspace}
\newcommand{\ie}{\hbox{\emph{i.e.},}\xspace}
\newcommand{\aka}{\hbox{\emph{aka.}}\xspace}
\newcommand{\st}{\hbox{\emph{s.t.}}\xspace}
\newcommand{\wrt}{\hbox{\emph{w.r.t.}}\xspace}
\newcommand{\etc}{\hbox{\emph{etc.}}\xspace}
\newcommand{\viz}{\hbox{\emph{viz.}}\xspace}

%Algorithm
\newcommand{\Summary}[1]{\hbox{$\mathcal{M}_{#1}$}\xspace}
\newcommand{\StateSum}[1]{\hbox{$m_{#1}$}\xspace}
\newcommand{\Access}[1]{\hbox{$\mathcal{A}_{#1}$}\xspace}
\newcommand{\RGraph}[1]{\hbox{$\mathcal{R}_{#1}$}\xspace}
\newcommand{\oRGraph}[1]{\hbox{$\overline{\mathcal{R}_{#1}}$}\xspace}
\newcommand{\Trans}[1]{\hbox{$t_{#1}$}\xspace}
\newcommand{\ThreadID}{\hbox{\textit{tid}}\xspace}
\newcommand\mycommfont[1]{\footnotesize\ttfamily\textcolor{blue}{#1}}
\SetCommentSty{mycommfont}
\SetKwComment{Comment}{$\triangleright$\ }{}
\newcommand{\MethodCall}[1]{\text{\textsc #1}}
\newcommand{\Let}{\hbox{\text{\textbf{let}}}\xspace}
\newcommand{\EventSet}{\hbox{$\mathcal{E}$}\xspace}
\newcommand{\Vertices}{\hbox{$V$}\xspace}
\newcommand{\Edges}{\hbox{$E$}\xspace}
\newcommand{\Node}{\hbox{$n$}\xspace}
\newcommand{\PrevStateTable}{\hbox{$\mathcal{H}$}\xspace}
\newcommand{\EventID}{\hbox{\emph{eid}}\xspace}
\newcommand{\Order}{\hbox{\emph{ord}}\xspace}
\newcommand{\GetPrev}{\hbox{\emph{past}}\xspace}
\newcommand{\GetFirst}{\hbox{\emph{first}}\xspace}
\newcommand{\GetLast}{\hbox{\emph{last}}\xspace}
\newcommand{\GetStates}{\hbox{\emph{states}}\xspace}
\newcommand{\GetEvent}{\hbox{\emph{event}}\xspace}
\newcommand{\HappensBefore}{\hbox{$\:\rightarrow_{s}\:$}\xspace}
\newcommand{\Reachable}{\hbox{$\:\rightarrow_{r}\:$}\xspace}
\newcommand{\Next}{\hbox{\emph{next}}\xspace}
\newcommand{\Dest}{\hbox{\emph{dst}}\xspace}
\newcommand{\Src}{\hbox{\emph{src}}\xspace}
\newcommand{\Conf}{\hbox{\textit{conflicts}}\xspace}
\newcommand{\State}[1]{\hbox{$s_{#1}$}\xspace}
\newcommand{\CurrentState}{\hbox{$s$}\xspace}

\newcommand{\States}{\hbox{$\textit{States}$}\xspace}
\newcommand{\Event}[1]{\hbox{$e_{#1}$}\xspace}
\newcommand{\tildeEvent}[1]{\hbox{$\tilde{e}_{#1}$}\xspace}
\newcommand{\Process}{\hbox{$p$}\xspace}
\newcommand{\LastCycleEvent}{\hbox{$\CycleEvent_{s}$}\xspace}
\newcommand{\CycleEvent}{\hbox{$\Event_{c}$}\xspace}
\newcommand{\Events}[1]{\hbox{$\mathcal{E}_{#1}$}\xspace}
\newcommand{\AllEvents}{\hbox{$E$}\xspace}
\newcommand{\BacktrackEvent}{\hbox{$b$}\xspace}
\newcommand{\Trace}[1]{\hbox{$\mathcal{S}_{#1}$}\xspace}
\newcommand{\Transitions}[1]{\hbox{$\mathcal{T}_{#1}$}\xspace}
\newcommand{\TransSystem}{\hbox{$A_G$}\xspace}
\newcommand{\Backtrack}{\hbox{\textit{backtrack}}\xspace}
\newcommand{\Enabled}{\hbox{\textit{enabled}}\xspace}
\newcommand{\Done}{\hbox{\textit{done}}\xspace}
\newcommand{\Inst}[1]{\textit{I}_{#1}}

% Methods
\newcommand{\MethodExplore}{\textsc{Explore}\xspace}
\newcommand{\MethodExploreAll}{\textsc{ExploreAll}\xspace}
\newcommand{\MethodIsFullCycle}{\textsc{IsFullCycle}\xspace}
\newcommand{\MethodUpdateStateSummary}{\textsc{UpdateStateSummary}\xspace}
\newcommand{\MethodUpdateBacktrackSet}{\textsc{UpdateBacktrackSet}\xspace}
\newcommand{\MethodUpdateBacktrackSetDFS}{\textsc{UpdateBacktrackSetDFS}\xspace}
\newcommand{\MethodUpdateBacktrackSetsFromGraph}{\textsc{UpdateBacktrackSetsFromGraph}\xspace}

\newcommand{\code}[1]{\text{\tt #1}}
\newcommand{\todo}[1]{\textbf{\color{red} #1}}
\newcommand{\TODO}[0]{\textbf{TODO}\xspace}

\newcommand{\mysection}[1]{\vspace{-.27em}\section{#1}\vspace{-.27em}}
\newcommand{\mysubsection}[1]{\vspace{-.27em}\subsection{#1}\vspace{-.27em}}
\newcommand{\mysubsubsection}[1]{\subsubsection{#1}}
\newcommand{\myparagraph}[1]{\noindent\textit{\textbf{{#1}.}}}
\newcommand{\mysubparagraph}[1]{\subparagraph{#1}}
\newcommand{\mycomment}[1]{}
\newcommand{\tuple}[1]{\ensuremath \langle #1 \rangle}
\newcommand{\Naive}[0]{Na\"{i}ve\xspace}
\newcommand{\naively}[0]{na\"{i}vely\xspace}
\newcommand{\Naively}[0]{Na\"{i}vely\xspace}
\newcommand{\naive}[0]{na\"{i}ve\xspace}
\newcommand{\interacts}[0]{\emph{interacts-with}\xspace}
\newcommand{\conflict}[0]{\emph{conflict}\xspace}
\newcommand{\Interacts}[0]{\emph{Interacts-with}\xspace}
\newcommand{\Conflict}[0]{\emph{Conflict}\xspace}
\newcommand{\Devicestate}[0]{\emph{Device}\xspace}
\newcommand{\Physical}[0]{\emph{Physical-Medium}\xspace}
\newcommand{\Globalstate}[0]{\emph{Global-Variable}\xspace}
\newcommand{\devicestate}[0]{\emph{device}\xspace}
\newcommand{\globalstate}[0]{\emph{global-variable}\xspace}
\newcommand{\physical}[0]{\emph{physical-medium}\xspace}
\newcommand{\conflictinreader}[0]{\emph{conflict in reader}\xspace}
\newcommand{\readop}[0]{\emph{read-from}\xspace}
\newcommand{\writeop}[0]{\emph{write-to}\xspace}
\newcommand{\readsop}[0]{\emph{reads-from}\xspace}
\newcommand{\writesop}[0]{\emph{writes-to}\xspace}
\newcommand{\doorlock}[0]{door-lock\xspace}
\newcommand{\doorunlock}[0]{door-unlock\xspace}
\newcommand{\doorlocks}[0]{door-locks\xspace}
\newcommand{\doorunlocks}[0]{door-unlocks\xspace}
\newcommand{\locationmode}[0]{\code{location.mode}\xspace}
\newcommand{\locationspacemode}[0]{\code{location.} \code{mode}\xspace}

% App names
\newcommand{\tool}[0]{\textsc{IoTCheck}\xspace}
\newcommand{\initialstateeventstreamer}[0]{\code{Initial}-\code{State}-\code{Event}-\code{Streamer}\xspace}
\newcommand{\bigturnon}[0]{\code{Big}-\code{Turn}-\code{ON}\xspace}
\newcommand{\firecoalarm}[0]{\code{FireCO2Alarm}\xspace}
\newcommand{\devicetamperalarm}[0]{\code{Device}-\code{Tamper}-\code{Alarm}\xspace}
\newcommand{\windowordooropen}[0]{\code{WindowOrDoorOpen}\xspace}
\newcommand{\lockitwhenileave}[0]{\code{Lock}-\code{It}-\code{When}-\code{I}-\code{Leave}\xspace}
\newcommand{\lockitataspecifictime}[0]{\code{Lock}-\code{It}-\code{at}-\code{a}-\code{Specific}-\code{Time}\xspace}
\newcommand{\autolockdoor}[0]{\code{Auto}-\code{Lock}-\code{Door}\xspace}
\newcommand{\thermostats}[0]{\code{Thermostats}\xspace}
\newcommand{\thermostatautooff}[0]{\code{Thermostat}-\code{Auto}-\code{Off}\xspace}
\newcommand{\neato}[0]{\code{Neato}-\code{(Connect)}\xspace}
\newcommand{\forgivingsecurity}[0]{\code{Forgiving}-\code{Security}\xspace}
\newcommand{\turnonatsunset}[0]{\code{Turn}-\code{On}-\code{at}-\code{Sunset}\xspace}
\newcommand{\lightupthenight}[0]{\code{Light}-\code{Up}-\code{the}-\code{Night}\xspace}
\newcommand{\bosesound}[0]{\code{Bose}-\code{SoundTouch}-\code{Control}\xspace}
\newcommand{\sprayercontroller}[0]{\code{Sprayer}-\code{Controller}-\code{2}\xspace}
\newcommand{\closethevalve}[0]{\code{Close}-\code{The}-\code{Valve}\xspace}
\newcommand{\influxdblogger}[0]{\code{InfluxDB}-\code{Logger}\xspace}
\newcommand{\buttoncontroller}[0]{\code{Button}-\code{Controller}\xspace}
\newcommand{\smartsecurity}[0]{\code{Smart}-\code{Security}\xspace}
\newcommand{\vacationlightingdirector}[0]{\code{Vacation}-\code{Lighting}-\code{Director}\xspace}
\newcommand{\lightingdirector}[0]{\code{Lighting}-\code{Director}\xspace}
\newcommand{\bigturnoff}[0]{\code{Big}-\code{Turn}-\code{OFF}\xspace}
\newcommand{\monitoronsense}[0]{\code{Monitor}-\code{on}-\code{Sense}\xspace}
\newcommand{\huemoodlighting}[0]{\code{Hue}-\code{Mood}-\code{Lighting}\xspace}
\newcommand{\doorstatetocolorlight}[0]{\code{Door}-\code{State}-\code{to}-\code{Color}-\code{Light}\xspace}
\newcommand{\greetingsearthling}[0]{\code{Greetings}-\code{Earthling}\xspace}
\newcommand{\hellohomephrasedirector}[0]{\code{Hello,}-\code{Home}-\code{Phrase}-\code{Director}\xspace}
\newcommand{\bonvoyage}[0]{\code{Bon}-\code{Voyage}\xspace}
\newcommand{\goodnight}[0]{\code{Good}-\code{Night}\xspace}
\newcommand{\wholehousefan}[0]{\code{Whole}-\code{House}-\code{Fan}\xspace}
\newcommand{\autohumidityvent}[0]{\code{Auto}-\code{Humidity}-\code{Vent}\xspace}
\newcommand{\keepmecozy}[0]{\code{Keep}-\code{Me}-\code{Cozy}\xspace}

\newcommand{\harry}[1] {{\textcolor{blue}{Harry: {#1}}}}
\newcommand{\rahmadi}[1]{{#1}}
\newcommand{\weiyu}[1]{{\color{purple}#1}}
\newcommand{\camready}[1]{#1}

\newcounter{isvmcaiversion}
% Set the following value to:
% - 1 for the VMCAI paper version.
% - 0 for the ArXiv tech report version.
\setcounter{isvmcaiversion}{0}

\ifthenelse{\value{isvmcaiversion}>0}
{
% VMCAI Paper
    \SetWatermarkText{\hspace*{6in}\raisebox{8.8in}{\includegraphics[scale=0.15]{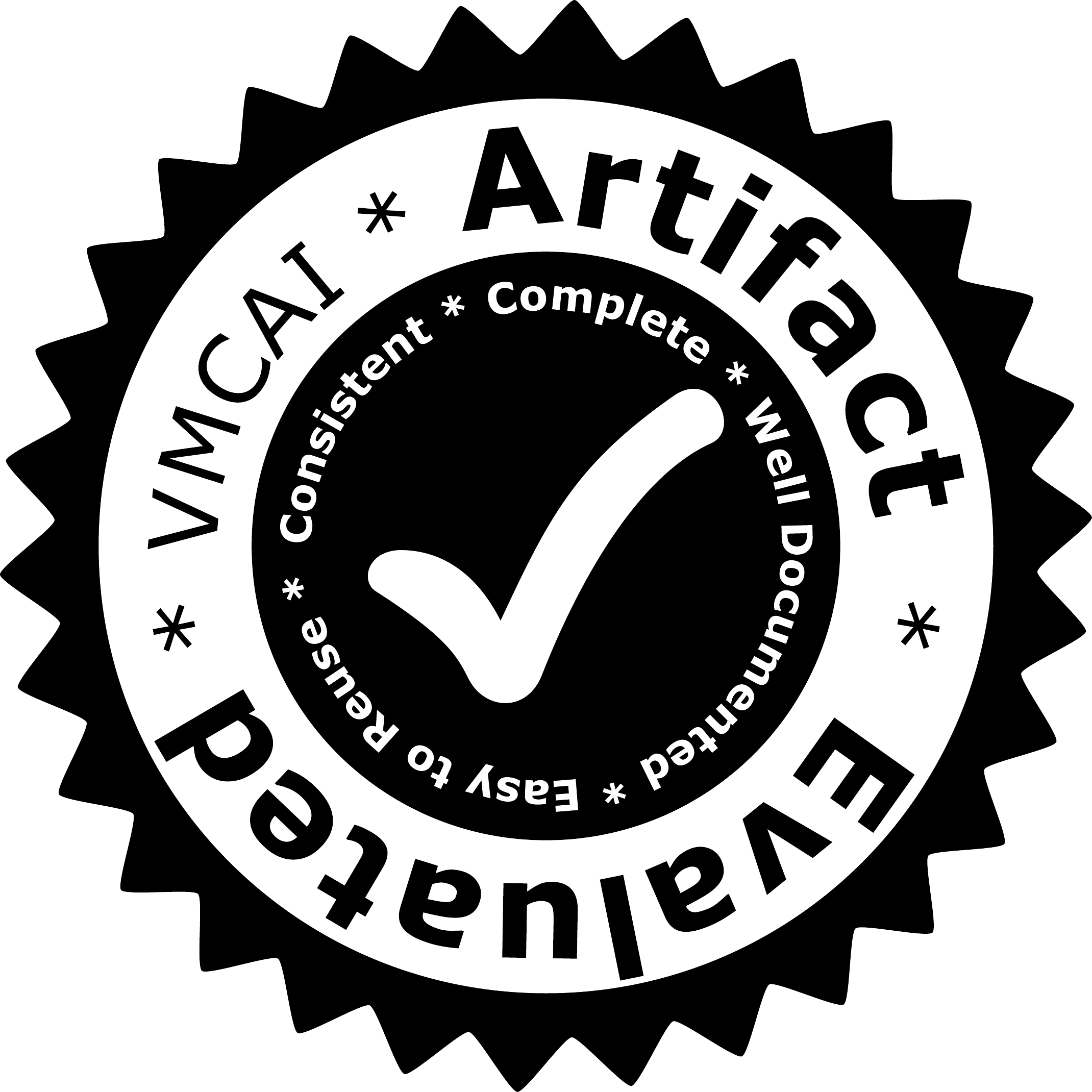}}}
    \SetWatermarkAngle{0}
    \pagenumbering{gobble} % Remove page numbers
    \newcommand{\appdx}[1]{\text{\color{blue} #1 in~\todo{\cite{ArXiv}}}}
    \newcommand{\invmcaiversion}[1]{#1}
    \newcommand{\inarxivversion}[1]{}
}
% else
{
% ArXiv Tech Report
    \SetWatermarkText{}
    \SetWatermarkAngle{0}
    \newcommand{\appdx}[1]{#1}
    \newcommand{\invmcaiversion}[1]{}
    \newcommand{\inarxivversion}[1]{#1}
}

\maketitle

%% Adding page number for now
\thispagestyle{plain}
\pagestyle{plain}

\begin{abstract}
Event-driven architectures are broadly used for systems that must
respond to events in the real world.  Event-driven applications are
prone to concurrency bugs that involve subtle errors in reasoning
about the ordering of events.  Unfortunately, there are
several challenges in using existing model-checking techniques on
these systems.  Event-driven applications often loop indefinitely and
thus pose a challenge for stateless model checking techniques.  On the
other hand, deploying purely stateful model checking can explore large
sets of equivalent executions.

In this work, we explore a new technique that combines dynamic partial
order reduction with stateful model checking to support
non-terminating applications.  Our work is (1) the first dynamic partial order reduction algorithm for stateful model checking that is sound for non-terminating applications and (2) the first dynamic partial reduction algorithm for stateful model checking of event-driven applications.
We experimented with the IoTCheck dataset---a study of interactions in smart home app pairs.
This dataset consists of app pairs originated from 198 real-world smart home apps.
Overall, our DPOR algorithm successfully reduced the search space for the app pairs, enabling 69 pairs of apps that did not finish without DPOR to finish and providing a 7$\times$ average speedup.
\end{abstract}

\section{Introduction\label{sec:intro}}

Event-driven architectures are broadly used to build systems that react to events in the real world.  They include smart home systems, GUIs, mobile applications, and servers.  
For example, in the context of smart home systems, event-driven systems include Samsung SmartThings~\cite{smartthings}, Android Things~\cite{androidthings}, Openhab~\cite{openhab}, and If This Then That (IFTTT)~\cite{ifttt}.

Event-driven architectures can have analogs of the concurrency bugs that are known to be problematic in multithreaded programming.  Subtle programming errors involving the ordering of events can easily cause event-driven programs to fail.  These failures can be challenging to find during testing as exposing these failures may require a specific set of events to occur in a specific order.
Model-checking tools can be helpful for finding subtle concurrency bugs or understanding complex interactions between different applications~\cite{trimanafse20}.  In recent years, significant work has been expended on developing model checkers for multithreaded concurrency~\cite{odpor,statelessmcr,jensen2015stateless,dporlocks,pldipsotso,inspect1,inspect2,inspect4}, but event-driven systems have received much less attention~\cite{jensen2015stateless,maiya2016por}.

Event-driven systems pose several challenges for existing \emph{stateless} and \emph{stateful} model-checking tools.  Stateless model checking of concurrent applications explores all execution schedules without checking whether these schedules visit the same program states.  Stateless model checking often uses dynamic partial order reduction (DPOR) to eliminate equivalent schedules. \mycomment{,\ie schedules that rearrange independent events}
While there has been much work on DPOR for stateless model checking of multithreaded programs~\cite{flanagan2005dynamic,odpor,dporlocks,pldipsotso,statelessmcr}, stateless model checking requires that the program under test terminates for fair schedules.  Event-driven systems are often intended to run continuously and may not terminate.  To handle non-termination, stateless model checkers require hacks such as bounding the length of executions to verify event-driven systems.

Stateful model checking keeps track of an application's states and avoids revisiting the same application states.  It is less common for stateful model checkers to use dynamic partial order reduction to eliminate equivalent executions.
Researchers have done much work on stateful model checking~\cite{jpf,spin,musuvathi2002cmc,gueta2007cartesian}.  While stateful model checking can handle non-terminating programs, they  miss an opportunity to efficiently reason about conflicting transitions to scale to large programs.  In particular, typical event-driven programs such as smart home applications have several event handlers that are completely independent of each other.  Stateful model checking enumerates different orderings of these event handlers, overlooking the fact that these handlers are independent of each other and hence the orderings are equivalent.

Stateful model checking and dynamic partial order reduction discover
different types of redundancy, and therefore it is beneficial to combine them to further improve model-checking scalability and efficiency.  For example, we have observed that some smart home systems have
several independent event handlers \mycomment{(\eg monitoring applications or
just independent systems)}in our experiments, and stateful model checkers can waste
an enormous amount of time exploring different orderings of these
independent transitions.  DPOR can substantially reduce the number of
states and transitions that must be explored.  Although work has been done to 
combine DPOR algorithms with stateful model
checking~\cite{yang2008efficient,yi2006stateful} in the context of
multithreaded programs, this line of work requires that the application has
an \emph{acyclic state space}, \ie it terminates under all schedules.  In particular, the
approach of Yang \etal\cite{yang2008efficient} is designed explicitly for
programs with acyclic state space and thus
cannot check programs that do not terminate.  
Yi \etal\cite{yi2006stateful} presents a DPOR algorithm for stateful model
checking, which is, however, incorrect for cyclic state spaces. For
instance, their algorithm fails to produce the asserting execution in
the example we will discuss shortly in
Figure~\ref{fig:example-non-terminating}.  As a result, prior DPOR
techniques all fall short for checking event-driven programs such as
smart home apps, that, in general, do not terminate.

\myparagraph{Our Contributions}
In this work, we present a stateful model checking technique for
event-driven programs that may not terminate. Such programs have cyclic state spaces, and existing algorithms can prematurely terminate an execution and thus fail to set
the necessary backtracking points to fully explore a program's
state space. Our \textbf{first} technical contribution is the \emph{formulation of a
sufficient condition to complete an execution of the application
that ensures that our algorithm fully explores the application's
state space}.

In addition to the early termination issue, for
programs with cyclic state spaces, a model checker can discover
multiple paths to a state \State{} before it explores the entire state
space that is reachable from state \State{}.  In this case, the
backtracking algorithms used by traditional DPOR techniques including
Yang \etal\cite{yang2008efficient} can fail to set the necessary
backtracking points.  Our \textbf{second} technical contribution is \emph{a graph-traversal-based algorithm to appropriately set backtracking points on all paths that can reach the current state}.

Prior work on stateful DPOR only considers the multithreaded case
and assumes algorithms know the effects of the next transitions
of all threads before setting backtracking points.   For multithreaded programs, this assumption is \emph{not} a serious limitation as transitions model low-level memory operations (\ie reads, writes, and RMW operations), and each transition involves a \emph{single} memory operation.
However, in the context of event-driven programs,  events can involve
many memory operations that access multiple memory locations,
and knowing the effects of a transition requires actually
executing the event.  While it is conceptually possible to execute events and then rollback to discover their effects, this approach is likely to incur large overheads as model checkers need to know the effects of enabled events at each program state. As our \textbf{third} contribution, \emph{our algorithm avoids this extra rollback overhead by waiting until an event is actually
executed to set backtracking points and incorporates a modified backtracking
algorithm to appropriately handle events}.

%\mysubsection{Contributions\label{subsec:contribution}}
%This paper makes the following contributions:
%\begin{itemize}
%\item \textbf{Stateful DPOR model checking of event-driven programs:}
%%We present the first algorithm for combining stateful model checking with DPOR for event-driven programs.

%\item \textbf{A DPOR algorithm that can handle non-terminating programs:}
%The algorithm presented in our work extends the state-of-the-art
%work in combining DPOR with stateful model checking to support computations with cyclic state %spaces.

%\item \textbf{Event-driven DPOR algorithm:}
%Our paper incorporates and adapts the existing DPOR technique 
%to further reduce the state space exploration in the context of
%event-driven programs.

%\item \textbf{Implementation and evaluation:}
\rahmadi{We have implemented the proposed algorithm in the Java Pathfinder model checker~\cite{jpf} and evaluated it on hundreds of real-world smart home apps. We have made our DPOR implementation publicly available~\cite{iotcheck-dpor-software}.}
%\end{itemize}

\myparagraph{Paper Structure}
The remainder of this paper is structured as follows: 
Section~\ref{sec:concurrency} presents the event-driven concurrency model that we use in this work.
Section~\ref{sec:definition} presents the definitions we use to describe our stateful DPOR algorithm.
Section~\ref{sec:basicideas} presents problems when using the classic DPOR algorithm to model check event-driven programs and the basic ideas behind how our algorithm solves these problems.
Section~\ref{sec:algorithm} presents our stateful DPOR algorithm for event-driven programs.
%Section~\ref{sec:proof} proves the correctness of our algorithm.
Section~\ref{sec:impeval} presents the evaluation of our algorithm 
implementation on hundreds of smart home apps.    
Section~\ref{sec:related} presents the related work; we conclude in
Section~\ref{sec:conclusion}. 
%\newline

\mysection{Event-Driven Concurrency Model\label{sec:concurrency}}
%\vspace{-2em}
\camready{
In this section, we first present the concurrency model of our event-driven system and then discuss the key elements of this system formulated as an event-driven concurrency model.} Our event-driven system is inspired by---and distilled from---smart home IoT devices and applications deployed widely in the real world. 
Modern smart home platforms support developers writing apps that 
implement useful functionality on smart devices.  
Significant efforts have been made to create integration platforms such as
Android Things from Google~\cite{androidthings}, SmartThings from Samsung~\cite{smartthings}, and the open-source openHAB platform~\cite{openhab}. All of these platforms allow users to create \emph{smart home apps} that integrate multiple devices and perform complex routines, such as implementing a home security system. 

The presence of multiple apps that can control the same device creates undesirable interactions~\cite{trimanafse20}. For example, a homeowner may install the \firecoalarm~\cite{fireco2alarm-app} app, which upon the detection of smoke, sounds alarms and unlocks the door. The same homeowner may also install the \lockitwhenileave~\cite{lockitwhenileave} app to lock the door automatically when the
homeowner leaves the house.  However, these apps can interact in surprising ways when installed together. For instance, if smoke is detected, \firecoalarm will unlock the door. If someone leaves home\mycomment{with the presence tag, this will make the presence sensor change its state from
\code{"present"} to \code{"not present"}}, the \lockitwhenileave app will lock the door.
This defeats the intended purpose of the \firecoalarm app. Due to the increasing popularity of IoT devices, understanding and finding such conflicting interactions has become a hot research topic~\cite{iagraphs,iagraphstechreport,vicaire2010physicalnet,vicaire2012bundle,wood2008context} in the past few years. Among the many techniques developed, model checking is a popular one~\cite{yagita-conflict-detection,trimanafse20}. However, existing DPOR-based model checking algorithms do not support non-terminating event-handling logic (detailed in Section~\ref{sec:basicideas}), which strongly motivates the need of developing new algorithms that are both sound and efficient in handling real-world event-based (\eg IoT) programs. 
\vspace{-.5em}

\subsection{Event-Driven Concurrency Model}
\rahmadi{
We next present our event-driven concurrency model (see an example of event-driven 
\invmcaiversion{systems in \appdx{Appendix A}).}
\inarxivversion{systems in \appdx{Appendix~\ref{sec:example-event-driven-system}}).}
}
We assume that
the event-driven system has a finite set $\Events{}$ of different event types.  Each event type $\Event{} \in \Events{}$ has a
corresponding event handler that is executed when an instance of the
event occurs.  We assume that there is a potentially shared state and
that event handlers have arbitrary access to read and write from this
shared state.

An event handler can be an arbitrarily long finite sequence of
instructions and can include an arbitrary number of accesses to shared
state.  We assume event-handlers are executed atomically by the
event-driven runtime system.
Events can be enabled by both external sources (\eg events in the
physical world) or event handlers.  Events can also be disabled by
the execution of an event handler.
We assume that the runtime system maintains an unordered set of
enabled events to execute.  It contains an event dispatch
loop that selects an arbitrary enabled event to execute next.

This work is inspired by smart-home systems that are widely deployed in the real world. However, the proposed techniques are general enough to handle other types of event-driven systems,
such as web applications, as long as the systems follow the concurrency model stated above.
\vspace{-.5em}

\subsection{Background on Stateless DPOR}

Partial order reduction is based on the observation that traces of concurrent systems are equivalent if they only reorder independent operations.  These equivalence classes are called Mazurkiewicz traces~\cite{tracetheory}.  The classical DPOR algorithm~\cite{flanagan2005dynamic} dynamically computes persistent sets for multithreaded programs and is guaranteed to explore at least one interleaving in each equivalence class.

The key idea behind the DPOR algorithm is to compute the next pending memory operation for each thread, and at each point in the execution to compute the most recent conflict for each thread's next operation.  These conflicts are used to set backtracking points so that future executions will reverse the order of conflicting operations and explore an execution in a different equivalence class. Due to space constraints, we refer the interested readers to \cite{flanagan2005dynamic} for a detailed description of the original DPOR algorithm.   

\mysection{Preliminaries\label{sec:definition}}
%\vspace{-.5em}
We next introduce the notations and definitions we use throughout this paper.

\vspace{.3em}
\myparagraph{Transition System}
We consider a transition system that consists of a finite set $\Events{}$ of events.
Each event $\Event{} \in \Events{}$ executes a sequence of instructions that
change the \emph{global} state of the system.

\vspace{.3em}
\myparagraph{States}
Let $\States{}$ be the set of the states of the system, where
$\State{0} \in \States{}$ is the initial state.
A state \State{} captures the heap 
of a running program and the values of global variables.

\vspace{.3em}
\myparagraph{Transitions and Transition Sequences}
\camready{
Let $\Transitions{}\:$ be the set of all transitions for the system.  Each
transition $\Trans{} \in \Transitions{}\:$ is a partial function from $\States{}$ to $\States{}$.
The notation $\Trans{s,e}=\textit{next}\left( \State{}, \Event{} \right)$ returns the transition $\Trans{s,e}$ from executing event \Event{} on program state \State{}.  We assume that the transition system is deterministic, and thus the destination state $\Dest{}(\Trans{s,e})$ is unique for a given state $\State{}$ and event $\Event{}$. If the execution of transition $\Trans{}$ from $\State{}$ produces state $\State{}'$, then we write $\State{} \xrightarrow{\Trans{}} \State{}'$.}

We formalize the behavior of the
system as a transition system $\TransSystem =
( \States{}, \Delta, \State{0})$, where
$\Delta \subseteq \States{} \times \States{}$ is the transition
relation defined by
\begin{align*}(\State{}, \State{}') \in \Delta \:\: \text{iff} \:\: \exists \Trans{} \in \Transitions{}: \State{} \xrightarrow[]{\Trans{}} \State{}'\end{align*}
and $\State{0}$ is the initial state of the system.

A transition sequence \Trace{} of the transition system is a finite
sequence of transitions $\Trans{1},\Trans{2},...,\Trans{n}$.  These
transitions advance the state of the system from the initial
state \State{0} to further states $\State{1},...,\State{i}$ such that

\hskip9.0em
$\Trace{}=
\State{0}\xrightarrow[]{\Trans{1}}
\State{1}\xrightarrow[]{\Trans{2}}
...\:
\State{i-1}\xrightarrow[]{\Trans{n}}
\State{i}.
$

\vspace{.3em}
\myparagraph{Enabling and Disabling Events} Events can be enabled and disabled.  We make the same assumption as Jensen \etal\cite{jensen2015stateless} regarding the mechanism for enabling and disabling events.  Each event has a special memory location associated with it.  When an event is enabled or disabled, that memory location is written to.  Thus, the same conflict detection mechanism we used for memory operations will detect enabled/disabled conflicts between events.

\vspace{.3em}
\myparagraph{Notation}
We use the following notations in our presentation:
\begin{itemize}
\item$\GetEvent(\Trans{})$ returns the event that performs the transition $\Trans{}$.

\item$\GetFirst(\Trace{}, \State{})$ returns the first occurrence of state \State{} in \Trace{}, \eg if $\State{4}$ is first visited at step 2 then
$\GetFirst(\Trace{}, \State{4})$ returns 2.

\item$\GetLast(\Trace{})$ returns the last state \State{} in a transition sequence \Trace{}.

\item$\Trace{}.\Trans{}$ produces a new transition sequence by extending the transition sequence $\Trace{}$ with the transition $\Trans{}$.

\item$\GetStates(\Trace{})$ returns the set of states traversed by the transition sequence $\Trace{}$.

\item $\Enabled(\State{})$ denotes the set of enabled events at \State{}.

\item $\Backtrack(\State{})$ denotes the backtrack set of state \State{}.

\item $\Done(\State{})$ denotes the set of events that have already been executed at \State{}.

\item $\textit{accesses}(\Trans{})$ denotes the set of memory accesses performed by the transition $\Trans{}$.  An access consists of a memory operation, \ie a read or write, and a memory location.
%\item $\HappensBefore$ denotes \emph{happens-before} relation
%between two events, \eg
%$\Trans{1}\HappensBefore\Trans{2}$ indicates that there is a
%\emph{happens-before} relationship between $\Trans{1}$ and $\Trans{2}$.
\end{itemize}

%\vspace{.3em}
\myparagraph{State Transition Graph}
In our algorithm, we construct a state transition graph \RGraph{} that
is similar to the visible operation dependency graph presented
in~\cite{yang2008efficient}.  The state transition graph records all
of the states that our DPOR algorithm has explored and all of the
transitions it has taken.  In more detail, a state transition graph
$\RGraph{}=\langle V,E \rangle$ for a transition
system is a directed graph, where
every node $\Node\in\Vertices$ is a visited state, and every edge $e \in
E$ is a transition explored in some execution.  We
use \Reachable to denote that a transition is reachable from another transition
in \RGraph{}, \eg $\Trans{1}\Reachable\Trans{2}$ indicates that
$\Trans{2}$ is reachable from $\Trans{1}$ in \RGraph{}.

\vspace{.3em}
\myparagraph{Independence and Persistent Sets}
We define the independence relation over transitions as follows:
\begin{definition}[Independence]
Let $\Transitions{} \:$ be the set of transitions.  An independence relation
$I \subseteq \Transitions{} \: \times \Transitions{}$ is a irreflexive and symmetric relation,
such that for any transitions $(\Trans{1}, \Trans{2}) \in I$ and any state $\State{}$ in the state space of a transition system $A_G$,
the following conditions hold:
\begin{enumerate}
  \item if $\Trans{1} \in \Enabled(\State{})$ and $\State{} \xrightarrow{\Trans{1}} \State{}'$, then $\Trans{2} \in \Enabled(\State{})$ iff $\Trans{2} \in \Enabled(\State{}')$.
  \item if $\Trans{1}$ and $\Trans{2}$ are enabled in $\State{}$, then there is a unique state $\State{}'$ such that
    $\State{} \xrightarrow{\Trans{1} \Trans{2}} \State{}'$ and $\State{} \xrightarrow{\Trans{2} \Trans{1}} \State{}'$.
\end{enumerate}
\end{definition}
\camready{
If $(\Trans{1}, \Trans{2}) \in I$, then we say $\Trans{1}$ and $\Trans{2}$ are independent. 
We also say that two memory accesses to a shared location \textit{conflict} if at least one of them is a write. 
Since executing the same event from different states can have different effects
on the states, \ie resulting in different transitions, we also define the notion of \textit{read-write independence}
between events on top of the definition of independence relation over transitions.}

%\vspace{.3em}
\camready{
\begin{definition}[Read-Write Independence]
We say that two events $x$ and $y$ are read-write independent, if for every transition
sequences $\tau$ where events $x$ and $y$ are executed, the transitions $\Trans{x}$ and
$\Trans{y}$ corresponding to executing $x$ and $y$ are independent, and
$\Trans{x}$ and $\Trans{y}$ do not have conflicting memory accesses.
\end{definition}}

\begin{definition}[Persistent Set]
A set of events $X \subseteq \Events{}$ enabled in a state $s$ is persistent in $s$ if
for every transition sequence from $s$ 
\[ \State{} \xrightarrow{\Trans{1}} \State{1} \xrightarrow{\Trans{2}} ... \xrightarrow{\Trans{n-1}} \State{n-1} \xrightarrow{\Trans{n}} \State{n} \]
where $\GetEvent{(\Trans{i})} \notin X$ for all $1 \leq i \leq n$, then $\GetEvent{(\Trans{n})}$ is read-write independent with
all events in $X$.
\end{definition}

\rahmadi{
\inarxivversion{In \appdx{Appendix~\ref{sec:proof}}, we prove that exploring a persistent set of}
\invmcaiversion{In \appdx{Appendix B}, we prove that exploring a persistent set of}
events at each state is sufficient to ensure the exploration of at least one execution per Mazurkiewicz trace for a program with cyclic state spaces and finite reachable states.}
\mysection{Technique Overview\label{sec:basicideas}}
\vspace{-.5em}
This section overviews our ideas. These ideas are discussed in the context of four problems
that arise when existing DPOR algorithms are applied directly to event-driven programs. For each problem, we first explain the cause of the problem and then proceed to discuss our solution.

\begin{figure}[!t]
  \centering
  \begin{center}
  { \footnotesize
  \begin{tabular}{| c | c | c | c |}
    \hline
    \begin{lstlisting}[numbers=none]
    /* Initial 
       condition: */
    x = y = z = 0;
    \end{lstlisting} &
    \begin{lstlisting}[numbers=none]
    /* T1: */
    r1 = x++;
    assert(r1 == 0);
    \end{lstlisting} &
    \begin{lstlisting}[numbers=none]
    /* T2: */
    while(true)
        r2 = y;
    \end{lstlisting} &
    \begin{lstlisting}[numbers=none]
    /* T3: */
    r3 = z;
    r4 = x++;
    assert(r4 == 0);
    \end{lstlisting}\\
    \hline
  \end{tabular}
  \includegraphics[width=0.65\linewidth]{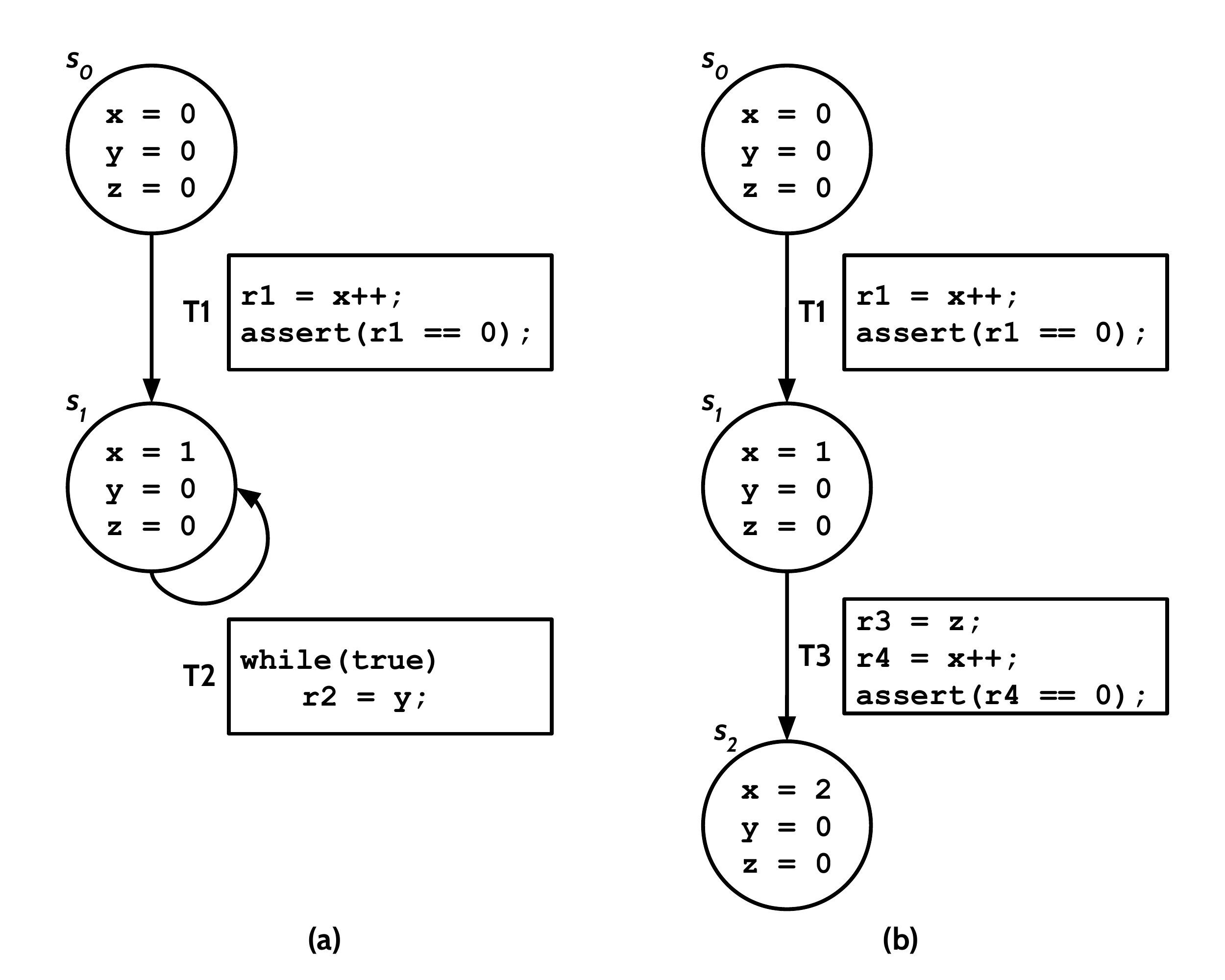}
  }
\end{center}
\vspace{-2em}
  \caption{Problem with existing stateful DPOR algorithms on a non-terminating multithreaded program. Execution (a) terminates at a state match without setting any backtracking points.  Thus, stateful DPOR would miss exploring Execution (b) which has an assertion failure.\label{fig:example-non-terminating}}
  \vspace{-1.5em}
\end{figure}

\vspace{-.5em}
\mysubsection{Problem 1: Premature Termination}
The first problem is that the naive application of existing stateless DPOR
algorithms to stateful model checking will prematurely terminate the execution
of programs with cyclic state spaces, causing a model checker to miss exploring portions of the state space.  This problem is known in the general POR literature~\cite{godefroidbook,valmari90,peled94} and various provisos (conditions) have been proposed to solve the problem.  
While the problem is known, all existing stateful DPOR algorithms produce incorrect results for programs with cyclic state spaces.  Prior work
by Yang \etal\cite{yang2008efficient} only handles programs with
acyclic state spaces.  Work by Yi \etal\cite{yi2006stateful} claims to
handle cyclic state spaces, but overlooks the need for a proviso for when it is
safe to stop an execution due to a state match and thus can produce
incorrect results when model checking programs with cyclic state
spaces.

Figure~\ref{fig:example-non-terminating} presents a simple multithreaded program that illustrates the problem of using a naive stateful adaptation of the DPOR algorithm to check programs with cyclic state spaces. Let us suppose that a stateful DPOR algorithm explores the state space from
$\State{0}$, and it selects thread \code{T$_1$} to take a step:
the state is advanced to state $\State{1}$.  However, when it
selects \code{T$_2$} to take the next step, it will revisit the same
state and stop the current execution (see
Figure~\ref{fig:example-non-terminating}-a). Since it did not set any
backtracking points, the algorithm prematurely finishes its
exploration at this point.  It misses the execution where both
threads~\code{T$_1$} and \code{T$_3$} take steps, leading to an assertion
failure. Figure~\ref{fig:example-non-terminating}-b shows this missing
execution.  The underlying issue with halting an execution when it
matches a state from the current execution is that the execution may
not have explored a sufficient set of events to create the necessary backtracking
points.  In our context, event-driven applications
are non-terminating.
Similar to our multithreaded example, executions in event-driven
applications may cause the algorithm to revisit a state and
prematurely stop the exploration.

\myparagraph{Our Idea}
Since the applications we are interested in typically have cyclic
state spaces, we address this challenge by changing our termination
criteria for an execution to require that an execution either (1)
matches a state from a previous execution or (2) matches a previously
explored state from the current execution and has explored every
enabled event in the cycle at least once since the first exploration of that state. The second criterion would prevent the DPOR algorithm from terminating prematurely after the exploration in Figure~\ref{fig:example-non-terminating}-a.

\begin{figure}[t!]
    \centering
	\includegraphics[width=0.8\linewidth]{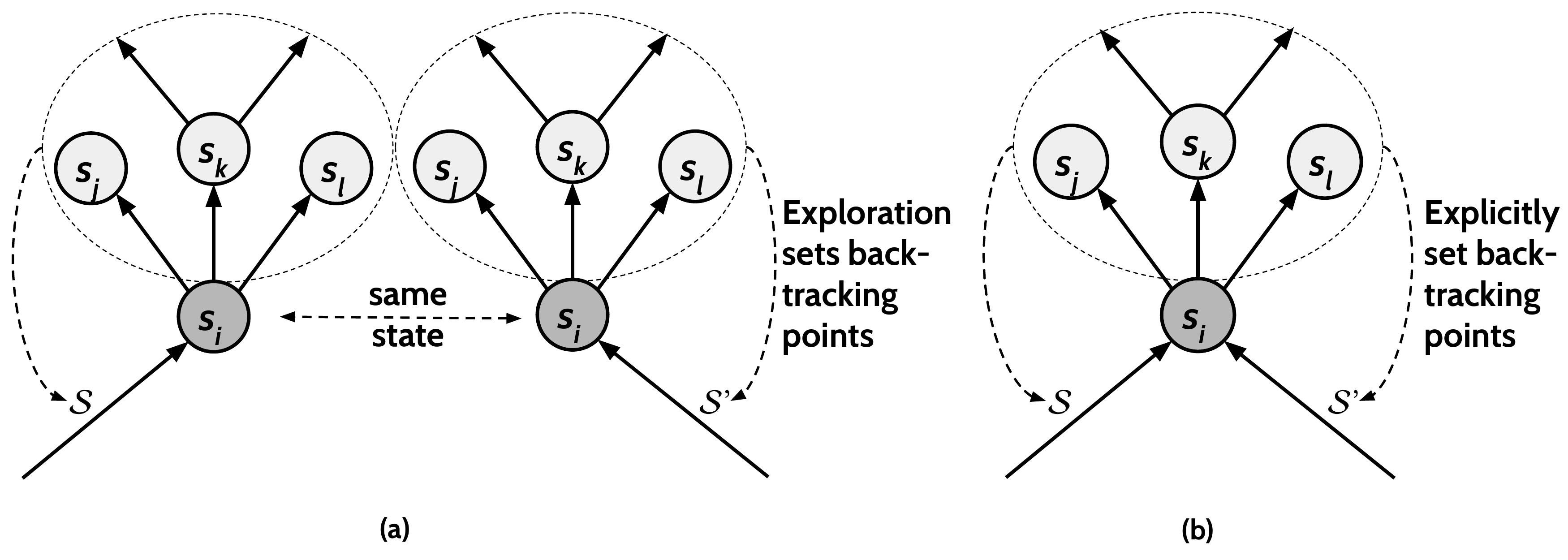}
	\vspace{-1.5em}
    \caption{(a) Stateless model checking explores $\State{i}$, $\State{j}$, $\State{k}$, and $\State{l}$ twice and thus sets backtracking points for both \Trace{} and \Trace{}'. (b) Stateful model checking matches state $\State{i}$ and skips the second exploration and thus we must explicitly set backtracking points.
      \label{fig:matching-explored-states}}
      \vspace{-1.5em}
\end{figure}

\vspace{-.5em}
\mysubsection{Problem 2: State Matching for Previously Explored States\label{sec:state-matching-for-previously-explored-states}}

Typically stateful model checkers can simply terminate an execution
when a previously discovered state is reached.  As mentioned in~\cite{yang2008efficient}, this handling is unsound in the presence
of dynamic partial order reduction.
Figure~\ref{fig:matching-explored-states} illustrates the issue:
Figure~\ref{fig:matching-explored-states}-a and b show the behavior of a classical
stateless DPOR algorithm as well as the
situation in a stateful DPOR algorithm, respectively. We assume that  \Trace{} was the first transition sequence to reach \State{i}
and \Trace{}' was the second such transition sequence.  The issue in Figure~\ref{fig:matching-explored-states}-b is
that after the state match for \State{i} in \Trace{}', the algorithm may \emph{inappropriately} skip setting backtracking points for the transition sequence \Trace{}', preventing the model checker from completely exploring the state space.

\myparagraph{Our Idea} Similar to the approach of Yang \etal\cite{yang2008efficient}, we propose to use a
graph to store the set of previously explored transitions that may set
backtracking points in the current transition sequence, so that the algorithm can set those backtracking points without reexploring the same state space. 

%\vspace{-.5em}
\mysubsection{Problem 3: State Matching Incompletely Explored States\label{sec:state-matching-incompletely-explored-states}}

\begin{figure}[!htb]
\vspace{-2em}
  \centering
  \begin{center}
  { \footnotesize
  \begin{tabular}{| l | l | l | l | l |}
    \hline
    \begin{lstlisting}[numbers=none]
    /* Initial
       condition: 
     */
    x = y = z = 0;
    \end{lstlisting} &
    \begin{lstlisting}[numbers=none]
    // e1:
    x = 1;
    z = 0;
    \end{lstlisting} &
    \begin{lstlisting}[numbers=none]
    // e2:
    r1 = z;
    y = 1;
    \end{lstlisting} &
    \begin{lstlisting}[numbers=none]
    // e3:
    y = 0;
    \end{lstlisting} &
    \begin{lstlisting}[numbers=none]
    // e4:
    if (x==0) {
        if (y==1)
            assert(y==0);
    } else {
        x = 0;
    }
    \end{lstlisting}\\
    \hline
  \end{tabular}
  %\vspace{1em}
  \includegraphics[width=1\linewidth]{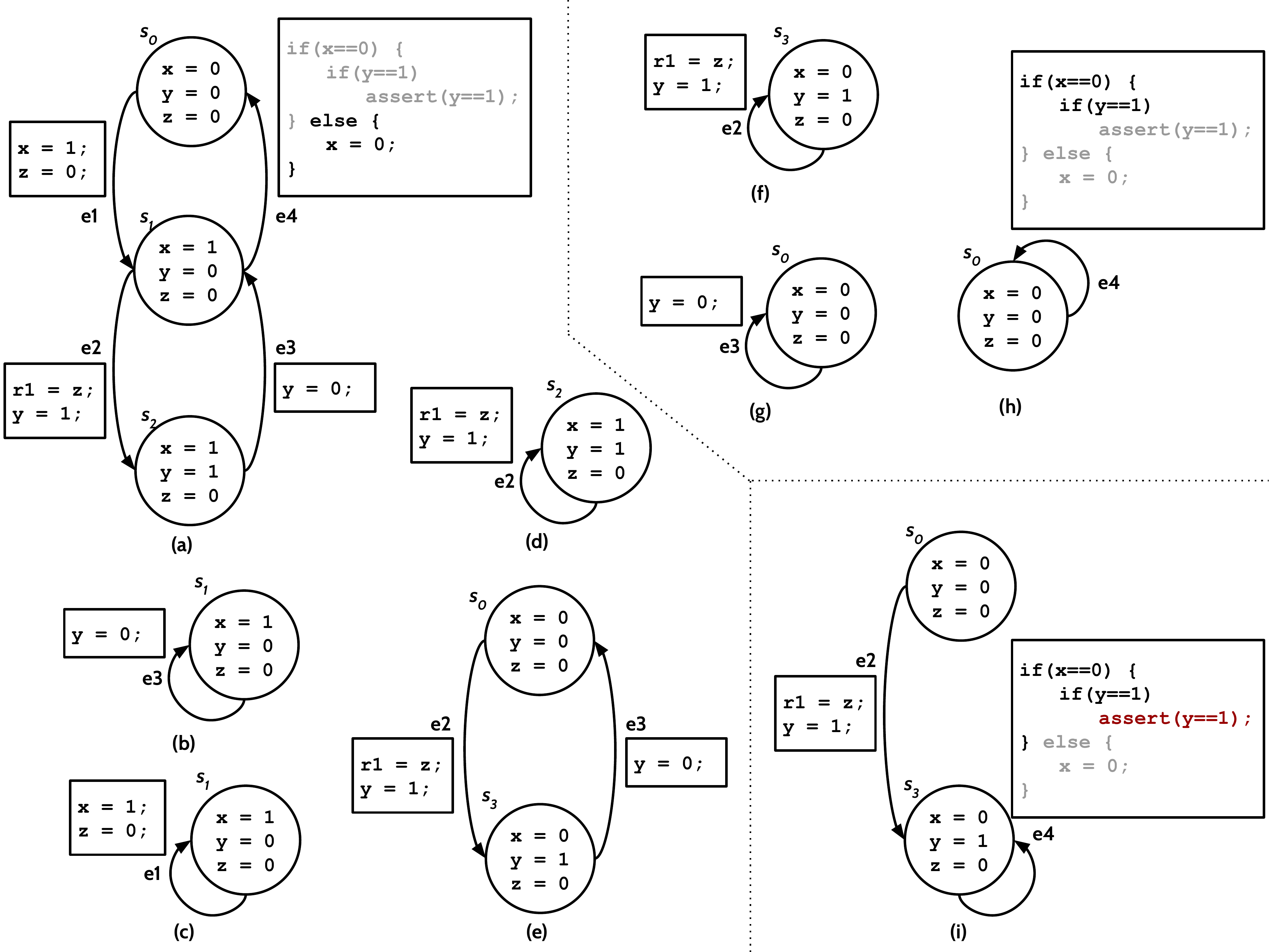}
    % Original example
    %/* E3: */
    %r2 = z;
    %y = 0;
  }
\end{center}
\vspace{-2em}
\caption{Example of a event-driven program that misses an execution.  We assume that \code{e1}, \code{e2}, \code{e3}, and \code{e4} are all initially enabled.\label{fig:example-missing-execution}}
\vspace{-1.5em}
\end{figure}

Figure~\ref{fig:example-missing-execution} illustrates another problem with cyclic state spaces---even if our new
termination condition and the algorithm for setting backtrack points for a
state match are applied to the stateful DPOR algorithm, it could still fail to
explore all executions.

With our new termination criteria, the stateful DPOR algorithm will first
explore the execution shown in
Figure~\ref{fig:example-missing-execution}-a. It starts from
$\State{0}$ and executes the events \code{e$_1$}, \code{e$_2$},
and \code{e$_3$}. While executing the three events, it puts event \code{e$_2$} in the
backtrack set of $\State{0}$ and event \code{e$_3$} in the backtrack set of
$\State{1}$ as it finds a conflict between the events \code{e$_1$}
and \code{e$_2$}, and the events \code{e$_2$} and \code{e$_3$}. Then, the algorithm
revisits $\State{1}$. At this point it updates the backtrack sets using the transitions that are reachable from state \State{1}: it puts event \code{e$_2$} in the backtrack set of
state \State{2} because of a conflict between \code{e$_2$}
and \code{e$_3$}.

However, with the new termination criteria, it does not stop its
exploration. It continues to execute event \code{e$_4$}, finds a
conflict between \code{e$_1$} and \code{e$_4$}, and
puts event \code{e4} into the backtrack set of $\State{0}$. The algorithm
now revisits state $\State{0}$ and updates the backtrack sets using the transitions reachable from state $\State{0}$: it puts event \code{e$_1$} in the backtrack set of $\State{1}$ because
of the conflict between \code{e$_1$} and \code{e$_4$}.
Figures~\ref{fig:example-missing-execution}-b, c, and d show the
executions explored by the stateful DPOR algorithm from the events \code{e$_1$} and \code{e$_3$} in the
backtrack set of $\State{1}$, and event \code{e$_2$} in the backtrack
set of $\State{2}$, respectively.

Next, the algorithm explores the execution from event \code{e$_2$} in the
backtrack set of $\State{0}$ shown in
Figure~\ref{fig:example-missing-execution}-e. The algorithm finds a
conflict between the events \code{e$_2$} and \code{e$_3$}, and it
puts event \code{e$_2$} in the backtrack set of $\State{3}$ and event \code{e$_3$} in
the backtrack set of $\State{0}$ whose executions are shown in
Figures~\ref{fig:example-missing-execution}-f and g, respectively.  Finally, the algorithm explores the execution from event \code{e$_4$} in the backtrack set of $\State{0}$ shown in Figure~\ref{fig:example-missing-execution}-h.
Then the algorithm stops, failing to explore the
asserting execution shown in
Figure~\ref{fig:example-missing-execution}-i.

The key issue in the above example is that the stateful DPOR algorithm by
Yang \etal\cite{yang2008efficient} does not consider all possible
transition sequences that can reach the current state but merely
considers the current transition sequence when setting backtracking points.  It thus does not add event \code{e$_4$} from the execution in Figure~\ref{fig:example-missing-execution}-h to the backtrack set of state \State{3}.

\myparagraph{Our Idea} Figure~\ref{fig:handling-loops} shows the core issue behind the problem.
When the algorithm sets backtracking points after
executing the transition \Trans{k}, the algorithm must consider both
the transition sequence that includes \Trans{h} and the transition
sequence that includes \Trans{i}.  The classical backtracking algorithm
would only consider the current transition sequence when setting
backtracking points.

\begin{figure}[t]
    \vspace{-.5em}
    \centering
	\includegraphics[width=0.4\linewidth]{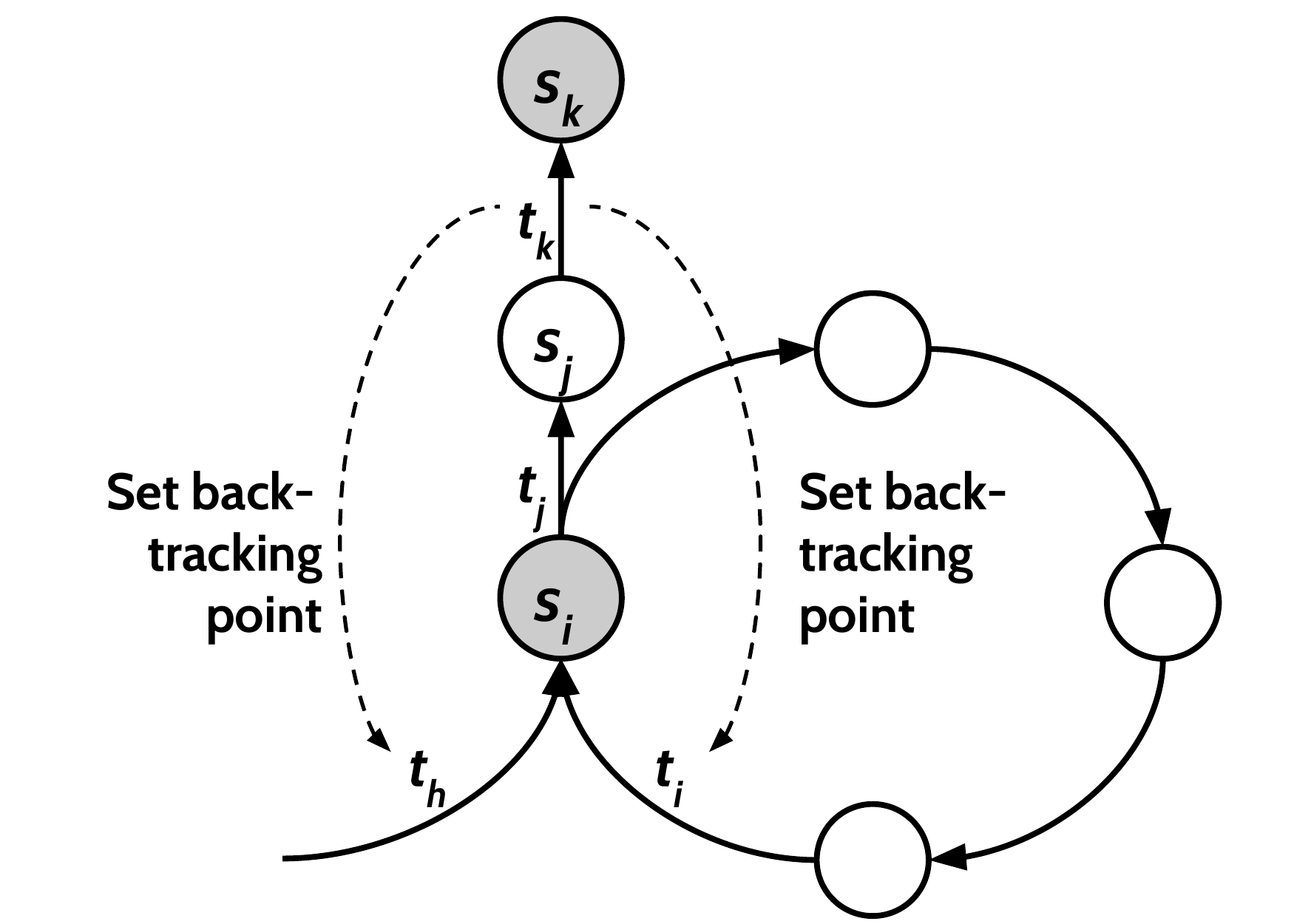}
	\vspace{-1em}
    \caption{Stateful model checking needs to handle loops caused by cyclic state spaces.
      \label{fig:handling-loops}}
      \vspace{-1.5em}
\end{figure}

We propose a new algorithm that uses a backwards depth first search on the state
transition graph combined with summaries to set backtracking points on previously discovered
paths to the currently executing transition.  Yi \etal\cite{yi2006stateful} uses a different approach for updating summary information to address this issue.

\mysubsection{Problem 4: Events as Transitions\label{sec:events-as-transitions}}
\camready{
The fourth problem, also identified in Jensen \etal\cite{jensen2015stateless}, is that existing stateful DPOR algorithms and most
DPOR algorithms assume that each transition only executes a single
memory operation, whereas an event in our context can consist of many
different memory operations. 
For example, the \code{e$_4$} handler in Figure~\ref{fig:example-missing-execution} reads \code{x} and  \code{y}.}

A related issue is that many DPOR
algorithms assume that they know, ahead of time, the effects of the next
step for each thread.  In our setting, however, since events contain many different memory
operations, we must execute an event to know its effects.
Figure~\ref{fig:example-per-memory} illustrates this problem.  In this
example, we assume that each event can only execute once.

\vspace{-1.5em}
\begin{figure}[!htb]
  \centering
  \begin{center}
  { \footnotesize
  \begin{tabular}{| l | l | l | l |}
    \hline
    \begin{lstlisting}[numbers=none]
    /* Initial condition: */
    x = y = 0;
    \end{lstlisting} &
    \begin{lstlisting}[numbers=none]
    /* e1: */
    y = 1;
    \end{lstlisting} &
    \begin{lstlisting}[numbers=none]
    /* e2: */
    x = 1;
    \end{lstlisting} &
    \begin{lstlisting}[numbers=none]
    /* e3: */
    if (x==1)
        assert(y == 1);
    \end{lstlisting}\\
    \hline
  \end{tabular}
  \includegraphics[width=0.7\linewidth]{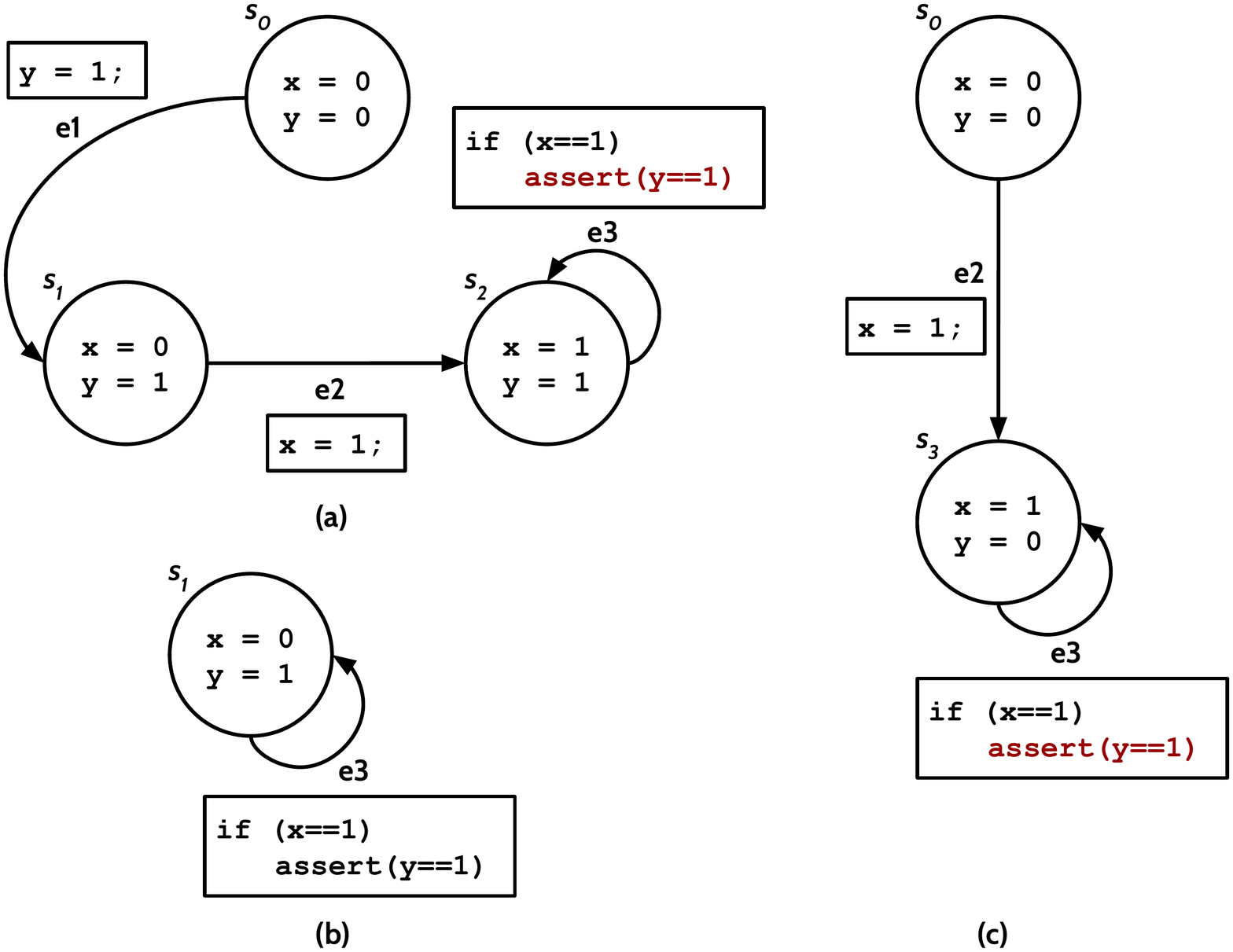}
  }
\end{center}
\vspace{-2em}
  \caption{Example of an event-driven program for which a naive application of the standard DPOR algorithm fails to construct the correct persistent set at state \State{0}.  We assume that \code{e1}, \code{e2}, and \code{e3} are all initially enabled.\label{fig:example-per-memory}}
  \vspace{-1.5em}
\end{figure}

Figure~\ref{fig:example-per-memory}-a shows the first execution of
these 3 events. The stateful DPOR algorithm finds a conflict between the events \code{e$_2$}
and \code{e$_3$}, adds event \code{e$_3$} to the backtrack set for
state \State{1}, and then schedules the second execution shown in
Figure~\ref{fig:example-per-memory}-b. At this point, the exploration
stops prematurely, missing the assertion violating execution
shown in Figure~\ref{fig:example-per-memory}-c.

The key issue is that the set $\{ \code{e$_1$} \}$ is not a persistent
set for state $\State{0}$.  Traditional DPOR algorithms fail to construct the
correct persistent set at state $\State{0}$ because the backtracking
algorithm finds that the transition for event \code{e$_3$} conflicts with the
transition for event \code{e$_2$} and stops setting backtracking points.  This
occurs since these algorithms do not separately track conflicts
from different memory operations in an event when adding
backtracking points---they simply assume transitions are comprised of
single memory operations.  Separately tracking different
operations would allow these algorithms to find a conflict
relation between the events \code{e$_1$} and
\code{e$_3$} (as both access the variable \code{y}) in the first execution, put event \code{e$_2$} into the backtrack set of $\State{0}$, and explore
the missing execution shown in Figure~\ref{fig:example-per-memory}-c.

\myparagraph{Our Idea}
In the classical
DPOR algorithm, transitions correspond to single instructions whose
effects can be determined ahead of time without executing the instructions~\cite{flanagan2005dynamic}.  Thus,
the DPOR algorithm assumes
that the effects of each thread's next transition are known.  Our
events on the other hand include many instructions, and thus,
as Jensen \etal\cite{jensen2015stateless} observes, determining the effects of an
event requires executing the event.  Our algorithm therefore
determines the effects of a transition when the transition is actually
executed.

A second consequence of having events as transitions is that
transitions can access multiple different memory locations.  Thus, as
the example in Figure~\ref{fig:example-per-memory} shows, it does not
suffice to simply set a backtracking point at the last conflicting
transition.  To address this issue, our idea is to compute conflicts
on a \emph{per-memory-location} basis.

\mysection{Stateful Dynamic Partial Order Reduction\label{sec:algorithm}}

\mycomment{
This section presents our algorithm, which extends DPOR to support stateful model checking of event-driven applications with cyclic state spaces. We first describe our algorithm and then present an example that shows how the algorithm works. 

\subsection{The Algorithm}

We first present the states that our algorithm maintains:
}

This section presents our algorithm, which extends DPOR to support stateful model checking of event-driven applications with cyclic state spaces. We first present the states that our algorithm maintains:

\begin{enumerate}
\item {\bf The transition sequence \Trace{}} contains the new transitions that the current execution explores.  Our algorithm explores a given transition in at most one execution.

\item {\bf The state history \PrevStateTable} is a set of program states that have been visited in completed executions.

\item {\bf The state transition graph \RGraph{}} records the states our algorithm has explored thus far.  Nodes in this graph correspond to program states and edges to transitions between program states.
\end{enumerate}

%\vspace{-2em}
\begin{algorithm}
\DontPrintSemicolon
\SetKwBlock{Begin}{}{end}
\SetAlgoLined
\Begin($\MethodExploreAll {(}{)}$){
$\PrevStateTable := \emptyset$\;  
$\RGraph{} := \emptyset$\;
$\Trace{} := \emptyset$\;
$\MethodExplore {(} \State{0} {)}$\;\label{line:firstexplore}
\While{$\exists \State{}, \: \Backtrack(\State{}) \neq \Done(\State{})$}
      {
$\MethodExplore {(} \State{} {)}$\; \label{line:explore}
      }
}
\caption{Top-level exploration algorithm.\label{alg:sdpormain}}
\end{algorithm}
%\vspace{-2em}

\begin{algorithm}
\small
\DontPrintSemicolon
\SetAlgoLined
\SetKwBlock{Begin}{}{end}
\Begin($\MethodExplore {(} \State{} {)}$)
{
\If{$\Backtrack(\State{}) = \Done(\State{}) \label{line:ifback}$}{ \label{line:enable-event-start}
\uIf{$\Done(\State{}) = \Enabled(\State{})$}{
\uIf{$\Enabled(\State{}) \text{ is not empty}$}{\label{line:checkenabled}
$\text{select} \: \Event{} \: \in \Enabled(\State{})\label{line:evselect}$\;
$\text{remove} \: \Event{} \: \text{from} \: \Done(\State{})$\;
} \Else {
$\text{add} \: \GetStates{}(\Trace{}) \: \text{to} \: \PrevStateTable$\;
$\Trace{} := \emptyset$\;
\Return{}
}
} \Else {
$\text{select} \: \Event{} \: \in \Enabled(\State{}) \setminus \Done(\State{})$\;\label{line:select-enabled-event}
$\text{add} \: \Event{} \: \text{to} \: \Backtrack(\State{})$\;
}
} \label{line:enable-event-end}
\While{$\exists \BacktrackEvent \in \Backtrack(\State{}) \setminus \Done(\State{})$\label{line:while}} {
  $\text{add} \: \BacktrackEvent \: \text{to} \: \State{}.done$\;
  $\Trans{} := \text{\Next} \left(\State{},\BacktrackEvent\right)$\;\label{line:algnext}
  $\State{}' := \Dest\left(\Trans{}\right)$\;\label{line:algdst}
  $\text{add transition} \: \Trans{} \:\text{to}\:\RGraph{}$\;\label{line:addtrans}
  \ForEach{$\Event{} \in \Enabled(\State{}) \setminus \Enabled(\State{}')$ \label{line:add-disabled} }{
    $\text{add} \: \Event \: \text{to} \: \Backtrack(\State{})$\;\label{line:adddisabled}
  }
  $\MethodUpdateBacktrackSet \left(\Trans{}\right)$\;\label{line:call-update-backtrack-set}
  \uIf{$\State{}' \in \PrevStateTable \vee \MethodIsFullCycle\left(\Trans{}\right)$\label{line:termination-condition}} {
    $\MethodUpdateBacktrackSetsFromGraph \left(\Trans{}\right)$\;\label{line:endbacktrack}
    $\text{add} \: \GetStates{}(\Trace{}) \: \text{to} \: \PrevStateTable$\;
    $\Trace{} := \emptyset$\;
  }
  \Else {
    \If{$\State{}' \in \GetStates{}(\Trace{}) $\label{line:cycle-current-exec-1}} {
      $\MethodUpdateBacktrackSetsFromGraph \left(\Trans{}\right)$\label{line:cycle-current-exec-2}
    }
    $\Trace{} := \Trace{}.\Trans{}$ \;
    $\MethodExplore {(} \CurrentState' {)}$\;\label{line:call-explore}
  }
}
}
 \caption{Stateful DPOR algorithm for event-driven applications.\label{alg:sdpor}}
\end{algorithm}

Recall that for each reachable state $\State{} \in \States{}$, our algorithm maintains the $\Backtrack(\State{})$ set that contains the events to be explored at $\State{}$, the $\Done(\State{})$ set that contains the events that have already been explored at $\State{}$, and the $\Enabled(\State{})$ set that contains all events that are enabled at $\State{}$.

Algorithm~\ref{alg:sdpormain} presents the top-level \MethodExploreAll procedure.  This procedure first invokes the \MethodExplore procedure to start model checking from the initial state.  However,
the presence of cycles in the state space means that our backtracking-based search algorithm may occasionally set new backtracking points for states in completed executions.  The \MethodExploreAll procedure thus loops over all states that have unexplored items in their backtrack sets and invokes the \MethodExplore procedure to explore those transitions.

Algorithm~\ref{alg:sdpor} describes the logic of the \MethodExplore{} procedure.
The \code{if} statement in line~\ref{line:ifback} checks if the
current state \State{}'s $\Backtrack$ set is the same as the
current state \State{}'s \Done set.  If so, the algorithm
selects an event to execute in the next transition.  If some
enabled events are not yet explored, it selects an unexplored event
to add to the current state's $\Backtrack$ set.  Otherwise, if
the $\Enabled$ set is not empty, it selects an enabled event to
remove from the \Done set.  Note that this scenario occurs only if
the execution is continuing past a state match to satisfy the termination condition.

Then the \code{while} loop in line~\ref{line:while} selects an event
$\BacktrackEvent$ to execute on the current state \State{} and
executes the event $\BacktrackEvent$ to generate the
transition \Trans{} that leads to a new state $\State{}'$.
At this point, the algorithm knows the memory accesses performed by
the transition \Trans{} and thus can add the event $\BacktrackEvent$
to the backtrack sets of the previous states.  This is done via the procedure \MethodUpdateBacktrackSet{}.

Traditional DPOR algorithms continue an execution until it terminates.  Since our programs may have cyclic state spaces, this would cause the model checker to potentially not terminate.  Our algorithm instead checks the conditions in 
line~\ref{line:termination-condition} to decide whether to terminate the execution.  These checks see whether the new state $\State{}'$ matches a state from a previous execution, or if the current execution revisits a state the current execution previously explored and meets other criteria that are checked in the \MethodIsFullCycle{} procedure.   If so, the algorithm calls the \MethodUpdateBacktrackSetsFromGraph{} procedure to set backtracking points, from transitions reachable from $\Trans{}$, to states that can reach $\Trans{}$.  An execution will also terminate if it reaches a state in which no event is enabled (line~\ref{line:checkenabled}).
It then adds the states from the current transition sequence to the set of previously visited states $\PrevStateTable$, resets the current execution transition sequence $\Trace{}$, and backtracks to start a new execution. 

If the algorithm has reached a state $\State{}'$ that was previously discovered in this execution, it sets backtracking points by calling the \MethodUpdateBacktrackSetsFromGraph{} procedure.  Finally, it updates the transition sequence $\Trace{}$ and calls \MethodExplore{}.

\vspace{-2em}
\begin{algorithm}
\small
\DontPrintSemicolon
\SetAlgoLined
\SetKwBlock{Begin}{}{end}
\Begin($\MethodUpdateBacktrackSetsFromGraph {(} \Trans{s}{)}$)
{
  $\RGraph{t} := \lbrace
    \Trans{}\in\RGraph{} \mid \:\: \Trans{s}\Reachable\Trans{}
    \rbrace$\;
  \ForEach{$\Trans{} \in \RGraph{t}$\label{line:check-cycle-in-prev-1}}
  {
     $\MethodUpdateBacktrackSet \left(\Trans{}\right)$\;\label{line:check-cycle-in-prev-2}
  }
}
 \caption{Procedure that updates the backtrack sets of states in previous executions.\label{alg:backtrack-update-graph}}
\end{algorithm}
\vspace{-2em}

Algorithm~\ref{alg:backtrack-update-graph} shows the \MethodUpdateBacktrackSetsFromGraph{} procedure.  This procedure takes a transition $\Trans{}$ that connects the current execution to a previously discovered state in the transition graph $\RGraph{}$.  Since our algorithm does \emph{not} explore all of the transitions reachable from the previously discovered state, we need to set the backtracking points that would have been set by these skipped transitions. This procedure therefore computes the set of transitions reachable from the destination state of $\Trans{}$ and invokes $\MethodUpdateBacktrackSet{}$ on each of those transitions to set backtracking points.

\vspace{-2em}
\begin{algorithm}
\small
\DontPrintSemicolon
\SetAlgoLined
\SetKwBlock{Begin}{}{end}
\Begin($\MethodIsFullCycle {(} \Trans{} {)}$)
{
\If{$\neg \textit{dst}(\Trans{}) \in \GetStates{}(\Trace{})$}{\Return{$\textit{false}$}}
  $\Trace{}^{\textit{fc}} := \lbrace
    \Trans{j} \in \Trace{} \mid
        i=\GetFirst(\Trace{}, \Dest(\Trans{})), \:\: \text{and} \:\:
        i < j \rbrace \cup \{ \Trans{} \}$\;
  $\Events{\textit{fc}} := \lbrace
    \GetEvent{}(\Trans{}') \mid
        \forall \Trans{}' \in \Trace{}^{\textit{fc}} \: \rbrace$\; \label{line:fullcycle-fc}
  $\Events{\textit{enabled}} := \lbrace
    \Enabled(\Dest{}(\Trans{}')) \mid
        \forall \Trans{}' \in \Trace{}^{\textit{fc}} \: \rbrace$\; \label{line:fullcycle-enabled}
  \Return{$\Events{\textit{fc}} = \Events{\textit{enabled}}\:$\label{line:check-full-cycle}}
}
 \caption{Procedure that checks the looping termination condition: a cycle that contains every event enabled in the cycle.
 \label{alg:complete-full-cycle}}
\end{algorithm}
\vspace{-2em}

Algorithm~\ref{alg:complete-full-cycle} presents
the \MethodIsFullCycle{} procedure.  This procedure first checks if
there is a cycle that contains the transition $\Trans{}$ in the state
space explored by the current execution.  The example from
Figure~\ref{fig:example-non-terminating} shows that such a state match
is not sufficient to terminate the execution as the execution may not
have set the necessary backtracking points.  Our algorithm stops the
exploration of an execution when there is a cycle that has explored
\emph{every event that is enabled in that cycle}.  This ensures that for
every transition \Trans{} in the execution, there is a future
transition \Trans{e} for each enabled event \Event{} in the cycle that
can set a backtracking point if \Trans{} and \Trans{e} conflict.

Algorithm~\ref{alg:backtrack-updatev2} presents the \MethodUpdateBacktrackSet{} procedure, which sets backtracking points.  There are two differences between our algorithm and traditional DPOR algorithms.  First, since our algorithm supports programs with cyclic state spaces, it is possible that the algorithm has discovered multiple paths from the start state $\State{0}$ to the current transition $\Trans{}$.  Thus, the algorithm must potentially set backtracking points on multiple different paths.  We address this issue using a backwards depth first search traversal of the $\RGraph{}$ graph.   Second, since our transitions correspond to events, they may potentially access multiple different memory locations and thus the backtracking algorithm potentially needs to set separate backtracking points for each of these memory locations.

The \MethodUpdateBacktrackSetDFS{} procedure implements a backwards
depth first traversal to set backtracking points.  The procedure takes
the following parameters: \Trans{\text{curr}} is the current transition in the
DFS, \Trans{\text{conf}} is the transition that we are currently setting
a backtracking point for, \Access{} is the set of accesses that the
algorithm searches for conflicts for, and \Transitions{\text{exp}} is
the set of transitions that the algorithm has explored down this
search path.  Recall that accesses consist of both an operation, \ie a read or write, and a memory location.  Conflicts are defined in the usual way---writes to a memory location conflict with reads or writes to the same location.

\vspace{-2em}
\begin{algorithm}
\small
\DontPrintSemicolon
\SetAlgoLined
\SetKwBlock{Begin}{}{end}
\Begin($\MethodUpdateBacktrackSet {(} \Trans{} {)}$)
{
 $\MethodUpdateBacktrackSetDFS\left( \Trans{}, \Trans{}, \textit{accesses}(\Trans{}), \{ \Trans{} \} \right)$
}

\Begin($\MethodUpdateBacktrackSetDFS {(} \Trans{\text{curr}}, \Trans{\text{conf}},  \Access{}, \Transitions{\text{exp}} {)}$)
{
  \ForEach{$\Trans{b} \in \textit{pred}_{\RGraph{}}(\Trans{\text{curr}}) \setminus \Transitions{\text{exp}} $\label{lin:loop}}{
  $\Access{b} := \textit{accesses}(\Trans{b})$\;
  $\Trans{\text{conf}}' := \Trans{\text{conf}}$\;
  \If{$\exists a \in \Access{}, \exists a_{b} \in \Access{b}, \Conf(a, a_{b})$\label{lin:conflict}} {
   \uIf{$\GetEvent(\Trans{\text{conf}}) \in \Enabled(\Src(\Trans{b}))$} {
   $\text{add} \: \GetEvent(\Trans{\text{conf}}) \: \text{to} \: \Backtrack(\Src(\Trans{b}))$\;
   } \Else {
      $\text{add} \: \Enabled(\Src(\Trans{b})) \: \text{to} \: \Backtrack(\Src(\Trans{b}))$\;
   }
   $\Trans{\text{conf}}' := \Trans{b}$\;
  }
  $\Access{r} := \{ a \in \Access{} \mid \neg \exists a_{b} \in \Access{b}, \Conf(a, a_{b})\}$\;  \label{lin:newaccesses}
  $\MethodUpdateBacktrackSetDFS \left( \Trans{b}, \Trans{\text{conf}}', \Access{r}, \Transitions{\text{exp}} \: \cup \{ \Trans{b} \} \right)$
  }
}

 \caption{Procedure that updates the backtrack sets of states for previously executed transitions that conflict with the current transition in the search stack.
  \label{alg:backtrack-updatev2}}
\end{algorithm}
\vspace{-2em}

Line~\ref{lin:loop} loops over each transition \Trans{b} that
immediately precedes transition \Trans{\text{curr}} in the state transition
graph and has not been explored.  Line~\ref{lin:conflict} checks for
conflicts between the accesses of \Trans{b} and the access
set \Access{} for the DFS.  If a conflict is detected, the
algorithm adds the event for transition \Trans{\text{conf}} to the backtrack set.
Line~\ref{lin:newaccesses} removes the accesses that conflicted
with transition \Trans{b}.  The search procedure then recursively
calls itself.  If the current transition \Trans{b} conflicts with the
transition \Trans{\text{conf}} for which we are setting a backtracking point, then
it is possible that the behavior we are interested in for \Trans{\text{conf}}
requires that \Trans{b} be executed first.  Thus, if there is a
conflict between \Trans{b} and \Trans{\text{conf}}, we pass \Trans{b} as the
conflict transition parameter to the recursive calls
to \MethodUpdateBacktrackSetDFS{}.

\rahmadi{
\inarxivversion{\appdx{Appendix~\ref{sec:proof}} proves correctness properties for our DPOR algorithm.}
\invmcaiversion{\appdx{Appendix B} proves correctness properties for our DPOR algorithm.}
\inarxivversion{\appdx{Appendix~\ref{sec:example}} revisits the example shown in}
\invmcaiversion{\appdx{Appendix C} revisits the example shown in} 
Figure~\ref{fig:example-missing-execution}. It describes how our DPOR algorithm explores all executions in Figure~\ref{fig:example-missing-execution}, including Figure~\ref{fig:example-missing-execution}-i.
}

\mycomment{
\mysubsection{Discussion}

Although the algorithm we present is designed to explore executions of event-driven systems, the key ideas are also applicable to model checking multithreaded programs.  To the best of our knowledge, there are no correct algorithms for stateful dynamic partial order reduction for multithreaded applications that potentially do not terminate.  The proposed technique bridges the gap between stateful model checking and DPOR, allowing easy adaptation of many existing DPOR algorithms proposed for stateless model checking~\cite{tasharofi2012transdpor,odpor,statelessmcr,chatterjee2019valuecentricdpor,jensen2015stateless} to the stateful context.
}

\mycomment{
\begin{table*}[!htb]
\vspace{-2em}
  \centering
  { \footnotesize
 \caption{Sample model-checked pairs that finished with DPOR but did not finish (\ie "DNF") without it (\textbf{Time} is in seconds). The complete list of 69 pairs is included in Table~\ref{tab:results-dnt-no-dpor-complete} in Appendix.
\label{tab:results-dnt-no-dpor}}
  \begin{tabular}{| r | p{42mm} | r | r | r | r | r | r | r |}
    \hline
    \textbf{No.} & \textbf{App} & \textbf{Evt.} & \multicolumn{3}{ c |}{\textbf{Without DPOR}} & \multicolumn{3}{ c |}{\textbf{With DPOR}}\\
    \cline{4-9}
        & & & \textbf{States} & \textbf{Trans.} & \textbf{Time} & \textbf{States} & \textbf{Trans.} & \textbf{Time}\\
        %& & & & & \textbf{(s)} & & & \textbf{(s)}\\
    \hline
        1  & initial-state-event-streamer---thermostat-auto-off     & 78 & DNF & DNF & DNF & 7,146 & 25,850 & 3,285 \\ \hline
        2  & unbuffered-event-sender---hvac-auto-off.smartapp       & 78 & DNF & DNF & DNF & 7,123 & 26,016 & 3,432 \\ \hline
        3  & initial-state-event-sender---hvac-auto-off.smartapp    & 78 & DNF & DNF & DNF & 7,007 & 25,220 & 3,215 \\ \hline
        4  & initial-state-event-streamer---hvac-auto-off.smartapp  & 78 & DNF & DNF & DNF & 7,007 & 25,220 & 3,230 \\ \hline
        5  & initialstate-smart-app-v1.2.0---hvac-auto-off.smartapp & 78 & DNF & DNF & DNF & 7,007 & 25,220 & 3,290 \\ \hline
        6  & lighting-director---circadian-daylight                 & 19 & DNF & DNF & DNF & 6,553 & 33,045 & 6,604 \\ \hline
        7  & initial-state-event-streamer---thermostat              & 81 & DNF & DNF & DNF & 5,646 & 26,620 & 2,965 \\ \hline
        8  & forgiving-security---unbuffered-event-sender           & 80 & DNF & DNF & DNF & 5,019 & 45,208 & 6,259 \\ \hline
        9  & forgiving-security---initial-state-event-streamer      & 80 & DNF & DNF & DNF & 4,902 & 44,230 & 5,697 \\ \hline
        10 & forgiving-security---initialstate-smart-app-v1.2.0     & 80 & DNF & DNF & DNF & 4,902 & 44,230 & 5,702 \\
    \hline
  \end{tabular}
  }
  %}
%\vspace{-0.5em}
\end{table*}
}

\mysection{Implementation and Evaluation\label{sec:impeval}}
In this section, we present the implementation of our DPOR algorithm 
(Section~\ref{sec:implementation}) and
its evaluation results (Section~\ref{sec:evaluation}).

\mysubsection{Implementation\label{sec:implementation}}
\camready{
We have implemented the algorithm by extending IoTCheck~\cite{trimanafse20}, a tool that model-checks pairs of Samsung's SmartThings smart home apps and reports conflicting updates to
the same device or global variables from different apps. IoTCheck extends Java Pathfinder, an explicit stateful model checker~\cite{jpf}.
In the implementation, we optimized our DPOR algorithm by caching the results of the graph search when \MethodUpdateBacktrackSetsFromGraph{} is called. The results are cached for each state as a summary of the potentially conflicting transitions that are reachable from the given state (see 
\invmcaiversion{\appdx{Appendix D}).}\inarxivversion{\appdx{Appendix~\ref{sec:optimizing-traversals}}).}}
%We used IoTCheck~\cite{trimanafse20} to test our implementation. IoTCheck was built with a testing harness for event-driven smart home apps from Samsung's SmartThings platform. We extended IoTCheck to also run with our DPOR algorithm.

\rahmadi{
We selected the SmartThings platform because it has an extensive collection of event-driven apps.
The SmartThings official GitHub~\cite{smartthings-github} has an active user community---the repository has been forked more than 84,000 times as of August 2021.}

\camready{
We did not compare our implementation against other systems, \eg event-driven systems~\cite{jensen2015stateless,maiya2016por}. Not only that these systems do not perform stateful model checking and handle cyclic state spaces, but also they implemented their algorithms in different domains: web~\cite{jensen2015stateless} and Android applications~\cite{maiya2016por}---it will not be straightforward to adapt and compare these with our implementation on smart home apps.}

\mysubsection{Evaluation\label{sec:evaluation}}
\myparagraph{Dataset}
Our SmartThings app corpus consists of 198 official and third-party apps
that are taken from the IoTCheck smart home apps dataset~\cite{iotcheck-dataset,trimanafse20}.
These apps were collected from different sources, including the official SmartThings
GitHub~\cite{smartthings-github}. In this dataset, the authors of IoTCheck formed pairs of apps to study the interactions between the apps~\cite{trimanafse20}.

We selected the 1,438 pairs of apps in the Device Interaction category as our benchmarks set. It contains a diverse set of apps and app pairs that are further categorized into 11 subgroups based on various device handlers~\cite{devicehandlers} used in each app. 
For example, the \firecoalarm~\cite{fireco2alarm-app} and the \lockitwhenileave~\cite{lockitwhenileave} apps both control and may interact through a door lock (see Section~\ref{sec:intro}). Hence, they are both categorized as a pair in the \code{Locks} group.
As the authors of IoTCheck noted, these pairs are challenging to model check---IoTCheck did not finish for 412 pairs.

\begin{table*}[!htb]
\vspace{-3em}
  \centering
  { \footnotesize
\caption{Sample model-checked pairs that finished with or without DPOR. \textbf{Evt.} is number of events and \textbf{Time} is in seconds. The complete list of results for 229 pairs that finished with or without DPOR is included in \inarxivversion{Table~\ref{tab:results-both-terminated-complete} in \appdx{Appendices}.}\invmcaiversion{Table A.2 in \appdx{Appendices}.}
\label{tab:results-both-terminated}}
  \begin{tabular}{| r | p{42mm} | r | r | r | r | r | r | r |}
    \hline
    \textbf{No.} & \textbf{App} & \textbf{Evt.} & \multicolumn{3}{ c |}{\textbf{Without DPOR}} & \multicolumn{3}{ c |}{\textbf{With DPOR}}\\
    \cline{4-9}
        & & & \textbf{States} & \textbf{Trans.} & \textbf{Time} & \textbf{States} & \textbf{Trans.} & \textbf{Time}\\
        %& & & & & \textbf{(s)} & & & \textbf{(s)}\\
    \hline
        1  & smart-nightlight--ecobeeAwayFromHome                              & 14 & 16,441 & 76,720  & 5,059 & 11,743 & 46,196 & 5,498 \\ \hline
        2  & step-notifier--ecobeeAwayFromHome                                 & 11 & 14,401 & 52,800  & 4,885 & 11,490 & 35,142 & 5,079 \\ \hline
        3  & smart-security--ecobeeAwayFromHome                                & 11 & 14,301 & 47,608  & 4,385 & 8,187  & 21,269 & 2,980 \\ \hline
        4  & keep-me-cozy--whole-house-fan                                     & 17 & 8,793  & 149,464 & 4,736 & 8,776  & 95,084 & 6,043 \\ \hline
        5  & keep-me-cozy-ii--thermostat-window-check                          & 13 & 8,764  & 113,919 & 4,070 & 7,884  & 59,342 & 4,515 \\ \hline
        6  & step-notifier--mini-hue-controller                                & 6  & 7,967  & 47,796  & 2,063 & 7,907  & 40,045 & 3,582 \\ \hline
        7  & keep-me-cozy--thermostat-mode-director                            & 12 & 7,633  & 91,584  & 3,259 & 6,913  & 49,850 & 3,652 \\ \hline
        8  & lighting-director--step-notifier                                  & 14 & 7,611  & 106,540 & 5,278 & 2,723  & 25,295 & 2,552 \\ \hline
        9  & smart-alarm--DeviceTamperAlarm                                    & 15 & 5,665  & 84,960  & 3,559 & 3,437  & 40,906 & 4,441 \\ \hline
        10 & forgiving-security--smart-alarm                                   & 13 & 5,545  & 72,072  & 3,134 & 4,903  & 52,205 & 5,728 \\
    \hline
  \end{tabular}
  }
  %}
\vspace{-1em}
\end{table*}

\myparagraph{Pair Selection}
In the IoTCheck evaluation, the authors had to exclude 175 problematic pairs. In our evaluation, we further excluded pairs.
First, we excluded pairs that were reported to finish their executions in 10 seconds or less---these typically will generate a small number of states (\ie less than 100) when model checked. 
%This leaves us with 693 pairs as our benchmarks set. 
Next, we further removed redundant pairs across the different 11 subgroups. An app may control different devices, and thus they may use various device handlers in its code. For example, the apps \firecoalarm~\cite{fireco2alarm-app} and \code{groveStreams}~\cite{grovestreams-app} both control door locks and thermostats in their code. Thus, the two apps are categorized as a pair both in the \code{Locks} and \code{Thermostats} subgroups---we need to only include this pair once in our evaluation.
These steps reduced our benchmarks set to 535 pairs.

\myparagraph{Experimental Setup}
Each pair was model checked on an Ubuntu-based server with Intel Xeon quad-core
CPU of 3.5GHz and 32GB of memory---we allocated 28GB of heap space for JVM.
In our experiments, we ran the model checker for every pair for at most 2 hours. We found that the model checker usually ran out of memory for pairs that had to be model checked longer.
Further investigation indicates that these pairs generate too many states even
when run with the DPOR algorithm. 
We observed that many smart home apps generate substantial numbers
of \emph{read-write} and \emph{write-write} conflicts when model checked---this is challenging for any DPOR algorithms.
In our benchmarks set, 300 pairs finished for DPOR and/or no DPOR.

\myparagraph{Results}
Our DPOR algorithm substantially reduced the search space for many pairs.
There are 69 pairs that were \emph{unfinished} (\ie \camready{``Unf''}) without DPOR. These pairs did not finish because their executions exceeded the 2-hour limit, or generated too many states quickly and consumed too much memory, causing the model checker to run out of memory within the first hour of their execution. When run with our DPOR algorithm, these pairs successfully finished---mostly in 1 hour or less. 
\rahmadi{\invmcaiversion{Table A.1}\inarxivversion{Table~\ref{tab:results-dnt-no-dpor-complete}}
in \appdx{Appendices} shows the results for pairs that finished with DPOR but did not finish without DPOR. Most notably, even for the pair \code{initial-state-event-streamer}---\code{thermostat-auto-off} that has the most number of states, our DPOR algorithm successfully finished model checking it within 1 hour.

Next, we discovered that 229 pairs finished when model checked with and without DPOR. Table~\ref{tab:results-both-terminated} shows 10 pairs with the most numbers of states (see the complete results in \invmcaiversion{Table A.2}\inarxivversion{Table~\ref{tab:results-both-terminated-complete}}
in \appdx{Appendices}). These pairs consist of apps that generate substantial numbers
of \emph{read-write} and \emph{write-write} conflicts when model checked with our DPOR algorithm. Thus, our DPOR algorithm did not significantly reduce the states, transitions, and runtimes for these pairs.}

Finally, we found 2 pairs that finished when run without our DPOR algorithm, but did not finish when run with it. These pairs consist of apps that are exceptionally challenging for our DPOR algorithm in terms of numbers of \emph{read-write} and \emph{write-write} conflicts. Nevertheless, these are corner cases---please note that our DPOR algorithm is effective in many pairs.

Overall, our DPOR algorithm achieved a \textbf{2$\times$} state reduction and a \textbf{3$\times$} transition
reduction for the 229 pairs that finished for both DPOR and no DPOR (geometric mean).
Assuming that \camready{``Unf''} is equal to 7,200 seconds (\ie 2 hours) of runtime, we achieved an overall speedup of \textbf{7$\times$} for the 300 pairs (geometric mean).
This is a lower bound runtime for the \camready{``Unf''} cases, in which executions exceeded the 2-hour limit---these pairs could have taken more time to finish.

\section{Related Work\label{sec:related}}

%\todo{Need a comparison with~\cite{godefroidbook}, especially Section 5.2 and Section 6 that talk about a similar problem but with a different technique.}
There has been much work on model checking.  Stateless model checking
techniques do not explicitly track which program states have been
visited and instead focus on enumerating
schedules~\cite{godefroidbook,godefroid1997model,godefroid2005software,musuvathi2007chess}.

To make model checking more efficient, researchers have
also looked into various partial order reduction techniques.
The original partial order reduction techniques (\eg persistent/stubborn sets~\cite{godefroidbook,stubbornsets} and sleep sets~\cite{godefroidbook}) can also be used in the context of cyclic state spaces when combined with a proviso that ensures that executions are not prematurely terminated~\cite{godefroidbook}, and ample sets~\cite{clarke1999model,clarke1999state} that are basically persistent sets with additional conditions. However, the persistent/stubborn set techniques ``suffer from severe fundamental limitation''\cite{flanagan2005dynamic}: the operations and their communication objects in future process executions are difficult or impossible to compute precisely through static analysis, while sleep sets alone only reduce the number of transitions (not states).  Work on \emph{collapses} by Katz and Peled also suffers from the same requirement for a statically known independence relation~\cite{collapses}.

The first DPOR technique was proposed by Flanagan and Godefroid~\cite{flanagan2005dynamic} to address those issues.
The authors introduced a
technique that combats the state space explosion by detecting
\emph{read-write} and \emph{write-write} conflicts on shared variable
on the fly.
Since then, a significant effort has been made to further improve
dynamic partial order reduction~\cite{sen2006automated,sen2006race,lauterburg2010evaluating,saarikivi2012improving,tasharofi2012transdpor}.
Unfortunately, a lot of DPOR algorithms assume the context of shared-memory concurrency in that each transition consists of a single memory operation.
In the context of event-driven applications, each transition is an event that can consist of different memory operations. Thus, we have to execute the event to know its effects and analyze it dynamically on the fly in our DPOR algorithm
(see Section~\ref{sec:events-as-transitions}).

\mycomment{
There has been much work on model checking.  Stateless model checking
techniques do not explicitly track which program states have been
visited and instead focus on enumerating
schedules~\cite{godefroidbook,godefroid1997model,godefroid2005software,musuvathi2007chess}.
To make model checking more efficient, researchers have
also looked into different DPOR techniques.  
The first DPOR technique was proposed by Flanagan and Godefroid~\cite{flanagan2005dynamic}.  
The authors introduced a
technique that combats the state space explosion by detecting
\emph{read-write} and \emph{write-write} conflicts on shared variable
on the fly.
This technique is based on the notion of independence between commutative transitions that borrows a similar concept from \emph{collapses} introduced in~\cite{katz1992defining}.
Since then, a significant effort has been made to further improve
dynamic partial order reduction~\cite{sen2006automated,sen2006race,lauterburg2010evaluating,saarikivi2012improving,tasharofi2012transdpor}.
}

Optimal DPOR~\cite{odpor} seeks to
make stateless model checking more efficient by skipping equivalent
executions.  Maximal causality reduction~\cite{statelessmcr} further
refines the technique with the insight that it is only necessary to
explore executions in which threads read different values.
Value-centric DPOR~\cite{chatterjee2019valuecentricdpor} has the insight that executions are equivalent if all of their loads read from the same writes.  Unfolding~\cite{unfoldings} is an alternative approach to POR for reducing the number of executions to be explored.  The unfolding algorithm involves solving an NP-complete problem to add events to the unfolding.

Recent work has extended these algorithms to handle the TSO and PSO
memory models~\cite{nidhugg,pldipsotso,mcrpsotso} and the release
acquire fragment of C/C++~\cite{statelessrelacq}.  The RCMC tool
implements a DPOR tool that operates on execution graphs for the RC11
memory model~\cite{rcmc}.  SAT solvers have been used to avoid
explicitly enumerating all executions.  SATCheck extends
partial order reduction with the insight that it is only necessary to
explore executions that exhibit new behaviors~\cite{satcheck}.
CheckFence checks code by translating it into
SAT~\cite{checkfence}.
\rahmadi{Other work has also presented techniques orthogonal to DPOR, either in a more general context~\cite{desai2015systematic} or platform specific (\eg Android~\cite{ozkan2015systematic} and Node.js~\cite{loring2017semantics}).
}

Recent work on dynamic partial order reduction for event-driven programs has
developed dynamic partial order reduction algorithms for stateless
model checking of event-driven
applications~\cite{jensen2015stateless,maiya2016por}.
Jensen \etal\cite{jensen2015stateless} consider a model similar to ours in which
an event is treated as a single transition, while Maiya \etal\cite{maiya2016por}
consider a model in which event execution interleaves concurrently
with threads.  Neither of these approaches handle cyclic state spaces
nor consider challenges that arise from stateful model checking.

Recent work on DPOR algorithms reduces the number of executions for programs with critical sections by considering
whether critical sections contain conflicting
operations~\cite{dporlocks}.  This work considers stateless
model checking of multithreaded programs, but like our work it must
consider code blocks that perform multiple memory
operations.

CHESS~\cite{musuvathi2007chess} is designed to find and reproduce
concurrency bugs in C, C++, and C\#.  It systematically explores
thread interleavings using a preemption bounded strategy.  The Inspect
tool combines stateless model checking and stateful model checking to
model check C and C++
code~\cite{inspect1,inspect2,inspect4}.

In stateful model checking, there has also been substantial work such as
SPIN~\cite{spin}, Bogor~\cite{dwyer2003bogor}, and JPF~\cite{jpf}.
In addition to these model checkers, other researchers have 
proposed different techniques to capture 
program states~\cite{musuvathi2002cmc,gueta2007cartesian}.

Versions of JPF include a partial order reduction algorithm.  The
design of this algorithm is not well documented, but some publications
have reverse engineered the pseudocode~\cite{porjpf}.  The algorithm
is naive compared to modern DPOR algorithms---this algorithm simply
identifies accesses to shared variables and adds backtracking points
for all threads at any shared variable access.  
%Researchers have implemented standard stateless DPOR in JPF~\cite{dporjpf}.
\mysection{Conclusion\label{sec:conclusion}}
In this paper, we have presented a new technique
that combines dynamic partial order reduction  with
stateful model checking to model check event-driven
applications with cyclic state spaces.
To achieve this, we introduce two techniques:  a new termination condition for looping executions and a new algorithm for setting backtracking points.    Our technique is the first stateful
DPOR algorithm that can model check event-driven applications with cyclic state spaces.  We have evaluated this work on a benchmark set of smart home apps.  Our results show that our techniques  effectively reduce the search space for these apps.
\inarxivversion{This is the extended version of our paper, with the same title, published at VMCAI 2022.}
\invmcaiversion{An extended version of this paper, including appendices, can \appdx{be found}.}

\camready{
\section*{Acknowledgment}
%\myparagraph{Acknowledgment}
We would like to thank our anonymous reviewers for their thorough comments and feedback. This project was supported partly by the National Science Foundation under grants CCF-2006948, CCF-2102940, CNS-1703598, CNS-1763172, CNS-1907352, CNS-2006437, CNS-2007737, CNS-2106838, CNS-2128653, OAC-1740210 and by the Office of Naval Research under grant N00014-18-1-2037.
}

\bibliographystyle{splncs04}
\bibliography{paper}

\ifthenelse{\value{isvmcaiversion}>0}
{
% VMCAI Paper
}
% else
{
% ArXiv Tech Report
    \newpage
    \appendix
    %{
%\centering
%\section*{Appendices}
%}

\title{Appendices}
\author{}
\institute{}
\maketitle

\renewcommand\thefigure{A.\arabic{figure}}
\renewcommand\thetable{A.\arabic{table}}
%\section{A nice appendix}
\setcounter{figure}{0}
\setcounter{table}{0}

\mysection{Example Event-Driven System\label{sec:example-event-driven-system}}
Figure~\ref{fig:concurrency-model} depicts components of an example
event-driven system.  The application \code{SmartLightApp} in the
figure is developed using this concurrency model and may run in the
cloud.  It has a fixed set of events that are associated with it and
their respective event handlers. It also has logic to decide whether
to perform specific actions based on the triggering events.  For
example, \code{SmartLightApp} can turn on and off a light bulb based
on specific triggers.

%\vspace{-2em}
\begin{figure}[t]
    \centering
	\includegraphics[width=0.8\linewidth]{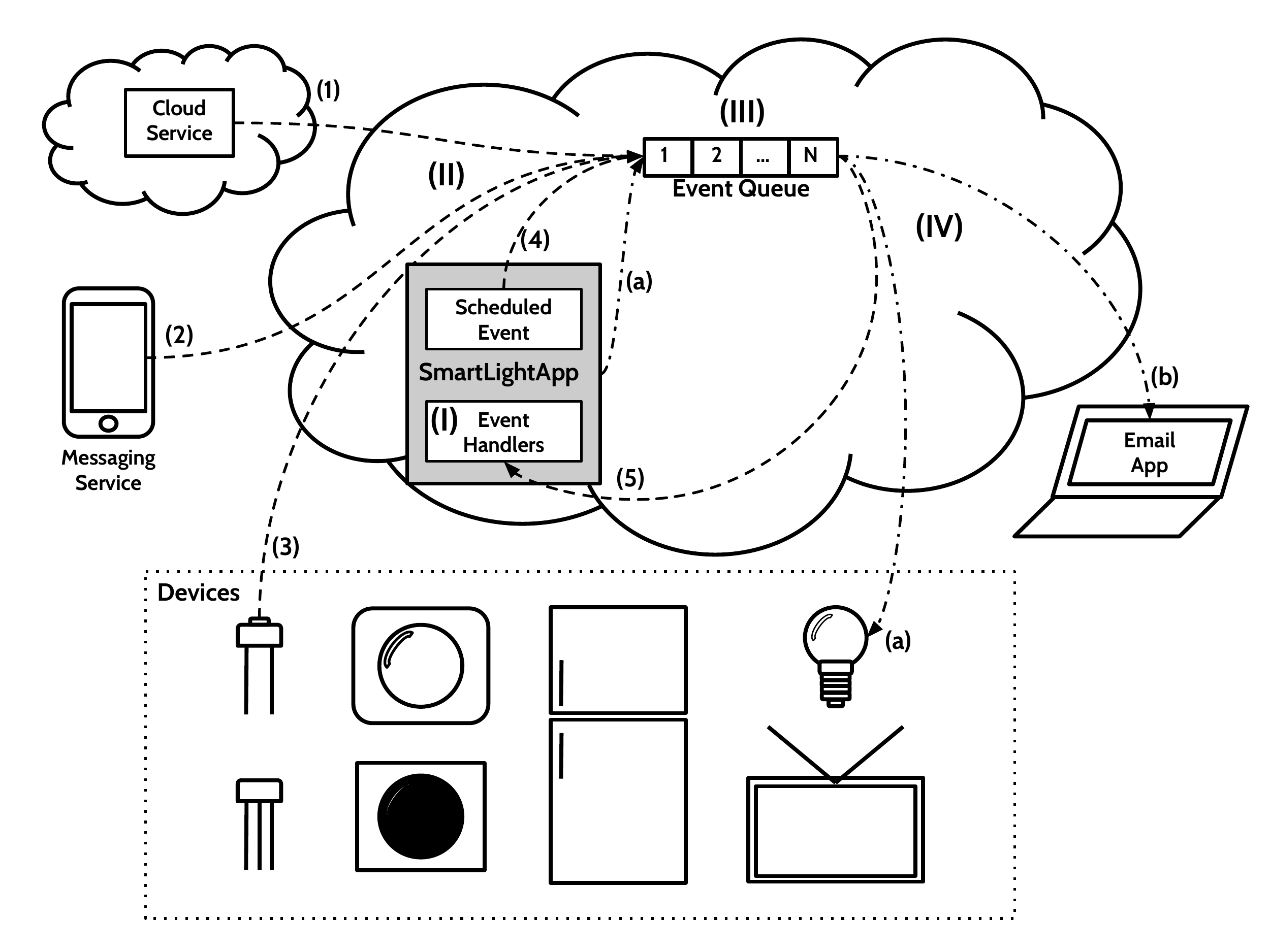}
	%\vspace{-2em}
    \caption{Example Event-driven System. 
      \label{fig:concurrency-model}}
      %\vspace{-1em}
\end{figure}

Our example event-driven system has the following 4 components:

\noindent \textbf{(I) Event handlers:} The application code is
composed of a set of event handlers.  Each event handler processes a
certain class of events.  When the event handler executes, it receives
the appropriate event object.  When the runtime system dequeues an event $\Event{}$, it executes the relevant event handler.

\noindent \textbf{(II) External events:}
An application can have a number of sources that deliver triggering events:
(1) a cloud service
(\eg a cloud service that reports the current outdoor temperature),
(2) a service that is running on a computing device
(\eg a messaging service that sends a trigger whenever there is an incoming message),
(3) a device 
(\eg an illuminance sensor that sends a trigger whenever there is a change of
light intensity in its environment),
(4) a scheduled event within the application
(\eg turn on a light bulb at 6pm), or
(5) an event generated by the application that triggers another component of that app:
these events potentially change the app's state (\eg a scheduled mode change to \textit{night} may trigger a light bulb to turn on).

\noindent \textbf{(III) Internal events:}
An application can directly generate events, for example to:
(a) actuate a certain device
(\eg turn on a light bulb) or
(b) trigger another app running on a computing device
(\eg trigger an email to be sent from an email app).

\noindent \textbf{(IV) Event queue:}
All events are pushed into an event queue to be delivered by the runtime to the appropriate event handlers.

Examples of platforms with this event-driven model include the Samsung SmartThings platform~\cite{smartthings,smartthings-github}, a popular platform used to create
apps that control devices to automate tasks for a smart home.

\section{Correctness\label{sec:proof}}

This section proves the correctness of our stateful DPOR algorithm with the following proof strategy: we first prove that the
stateless analog of our dynamic partial order reduction algorithm is correct in
Theorem~\ref{thm:stateless}.  We then prove that the stateful
algorithm explores all transitions that the stateless algorithm
explores.

Given a transition sequence $\Trace{}$, we use the following notations in our proof:
\begin{itemize}
 \item $\lvert \Trace{} \rvert$ returns the length of $\Trace{}$. 
 \item $\Trace{i}$ denotes the $i$-th transition in $\Trace{}$, and $\Trace{i...j}$ denotes the subsequence of $\Trace{}$ from the $i$-th position to the $j$-th position.
\end{itemize}

We also define the notion of transitive dependence as follows:
\begin{definition}[Transitive Dependence]
Given a transition sequence $\tau$: 
$\State{0} \xrightarrow{\Trans{1}} \State{1} \xrightarrow{\Trans{2}} \State{2} ... \xrightarrow{\Trans{n-1}} \State{n-1} \xrightarrow{\Trans{n}} \State{n}$,
the transitive dependence relation $\rightarrow_{\tau}$ is the smallest partial
order on $\tau$ such that if $i < j$ and $\Trans{i}$ is dependent on $\Trans{j}$, then
$\Trans{i} \rightarrow_{\tau} \Trans{j}$. 
\end{definition}

The notion of transitive dependence enables partial order reduction.
The transition sequence $\tau$ is one linearization of the partial order
$\rightarrow_{\tau}$.  A search algorithm only needs to explore one linearization,
because other linearizations of $\rightarrow_{\tau}$ yield
``equivalent'' transition sequences of $\tau$, which can be obtained
by swapping adjacent independent transitions. 

\mysubsection{Correctness of the Stateless Algorithm}
For the stateless analog of the algorithm, the \code{if} conditions
in lines~\ref{line:termination-condition} and 
\ref{line:cycle-current-exec-1} of Algorithm~\ref{alg:sdpor} are always false, and the data structures
\Trace{} and \PrevStateTable in Algorithm~\ref{alg:sdpormain} and~\ref{alg:sdpor} are \emph{don't-care} terms.
When all events in $\Backtrack(\State{})$ have been explored, the search from $\State{}$ is over, and we say that the state $\State{}$ is ``backtracked''.

%Recall that for an acyclic transition system, a selective search using persistent sets is guaranteed to visit all the reachable local states of every process in the system (see Theorem 6.14 in Godefroid\cite{godefroidbook}).
%Therefore, we phrase Theorem~\ref{thm:stateless} in terms of persistent sets.

\begin{theorem}
\label{thm:stateless}
Let $\Events{}$ be a finite set of events, $\tau$ be a transition sequence
from the initial state $\State{0}$ explored by the stateless analog of
the algorithm in an acyclic transition
system, and $\State{0} \xrightarrow{\tau} \State{}$.  Then, when $\State{}$ is backtracked,
\begin{enumerate}
  \item the set of events that have been explored from $\State{}$ is a persistent set in $\State{}$, and
  \item every trace $\tau \cdot \lambda$ in the state space of $A_G$ is a prefix
    of a linearization of $\rightarrow_{\tau \cdot \lambda'}$ for 
    some explored trace $\tau \cdot \lambda'$. 
\end{enumerate}
\end{theorem}

\begin{proof}
Acyclicity implies that the system
satisfies the descending chain condition, and we can apply well-founded induction. 
We need to prove the statements for the trace
$\State{0} \xrightarrow{\tau} \State{}$,
assuming that the statements hold for all longer traces
$\State{0} \xrightarrow{\tau} \State{} \xrightarrow{\tau'} \State{}'$,
where $\State{}'$ is reachable from $\State{}$ in finite steps.

Let $T := \Backtrack(\State{})$.  For the first statement, \ie
$T$ is a persistent set when $\State{}$ is backtracked, we will prove it by contradiction.
If $T$ is not a persistent set, then there must exist a transition sequence from $\State{}$,
\[ \State{} \xrightarrow{\Trans{1}} \State{1} \xrightarrow{\Trans{2}} \State{2} ... \xrightarrow{\Trans{n-1}} \State{n-1} \xrightarrow{\Trans{n}} \State{n}, \]
where $\GetEvent{(\Trans{1})}, ..., \GetEvent{(\Trans{n})} \notin T$,
events $\GetEvent{(\Trans{1})}, ..., \GetEvent{(\Trans{n-1})}$ are read-write
independent with all events in $T$,
and $\GetEvent{(\Trans{n})}$ is dependent with some $x \in T$ at state $\State{n-1}$.
Select the shortest such transition sequence.
Let $\sigma$ be the transition sequence $\Trans{1} ... \Trans{n-1}$.
For simplicity, we will label $\GetEvent(\Trans{i})$ as $\tildeEvent{i}$ for $1 \leq i \leq n$.
Note that $\tildeEvent{i}$ and $\tildeEvent{j}$ could be the same event even if $i \neq j$.

Since $x$ is read-write independent with $\tildeEvent{i}$ for all
$1 \leq i \leq n-1$, we have the following transition diagram:
\[
\xymatrix{ \State{}' \ar[r]^{\sigma} & \State{n-1}' \\
           \State{} \ar[r]^{\sigma} \ar[u]^{x} & \State{n-1} \ar[u]^{x} }
\]

There are two cases: (1) $\tildeEvent{n}$ is enabled in $\State{}$;
and (2) $\tildeEvent{n}$ is disabled in $\State{}$.

\vspace{.3em}
\textbf{Case 1:} Suppose that $\tildeEvent{n}$ is enabled
in $\State{}$ but disabled in $\State{}'$.
Since $x \in T$ and $\State{} \xrightarrow{x} \State{}'$ has been explored
when $\State{}$ is backtracked, the \code{for} loop at line~\ref{line:add-disabled}
of Algorithm~\ref{alg:sdpor} will add $\tildeEvent{n}$ to $T$, contradicting the
assumption that $\tildeEvent{n} \notin T$.

Suppose that $\tildeEvent{n}$ is enabled in $\State{}$ and $\State{}'$.
Then, $x$ does not enable or disable $\tildeEvent{n}$.  Since $x$
is read-write independent with all corresponding events in $\sigma$,
$\tildeEvent{n}$ is also enabled in $\State{n-1}'$.
Recall that $\tildeEvent{n}$ is dependent with $x$ at $\State{n-1}$.
Since $x$ does not enable or disable $\tildeEvent{n}$,
$x$ and $\tildeEvent{n}$ must have
conflicting memory accesses $\Access{x, n}$ in the transition
sequence $\tau \cdot \sigma \cdot \Trans{x} \cdot \Trans{n}$, which is equivalent to
the transition sequence $\tau \cdot \Trans{x} \cdot \sigma \cdot \Trans{n}$,
where $\GetEvent({\Trans{x}}) = x$. 
Since $x \in T$ and by the inductive assumption,
we have explored some transition sequence $\tilde{\tau}$ where
some linearization of $\rightarrow_{\tilde{\tau}}$ has prefix
$\tau \cdot \Trans{x} \cdot \sigma \cdot \Trans{n}$.
Thus, line~\ref{line:call-update-backtrack-set} in Algorithm~\ref{alg:sdpor}
would have called $\MethodUpdateBacktrackSet{}$ and found the conflicting memory
accesses $\Access{x, n}$, which is not empty until some
$\tildeEvent{i}$ is added to $T$.  Thus, we have a contradiction.

\vspace{.3em}
\textbf{Case 2:} Suppose that $\tildeEvent{n}$ is disabled in $\State{}$.
Since $\tildeEvent{n}$ is enabled in $\State{n-1}$,
there exists some transition $\Trans{i}$ in $\sigma$ that writes to the
shadow location of $\tildeEvent{n}$.  Since $x$ is
read-write independent with $\tildeEvent{1}, ...,
\tildeEvent{n-1}$ by assumption,
$x$ does not write to the shadow location of $\tildeEvent{n}$,
\ie $x$ does not enable or disable $\tildeEvent{n}$.
Then, $\tildeEvent{n}$ is also enabled in $\State{n-1}'$.
Therefore, the same argument in the second paragraph of Case 1 applies,
and we have a contradiction. \\

%By assumption $\sigma$ is the shortest transition sequence
%from $\State{}$ that enables $\GetEvent{(\Trans{n})}$.
%Since $x$ is in $T$, by the inductive assumption, 
%we have explored some transition sequence $\tilde{\tau}$ where
%some linearization of $\rightarrow_{\tilde{\tau}}$ has prefix
%$\tau \cdot \Trans{x} \cdot \sigma \cdot \Trans{n}$.
%Recall that $x$ is dependent with $\tildeEvent{n}$,
%and $x$ does not write to the shadow location of $\tildeEvent{n}$,
%so $x$ and $\Trans{n}$ have conflicting memory accesses $\Access{x, n}$.
%Because $x$ is \todo{read-write} independent with $\tildeEvent{1}, ...,
%\tildeEvent{n-1}$, the set of memory accesses $\Access{x, n}$
%is not empty until some $\Trans{i}$ in $\sigma$
%is added to $T$ by the \MethodUpdateBacktrackSetDFS.
%Thus, $\tildeEvent{i} \in T$ and we have a contradiction. \\

We have shown that $T$ is a persistent set.  
We now move to prove the second statement in this theorem, that every trace
$\State{0} \xrightarrow{\tau} \State{} \xrightarrow{\lambda} u$ is a prefix of a linearization
of $\rightarrow_{\tau \cdot \lambda'}$ for some explored trace $\tau \cdot \lambda'$.
If $\lambda$ is the null sequence, then we are done,
because $\State{0} \xrightarrow{\tau} \State{}$
is already explored. 

Let $\lambda = \pi \cdot \eta$, where $\pi$ is the maximal transition 
sequence such that no events executed in $\pi$ is in $T$.
If $\eta$ is not the null sequence, then the first transition in $\eta$,
denoted by $\eta_1$, is in $T$. 
Since $T$ is a persistent set, $\eta_1$ is independent with all transitions in 
$\pi$.  Thus, $\tau \cdot \pi \cdot \eta$ is equivalent to
$\tau \cdot \eta_1 \cdot \pi \cdot \eta_{2...\lvert \eta \rvert}$.
Note that $\tau \cdot \eta_1$ is a transition sequence
that is explored by the algorithm.  Thus, by the inductive assumption,
$\tau \cdot \eta_1 \cdot \pi \cdot \eta_{2...\lvert \eta \rvert}$
is a prefix of a linearization of $\rightarrow_{\tau \cdot \eta_1 \cdot \lambda'}$
for some explored trace $\tau \cdot \eta_1 \cdot \lambda'$.  
Then, $\tau \cdot \lambda = \tau \cdot \pi \cdot \eta$ is a prefix of 
another linearization of $\rightarrow_{\tau \cdot \eta_1 \cdot \lambda'}$. 
The same argument holds if $\pi$ is the null sequence. 

Let $\eta$ be the null sequence.  Pick $y \in T$.
Then, $\tau \cdot \Trans{y}$ is an explored transition sequence,
where $\GetEvent(\Trans{y}) = y$.
The inductive assumption implies that $\tau \cdot \Trans{y} \cdot \pi$
is a prefix of a linearization of $\rightarrow_{\tau \cdot t_y \cdot \lambda'}$
for some explored trace $\tau \cdot \Trans{y} \cdot \lambda'$.
Because $\Trans{y}$ is independent with all transitions in $\pi$, $\tau \cdot \pi \cdot \Trans{y}$
is also a prefix of some linearization of $\rightarrow_{\tau \cdot t_y \cdot \lambda'}$. 
Hence, $\tau \cdot \pi$ is a prefix of some linearization of
$\rightarrow_{\tau \cdot t_y \cdot \lambda'}$ as well.
\end{proof}

\mysubsection{Correctness of the Stateful Algorithm}

Theorem~\ref{thm:stateless} assumes that the transition system $A_G$ has an acyclic state space.
In general, we do not expect our programs to have acyclic state spaces.  Therefore, we make an acyclic version of $A_G$ by constructing a $k$-bounded instantiation of the transition system $A_G$ in which each event can run at most $k$ times.  This is equivalent to transforming the program to add a counter per event type that permanently disables an event after the event handler has been executed $k$ times.  Definition~\ref{def:kbound} formalizes the notion of the $k$-bounded instantiation.
This would conceptually be implemented by adding a separate mechanism to
the algorithm that only adds events that have not reached their $k$ bound
to the backtrack set and would not add events that have reached
their $k$ bound to the backtrack set in line~\ref{line:adddisabled}
of Algorithm~\ref{alg:sdpor}.

\begin{definition}[Strictly k-Bounded Instantiation]\label{def:kbound}
A $k$-bounded instantiation of a transition system $A_G$ is a program, where
the execution continues until it reaches a state, in which all enabled events have run $k$ times.
\end{definition}

We then will demonstrate the correctness of the results of a run of the stateful algorithm on the original, unbounded transition system $A_G$ by showing that for any arbitrary $k$, there exists a run of the stateless algorithm on a $k$-bounded instantiation of $A_G$ that explores a subset of transitions explored by the stateful algorithm on the original transition system.  Our strict, local notion of a $k$-bounded instantiation of $A_G$ is not sufficient to show this because the stateful algorithm could select an event $\Event{}$ to execute at state $\State{}$ that has already reached its $k$-bound while there exist other events that have not reached their $k$-bound.  Therefore, we define a looser notion of a $k$-bounded instantiation of $A_G$ in Definition~\ref{defn:loosek}.

\begin{definition}[Loosely k-Bounded Instantiation]\label{defn:loosek}
A loosely $k$-bounded instantiation of a transition system $A_G$ is a program,
where the execution continues until it reaches a state, in which all enabled events have run \emph{at least} $k$ times.
\end{definition}

We define a run of the stateful or stateless algorithm as
a full invocation of the \MethodExploreAll procedure,
which explores different executions of the targeted program.
There can be multiple runs of the stateless algorithm on a bounded
instantiation of $A_G$, because adding to and extracting events from the
backtrack sets are non-deterministic in
lines~\ref{line:evselect},~\ref{line:select-enabled-event} and~\ref{line:while} of Algorithm~\ref{alg:sdpor}.
However, Theorem~\ref{thm:stateless} shows that different runs of the stateless
algorithm on a strictly $k$-bounded instantiation are equivalent in the sense that
all of them explore all reachable local states of all the events in
the $k$-bounded instantiation.
%We use $\textit{Execs}_{\Inst{}}$ to denote the set of runs of the stateless algorithm
%on the instantiation $\Inst{}$.
We next prove Lemma~\ref{lemma:kbound} that shows that stateless model checking a loosely $k$-bounded instantiation of $A_G$ is sufficient to explore the behaviors of a strictly $k$-bounded instantiation of $A_G$.

%\todo{Seems that we could have a stronger lemma
%that says for every execution of the $I$ is a prefix of a
%linearization of $\Trace{}$}.

\begin{lemma}\label{lemma:kbound}
Let $k$ be any positive integer, $\Inst{}$ be a loosely $k$-bounded instantiation
of a transition system $A_G$, and $\Inst{0}$ be the strictly $k$-bounded instantiation of $A_G$.
Let $\text{Execs}$ be the set of runs of the stateless algorithm on $\Inst{0}$.
Let $E^{\Inst{}}$ be a run of the stateless algorithm on $I$.  Then, there exists a
run $E \in \text{Execs}$ such that for every transition sequence
$\tau$ explored by $E$, there is a transition sequence
$\tau'$ explored by $E^{\Inst{}}$ such that
$\tau$ is a prefix of a linearization of
$\rightarrow_{\tau'}$.
\end{lemma}

\begin{proof}
%Let $k$ be a positive integer. 
We will prove it by induction on the execution trees
of loosely $k$-bounded instantiations explored by the stateless algorithm.
The idea is that given the strictly $k$-bounded instantiation $\Inst{0}$, we can construct
the execution trees of an arbitrary loosely $k$-bounded instantiation by injecting events one
at a time into the execution trees of $\Inst{0}$.
%Each run of the stateless algorithm on a strictly $k$-bounded instantiation corresponds to
%an execution tree.  

%Let the execution tree for the run $E_{\text{loose}}$ be the execution tree for the loosely $k$-bounded instantiation that we wish to construct.
Let $E^{\Inst{}}$ be a run on the instantiation $\Inst{}$,
we will use induction to construct the execution tree of $E^{\Inst{}}$. 
By definition, $\Inst{0}$ is a loosely $k$-bounded instantiation,
and it is the simplest loosely $k$-bounded instantiation.
Thus, the base case for induction is the strictly $k$-bounded instantiation $\Inst{0}$.
It is clear that for every run of the stateless algorithm on $\Inst{0}$,
there exists a run $E \in \textit{Execs}$ that satisfies this lemma.

For the inductive step, let $\Inst{n}$ be a loosely $k$-bounded instantiation.
We assume that for every run $E^{\Inst{n}}$
of the stateless algorithm on $\Inst{n}$, there exists a run $E \in \textit{Execs}$
such that for every transition sequence $\tau$ explored by $E$,
there is a transition sequence $\tau'$ explored by $E^{\Inst{n}}$
such that $\tau$ is a prefix of a linearization of $\rightarrow_{\tau'}$.
Since we are incrementally constructing the execution tree of $E^{\Inst{}}$,
we can in addition assume that there exists a run $E_{a}^{\Inst{n}}$ of the stateless 
algorithm on $\Inst{n}$ such that the execution tree of $E_{a}^{\Inst{n}}$
has a common prefix as the execution tree of $E^{\Inst{}}$.

We will construct a run of a new loosely $k$-bounded instantiation $\Inst{n+1}$ by
injecting into $E_{a}^{\Inst{n}}$
the first event $\Event{} \in \Events{}$ that is in the execution tree
of $E^{\Inst{}}$ but that is not in the execution tree of $E_{a}^{\Inst{n}}$.
We will label the newly constructed run as $E^{\Inst{n+1}}$.
Let $\Trans{}$ be the injected transition.
%Consider a run $E^{\Inst{n+1}}$ of the stateless algorithm on
%$\Inst{n+1}$.
We will show that there exists a run $E_{b}^{\Inst{n}}$ of
the stateless algorithm on $\Inst{n}$ such that for every transition sequence
$\Trace{}$ explored by $E_{b}^{\Inst{n}}$, there exists a
transition sequence $\Trace{}'$ explored by $E^{\Inst{n+1}}$
where $\Trace{}$ is a prefix of a linearization of
$\rightarrow_{\Trace{}'}$.

Let $\Trace{}$ be a transition sequence explored by
$E_{a}^{\Inst{n}}$ where $\Trans{}$ is injected to obtain $E^{\Inst{n+1}}$.
Let $\tilde{\Trace{}}$ be the corresponding transition sequence
explored by $E^{\Inst{n+1}}$ where $\Trace{}$
and $\tilde{\Trace{}}$ have a common prefix $\tilde{\Trace{}}^{\textit{pre}}$ that
is the prefix of
$\tilde{\Trace{}}$ before the transition $\Trans{}$.
Let $\tilde{T}$ be the subtree of the execution tree of $E^{\Inst{n+1}}$ whose prefix is
$\tilde{\Trace{}}^{\textit{pre}}$.
%, and $T$ be the corresponding
%subtree of the execution tree of $E_{a}^{\Inst{n}}$.

%We will show that for each transition sequence $\Trace{}^{T}$ in $T$,i
%there exists a transition sequence $\Trace{}^{\tilde{T}}$ in $\tilde{T}$ such that
%$\Trace{}^{T}$ is a prefix of a linearization of
%$\rightarrow_{\Trace{}^{\tilde{T}}}$.
Consider all executions in $\tilde{T}$ of the form
$\tilde{\Trace{}}^{\textit{pre}} \cdot \Trans{} \cdot \tilde{\Trace{}}^{\ddagger}$.
We have two cases to consider.

\textbf{Case 1:} All transitions in all
explored  $\tilde{\Trace{}}^{\ddagger}$ are independent with $\Trans{}$.
Consider the backtrack set
$\Backtrack(\Dest{}(\Trans{}))$ in $\tilde{T}$.
We can construct another run $E_{b}^{\Inst{n}}$ of $\Inst{n}$ by 
having the run $E_{a}^{\Inst{n}}$ select the event
$\Event{}'$ that was the first event added to $\Backtrack(\Dest{}(\Trans{}))$,
and add to the
backtrack set of $\GetLast{}(\tilde{\Trace{}}^{\textit{pre}})$ when it
first explores the transition sequence
$\tilde{\Trace{}}^{\textit{pre}}$.
Let $T_b$ be the subtree of the execution tree of $E_{b}^{\Inst{n}}$ whose
prefix is $\tilde{\Trace{}}^\textit{pre}$.  Then, for each transition sequence
$\tilde{\Trace{}}^{\textit{pre}} \cdot \Trace{}^{\#}$ explored
in $T_b$, there is a transition sequence
$\tilde{\Trace{}}^{\textit{pre}} \cdot \Trans{} \cdot \Trace{}^{\#}$
explored in $\tilde{T}$.  Since $\Trans{}$ commutes with $\Trace{}^{\#}$,
the transition sequence
$\tilde{\Trace{}}^{\textit{pre}} \cdot \Trans{} \cdot \Trace{}^{\#}$
is equivalent to
$\tilde{\Trace{}}^{\textit{pre}} \cdot \Trace{}^{\#} \cdot \Trans{}$.
Furthermore, since $\Trans{}$ commutes with all later transitions,
$\tilde{T}$ will set the same set of backtrack points in the transition
sequence $\tilde{\Trace{}}^{\textit{pre}}$ as $T_b$ does.

\textbf{Case 2:} There exists some transition $\Trans{}'$ in
some $\tilde{\Trace{}}^{\ddagger}$ that conflicts with $\Trans{}$.
Consider the first such transition sequence explored by the run
$E^{\Inst{n+1}}$.  This execution would add some event
$\Event{}'$ to the backtrack set
$\Backtrack(\Src{}(\Trans{}))$.  We can construct a new run $E_{b}^{\Inst{n}}$
of $\Inst{n}$ by having the run
$E_{a}^{\Inst{n}}$ select the same event $\Event{}'$ to add to the
backtrack set of $\GetLast{}(\tilde{\Trace{}}^{\textit{pre}})$ when
it first explores the transition sequence
$\tilde{\Trace{}}^{\textit{pre}}$.  Then, the backtrack set at the
state $\GetLast{}(\tilde{\Trace{}}^{\textit{pre}})$ for
$E_{b}^{\Inst{n}}$ would be a subset of the backtrack set at the
state $\GetLast{}(\tilde{\Trace{}}^{\textit{pre}})$ for
$E^{\Inst{n+1}}$. 
Let $T_b$ be the subtree of the execution tree of $E_{b}^{\Inst{n}}$ whose
prefix is $\tilde{\Trace{}}^\textit{pre}$.
Then, any execution explored in $T_b$ would also
be explored in $\tilde{T}$.  Furthermore, since $\tilde{T}$ explores all executions
that are explored in $T_b$, in the transition sequence
$\tilde{\Trace{}}^{\textit{pre}}$ explored by $E^{\Inst{n+1}}$,
the exploration of $\tilde{T}$ will set at least the backtrack points
as $T_b$ does for $\tilde{\Trace{}}^{\textit{pre}}$ in $E_b^{\Inst{n}}$.

We have now shown the existence of $E_b^{\Inst{n}}$ and that for
each transition sequence $\Trace{}$ in $T_b$,
there exists a transition sequence $\Trace{}'$ in $\tilde{T}$ such that
$\Trace{}$ is a prefix of a linearization of
$\rightarrow_{\Trace{}'}$.

Outside of $\tilde{T}$, the execution $E^{\Inst{n+1}}$ will explore at
least the set of executions as $E_b^{\Inst{n}}$ does outside of $T_b$, because
$E^{\Inst{n+1}}$ sets at least the backtracking points in
$\tilde{\Trace{}}^{\textit{pre}}$ as $E_b^{\Inst{n}}$ does
in $\tilde{\Trace{}}^{\textit{pre}}$.
\end{proof}

Since the stateless algorithm does not recognize visited states, we assume
each state encountered during the search of a run $E$ of the stateless algorithm
has a unique identifier.  In other words, there can be multiple nodes (states) in
the $\RGraph{}$ graph of $E$ that are equal but have different identifiers.
However, the stateful algorithm recognizes states that are equal but have different
identifiers.

Therefore, we define a map $\alpha$ that
maps states in $\RGraph{}$ that are equal but have different
identifiers to a single state.
If $t$ is a transition in $\RGraph{}$, then we define
$\alpha(\Trans{})$ to be the transition whose source and destination are both
mapped by $\alpha$. If $\Trace{}$ is a transition sequence
in $\RGraph{}$, then $\alpha(\Trace{})$ is defined similarly.
The state transition graph $\alpha(\RGraph{})$ is the
transformed graph, where nodes that are equal but have different identifiers in
$\RGraph{}$ are collapsed to a single node,
and the edges (transitions) are also collapsed such that
if $\State{1} \rightarrow \State{2}$ is an edge in $\RGraph{}$,
then $\alpha(\State{1}) \rightarrow \alpha(\State{2})$ is an edge in
$\alpha(\RGraph{})$.

In Theorem~\ref{thm:stateful}, although $\alpha(\RGraph{})$ is technically not a
subgraph of $\oRGraph{}$, we can think of $\alpha(\RGraph{})$ as a graph being embedded
in $\oRGraph{}$.  Similarly, if $\State{}$ is a state in $\RGraph{}$, we will think of
$\alpha(\State{})$ as a state in $\oRGraph{}$.  If $\Trans{}$ is a transition 
in $\RGraph{}$, we will think of $\alpha(\Trans{})$ as a transition in $\oRGraph{}$.

\begin{theorem}\label{thm:stateful}
%For each execution $E_\textit{stateful}$ of the stateful algorithm on the unbounded instantiation of the transition system $A_G$, for any finite $k$
Let $\overline{E}$ be a terminating run of the stateful algorithm
on the unbounded instantiation of a transition system $A_G$.
Let $\oRGraph{}$ be the state transition graph when $\overline{E}$
terminates.
Then, for any positive integer $k$,
there is a run $E$ of the stateless algorithm on a loosely $k$-bounded instantiation
of $A_G$ such that if $E$ explores a transition sequence
$\State{0} \rightarrow \State{}'$, then
\begin{enumerate}
  \item $\alpha(\RGraph{}) \subseteq \oRGraph{}$, and
  \item $\forall \State{} \in \States{}, \Backtrack(\State{})_{\textit{stateless}} \subseteq \Backtrack(\alpha(\State{}))_{\textit{stateful}}$,
\end{enumerate}
where $\RGraph{}$ is the state transition graph when
$E$ reaches $\State{}'$ and $\States{}$ is the set of states
that $E$ have encountered when reaching $\State{}'$.
\end{theorem}

%\todo{Do we need the notion of $\State{}'$?  What does it get us?}

\begin{proof}
We will apply the principle of structural induction on the execution tree of
$E$. For the base case, we are at $\State{0}$, and no transitions
have been explored. Thus, $\RGraph{}$ is empty and all backtrack sets are
empty, and we trivially establish
$\alpha(\RGraph{}) \subseteq \oRGraph{}$ and
$\forall \State{} \in \States{}, \Backtrack(\State{})_{\textit{stateless}} \subseteq \Backtrack(\alpha(\State{}))_{\textit{stateful}}$.
 
Let $\State{0} \rightarrow \State{c}$ be a transition sequence
explored by $E$.
For the inductive step, we assume that 
$\alpha(\RGraph{}) \subseteq \oRGraph{}$ and
$\forall \State{} \in \States{}, \Backtrack(\State{})_{\textit{stateless}} \subseteq \Backtrack(\alpha(\State{}))_{\textit{stateful}}$
hold the moment before $\MethodExplore{(\State{c})}$ is called.
We will prove that the invariants hold during the call of $\MethodExplore{(\State{c})}$.
 
The \MethodExplore procedure can be decomposed
into two segments: the initial setup of backtrack
sets in lines~\ref{line:enable-event-start} to~\ref{line:enable-event-end}
and iterations of the while loop in line~\ref{line:while}.

We first consider executions of the code for the initial setup of
backtrack sets in $\MethodExplore(\State{c})$.
Suppose $\Trans{p}$ is the transition that leads to $\State{c}$.
Then, before $\MethodExplore(\State{c})$ is called, $\Trans{p}$ and
$\State{c}$ have already been added to $\RGraph{}$.
Since we have $\alpha(\RGraph{}) \subseteq \oRGraph{}$
by inductive assumption, $\alpha(\State{c})$ is in $\oRGraph{}$
and the stateful algorithm has explored the state $\alpha(\State{c})$.
If the enabled set for $\alpha(\State{c})$ is not empty, then
the stateful algorithm has at least one event in
$\Backtrack(\alpha(\State{c}))_{\textit{stateful}}$.
Construct $E$ to select that same initial event as the one that $\overline{E}$
first added to $\Backtrack(\alpha(\State{c}))_{\textit{stateful}}$
in line~\ref{line:evselect}.

Otherwise, if the enabled set for $\alpha(\State{c})$ was empty, then the
stateful and stateless algorithms would have empty sets for
$\Backtrack(\alpha(\State{c}))_{\textit{stateful}}$
and $\Backtrack(\State{c})_\textit{stateless}$, respectively.
Thus, in either case we preserve the second invariant.
Since the graph $\RGraph{}$ is not changed during the initial
setup of the backtrack sets, we also preserve the first invariant.

We next consider an execution of an iteration of the while loop.
Since in the initial setup of the backtrack set of $\State{c}$, we have
$\Backtrack(\State{c})_\textit{stateless} \subseteq \Backtrack(\alpha(\State{c}))_{\textit{stateful}}$.
Then, the while loop will select a
$b \in \Backtrack(\alpha(\State{}))_\textit{stateful}$.

At line~\ref{line:algnext}, the stateless algorithm will
compute a transition $\Trans{c}$ such that $\alpha(\Trans{c})$
is in $\oRGraph{}$ because $\Next$ is deterministic and $E$
executes the same event $b$ at $\State{c}$ as
$\overline{E}$ does at $\alpha(\State{c})$.

After executing line~\ref{line:addtrans}, we have
$\RGraph{}' = \RGraph{} \cup \{ \Trans{c} \}$.
Since $\alpha(\Trans{c}) \in \oRGraph{}$ and
$\alpha(\RGraph{}) \subseteq \oRGraph{}$, we
have $\alpha(\RGraph{}') \subseteq \oRGraph{}$.  The first invariant holds.

Since \textit{next} and \textit{dst} in lines~\ref{line:algnext} and~\ref{line:algdst}
are deterministic, the stateless algorithm computes $\State{d} = \Dest(\Trans{c})$
in line~\ref{line:algdst} such that $\alpha(\State{d}) = \textit{dst}(\alpha(\Trans{c}))$.
In line~\ref{line:add-disabled}, the enabled set at $\State{c}$
is a subset of that at $\alpha(\State{c})$ because some events in the
stateless algorithm may become disabled due to exceeding their bounds.
Recall that the special mechanism that we use to implement bounded instantiations
ensures that events that are enabled at $\State{c}$ but disabled at $\State{d}$
due to exceeding their bounds are not added to the backtrack set of $\State{c}$.
Therefore, after finishing the for loop in line~\ref{line:add-disabled}, we still have
$\Backtrack(\State{c})_\textit{stateless} \subseteq \Backtrack(\alpha(\State{c}))_{\textit{stateful}}$,
and the second invariant is maintained at this step.

In line~\ref{line:call-update-backtrack-set}, $\MethodExplore(\State{c})$
calls $\MethodUpdateBacktrackSet{}(\Trans{c})$.  Consider an event
$\Event{}$ that the stateless algorithm adds to the backtrack set
of $\State{\textit{pre}}$, where $\State{\textit{pre}}$
is some explored state in $\RGraph{}$.
Then, there must exist a path $\Trace{}^{b}$
in $\RGraph{}$ from $\State{\textit{pre}}$ to $\State{c}$.
Since $\RGraph{} \subseteq \oRGraph{}$,
$\alpha(\Trace{}^{b})$ is also a path from $\alpha(\State{\textit{pre}})$
to $\alpha(\State{c})$ in $\oRGraph{}$.
Consider the last transition $\Trans{\textit{last}}$ added to the path
$\alpha(\Trace{}^{b})$ during the execution of $\overline{E}$.
We have three cases to consider.

\textbf{Case 1:} $\Trans{\textit{last}} = \alpha(\Trans{c})$.
Then, when calling $\MethodUpdateBacktrackSet{}(\alpha(\Trans{c}))$,
$\overline{E}$ would traverse back the path $\alpha(\Trace{}^{b})$
from $\alpha(\State{c})$ to $\alpha(\State{\textit{pre}})$ via a
depth first search, and in the worst case, traverse a path as long as
the length of $\Trace{}^{b}$.  Since $e$ is added to the backtrack set of
$\State{\textit{pre}}$ in the call of $\MethodUpdateBacktrackSet{}(\Trans{c})$,
then $e$ will also be added to the backtrack set of $\alpha(\State{\textit{pre}})$
in the call of $\MethodUpdateBacktrackSet{}(\alpha(\Trans{c}))$.

\textbf{Case 2:} $\Trans{\textit{last}}$ is in the middle of
the path $\alpha(\Trace{}^{b})$, and its destination state
$\textit{dst}(\Trans{\textit{last}})$ was either in $\PrevStateTable$
or satisfied the full cycle predicate.  In this case, $\overline{E}$
calls $\MethodUpdateBacktrackSetsFromGraph{}(\Trans{\textit{last}})$
in line~\ref{line:endbacktrack} some time during its execution,
which will find that $\alpha(\Trans{c})$ is reachable from
$\Trans{\textit{last}}$ and call $\MethodUpdateBacktrackSet{}(\alpha(\Trans{c}))$.
Therefore, $e$ will also be added to the backtrack set of $\alpha(\State{\textit{pre}})$.

\textbf{Case 3:} $\Trans{\textit{last}}$ is in the middle of
the path $\alpha(\Trace{}^{b})$, and its destination
state $\textit{dst}(\Trans{\textit{last}})$ was discovered in the current
execution and does not satisfy the full cycle predicate.
In that case, $\overline{E}$ calls $\MethodUpdateBacktrackSetsFromGraph$
on $\Trans{\textit{last}}$ in line~\ref{line:cycle-current-exec-2}, which will find
that $\Trans{c}$ is reachable from $\Trans{\textit{last}}$
and call $\MethodUpdateBacktrackSet{}(\alpha(\Trans{c}))$.
Therefore, $e$ will also be added to the backtrack set of $\alpha(\State{\textit{pre}})$.

Thus, we have shown that at the end of this loop iteration, both invariants
$\alpha(\RGraph{}) \subseteq \oRGraph{}$ and
$\forall \State{} \in \States{}, \Backtrack(\State{})_{\textit{stateless}} \subseteq \Backtrack(\alpha(\State{}))_{\textit{stateful}}$ hold.
\end{proof}

%\mysubsection{Other ideas}

Finally, we tie all of the theorems together to obtain Theorem~\ref{thm:final},
which relates finite transition sequences of arbitrary length coming from
the initial state $\State{0}$ in the state space of $A_G$ to the
results of the stateful DPOR algorithm. 

\begin{theorem}\label{thm:final}
Let $A_G$ be a transition system with finite reachable states.  Then,
\begin{enumerate}
\item the stateful algorithm on the unbounded instantiation of $A_G$ terminates.
\item For any finite transition sequence $\Trace{}$ from $\State{0}$ in the state space of $A_G$, the state transition graph $\oRGraph{}$ computed by the stateful algorithm run contains a path $\Trace{}'$ such that $\alpha(\Trace{})$ is a prefix of some linearization of $\rightarrow_{\Trace{}'}$.
\end{enumerate}
\end{theorem}

\begin{proof}
(1) We will first show that running the stateful algorithm on the unbounded instantiation
of $A_G$ terminates.  Recall that the unbounded instantiation of $A_G$ means that
each event type is allowed to execute infinite number of times.
Since $A_G$ has finite reachable states, the set of events $\Events{}$ is a finite set.
We claim that each transition sequence explored by $\MethodExplore$ has a finite length,
because the transition sequence will either hit a state stored in $\PrevStateTable$
or eventually satisfy the full cycle condition. 

Let $\Trace{a}$ be a transition sequence explored by $\MethodExplore$.
If $\Trace{a}$ hits a state stored in $\PrevStateTable$, then the claim is proven.
Thus, we will consider the other case and suppose that $\Trace{a}$ never hits a state
stored in $\PrevStateTable$.
Since $A_G$ has finite reachable states, $\Trace{a}$ will eventually revisit some state
$\State{}$ after executing some execution $\Trans{a}$.
At the point, $\MethodIsFullCycle(\Trans{})$ is called.
If the set $\Events{\textit{fc}}$ in line~\ref{line:fullcycle-fc} of
Algorithm~\ref{alg:complete-full-cycle} is not equal to the set $\Events{\textit{enabled}}$
computed in line~\ref{line:fullcycle-enabled}, pick
$\Event{} \in \Events{\textit{enabled}} \setminus \Events{\textit{fc}}$.
Then, $\Event{} \in \Enabled(\Dest(\Trans{e}))$ for some $\Trans{e} \in \Trace{}^{\textit{fc}}_{a}$.
Let $\State{e} = \Dest(\Trans{e})$.

Since $A_G$ has finite reachable states, $\Trace{a}$ will eventually revisit
the state $\State{e}$.  If $\Trace{a}$ does not revisit $\State{e}$, there must
be a state $\State{x}$ that $\Trace{a}$ revisits an infinite amount of times.
However, the first $\lvert\Enabled(\State{x})\rvert$ visits will explore
all events in $\Enabled(\State{x})$ from $\State{x}$, and all later visits to
$\State{x}$ will revisit some other visited states.
If all such later visits to $\State{x}$ do not revisit $\State{e}$, then
there must be a cycle containing $\State{x}$ where all states in the
cycle have had their enabled sets explored, and this cycle will satisfy
the full cycle condition.  Thus, $\Trace{a}$ will either revisit $\State{e}$
or satisfy the full cycle condition somewhere else. 
After revisiting the state $\State{e}$ for $\lvert\Enabled(\State{e})\rvert$
times, $\Trace{a}$ will explore $\Event{}$ from $\State{e}$.

Similarly, $\Trace{a}$ will revisit $\State{}$ or terminate
due to satisfying the full cycle condition. 
If it revisits $\State{}$ the next time, then there is at least
one less element in $\Events{\textit{enabled}} \setminus \Events{\textit{fc}}$. 
Therefore, we can apply the same argument for a finite number of times and show that
the transition sequence $\Trace{a}$ will eventually satisfy the full cycle condition
and terminate.

We have shown that each transition sequence explored by an $\MethodExplore$ call
has finite length.  We will next show that each $\MethodExplore$ call terminates.
Consider a call $\MethodExplore(\State{})$.
We define a number $C$ as the sum of the number of unexplored reachable states
in $A_G$, the number of unexplored transitions in $A_G$.

Since $A_G$ has finite reachable states, for simplicity,
we will ignore unreachable states and transitions
that are potentially infinite and assume all states and transitions in $A_G$
are reachable from $\State{0}$. 
The number of unexplored reachable states and transitions are finite.  So, $C$ is finite.
Consider a subcall $\MethodExplore(\State{\textit{sub}})$ issued by $\MethodExplore(\State{})$.
There are two cases.

\textbf{Case 1:} $\Done(\State{\textit{sub}})$ is not equal
to $\Enabled(\State{\textit{sub}})$.
Then, in one iteration of the while loop in line~\ref{line:while}, the number of 
unexplored transitions is reduced by one.
The number of unexplored states either 
decreases by one or not, depending on whether the most recently explored transition 
leads to an unexplored state or not.

\textbf{Case 2:} $\Done(\State{\textit{sub}})$ is equal to $\Enabled(\State{\textit{sub}})$.
Then, an already explored event is removed from $\Done(\State{\textit{sub}})$,
and the while loop in line~\ref{line:while} reexplores the same event.
In this case, the number of unexplored states and transitions remains the same.
This transition sequence will eventually terminate, and in the worst case,
it terminates without reducing the number of unexplored states and transitions
since visiting the state $\State{\textit{sub}}$.
However, the prefix of such transition sequence before visiting
the state $\State{\textit{sub}}$ must contain a transition that is not explored
in other transition sequences.
Otherwise, the subcall $\MethodExplore(\State{\textit{sub}})$ will not be invoked. 
Therefore, there cannot be an infinite number of $\MethodExplore$ calls where
$C$ is not reduced throughout the call.

Considering the above two cases, $C$ is either reduced or remains the same
in the subcalls of $\MethodExplore$ invoked by $\MethodExplore(\State{})$,
and there is only a finite number of $\MethodExplore$ subcalls
where $C$ stays the same.  Since $C$ is finite, $\MethodExplore(\State{})$
eventually terminates.

(2) Now, we will prove the second statement.
Since $\Trace{}$ is finite, there is a $k$ such that $\Trace{}$ contains fewer than $k$ instances of any event type.
We will assume that all the same states visited by $\Trace{}$ has a unique identifier.
Thus, $\Trace{}$ is a prefix of an execution of a strictly $k$-bounded instantiation of $A_G$.

By Theorem~\ref{thm:stateful}, $\alpha(\RGraph{}) \subseteq \oRGraph{}$
and therefore there exists a run $\tilde{E}$ of the stateless algorithm
on a loosely $k$-bounded instantiation of $A_G$ such that
every transition sequence $\tilde{\Trace{}}$ explored by $\tilde{E}$
is mapped to a path in $\oRGraph{}$ by $\alpha$.
By Lemma~\ref{lemma:kbound}, there exists a run $E$ of the stateless algorithm on
the strictly $k$-bounded instantiation of $A_G$ such that for every
transition sequence $\Trace{a}$ explored by $E$ there is
a transition sequence $\tilde{\Trace{}}$ explored by $\tilde{E}$
where $\Trace{a}$ is a prefix of some linearization of $\rightarrow_{\tilde{\Trace{}}}$.
By Theorem~\ref{thm:stateless} the run $E$
explores some transition sequence $\Trace{a}'$ for which $\Trace{}$ is a prefix of
a linearization of $\rightarrow_{\Trace{a}'}$.

Therefore, the run $\tilde{E}$ explores a transition sequence $\tilde{\Trace{}}_{a}$ such that
$\Trace{}$ is a prefix of a linearization of $\rightarrow_{\tilde{\Trace{}}_{a}}$.
Define $\Trace{}' = \alpha(\tilde{\Trace{}}_{a})$.  Then, $\alpha(\Trace{})$ is a prefix
of a linearization of $\rightarrow_{\Trace{}'}$
\end{proof}

\mycomment{
\mysubsection{Some ideas}

\begin{proof}
\begin{enumerate}
\item Lemma 6.4 shows that $SPACE(\textit{stateless}_k) \subseteq SPACE(\textit{stateless}_k^{\textit{loose}})$ for all $k > 0$. 

\item Theorem 6.5 shows that $SPACE(\textit{stateless}_k^{\textit{loose}}) \subseteq SPACE(\textit{stateful})$ for all $k>0$. So we have
$SPACE(\textit{stateless}_1) \subset SPACE(\textit{stateless}_2) \subset ... 
\subset SPACE(\textit{stateless}_n) \subseteq ... \subseteq SPACE(\textit{stateful})$.

\item If the stateful algorithm terminates, then $SPACE(\textit{stateful})$ is finite.
Thus there is some $k_0 > 0$ such that 
$SPACE(\textit{stateless}_n) = SPACE(\textit{stateless}_{k_0})$ for all $n \geq k_0$.
In other words, running the stateless algorithm on strictly $k_0$-bounded instantiation of
$A_G$ explores all reachable states. 

\item Then running the stateful algorithm on $A_G$ also explores all reachable states. 
\end{enumerate}
\end{proof}

\begin{theorem}
Given a set of events $\Events{}$, if each event $e \in \Events{}$ has
a finite set of reachable states, then SDPOR on unbounded instantiations
of $\Events{}$ terminates and explores all reachable local states for each event.
\end{theorem}

Denote the state space explored by running SDPOR on a $k$-bounded instantiation
on $\Events{}$ by $SPACE(SDPOR_k)$.

\begin{proof}
Let $e \in \Events{}$ be an event.  Since $e$ has a finite set of reachable states,
every such state can be reached by some finite transition sequences. 
Therefore, there exists some $k$ such that all reachable states of $e$
are explored by running SDPOR on the $k$-bounded instantiation of $\Events{}$.

The number of events in $\Events{}$ is finite, so we can find the largest
(denoted by $\hat{k}$) among all such $k$'s such that running SDPOR
on the $\hat{k}$-bounded instantiation explores all reachable local states
for all events. 

\textbf{Claim 1}: SDPOR on the unbounded instantiation terminates, and when it
terminates, the state space explored equals $SPACE(SDPOR_{\hat{k}})$.

\textbf{Claim 2}: $SPACE(SDPOR_1) \subset SPACE(SDPOR_2) \subset ... \subset SPACE(SDPOR_n) \subset ... \subset SPACE(SDPOR_{\hat{k}})$
\end{proof}
}

%\mysubsection{Example\label{sec:example}}
\mysection{Example for Stateful DPOR Algorithm\label{sec:example}}
Here, we illustrate how our new stateful DPOR algorithm works with the example from
Figure~\ref{fig:example-missing-execution}.

First, the algorithm starts by calling the \MethodExploreAll procedure in 
Algorithm~\ref{alg:sdpormain}. After the state is initialized, it calls the \MethodExplore procedure in line~\ref{line:firstexplore}.
Let us suppose that the algorithm first explores the execution shown in 
Figure~\ref{fig:example-missing-execution}-a.
Line~\ref{line:evselect} in Algorithm~\ref{alg:sdpor} chooses
event \code{e$_1$} among the enabled events.
Next, event \code{e$_1$} is executed: it produces a new state $\State{1}$ and adds
the corresponding transition to \RGraph{}.
Line~\ref{line:call-update-backtrack-set} calls the \MethodUpdateBacktrackSet{} procedure,
but it does not update the backtrack set of any state since the algorithm has only
explored one event. 
Line~\ref{line:termination-condition} evaluates to \emph{false} since the termination
condition is not met yet: $\State{1}$ is not found in \PrevStateTable and
it is not a full cycle.
Line~\ref{line:cycle-current-exec-1} also evaluates to \emph{false} since $\State{1}$
is not found in \Trace{}.
Finally, Algorithm~\ref{alg:sdpormain} calls the \MethodExplore procedure recursively to
explore the next transition.

In the subsequent transitions, line~\ref{line:evselect} chooses event \code{e$_2$} and then event \code{e$_3$}.
For each of these events, the \MethodUpdateBacktrackSet{} procedure invokes
\MethodUpdateBacktrackSetDFS{} in Algorithm~\ref{alg:backtrack-updatev2} to 
perform the DFS, find conflicts, and update the corresponding backtrack 
sets---it puts event \code{e$_2$} in the
backtrack set of $\State{0}$ and event \code{e$_3$} in the backtrack set of
$\State{1}$ as it finds conflicts between events \code{e$_1$}
and \code{e$_2$} that access the shared variable \code{z}, and events \code{e$_2$} and \code{e$_3$} that access the shared variable \code{y}.
Event \code{e$_2$} produces a new state $\State{2}$, while \code{e$_3$} 
causes the execution to revisit state $\State{1}$.
At this point, line~\ref{line:cycle-current-exec-1} evaluates
to \emph{true} and the algorithm calls the \MethodUpdateBacktrackSetsFromGraph{} 
procedure.
This procedure (shown in Algorithm~\ref{alg:backtrack-update-graph}) finds that 
the transition that corresponds to event \code{e$_2$} is reachable from 
the transition that corresponds to event \code{e$_3$} in \RGraph{} and calls
the \MethodUpdateBacktrackSet{} procedure. 
\MethodUpdateBacktrackSet{} puts event \code{e$_2$} in the backtrack
set of $\State{2}$ as it finds a conflict between the events \code{e$_2$} and \code{e$_3$}
that access the shared variable \code{y}.

In the final transition of the first execution, line~\ref{line:evselect} chooses event \code{e$_4$}.
The \MethodUpdateBacktrackSet{} procedure finds a conflict between \code{e$_4$} and \code{e$_1$} that both access the shared variable \code{x}:
it puts event \code{e$_4$} in the backtrack set of state $\State{0}$.
Event \code{e$_4$} directs the algorithm to revisit state $\State{0}$. This time
line~\ref{line:termination-condition} evaluates to \emph{true} as the termination
condition is satisfied.
The algorithm calls the \MethodUpdateBacktrackSetsFromGraph{} procedure---it finds
that the events \code{e$_1$}, \code{e$_2$}, and \code{e$_3$} are reachable from the transition for event \code{e$_4$}.
Next, it calls the \MethodUpdateBacktrackSet{} procedure for the reachable
transitions. Among all the conflicts found by the procedure, it
puts event \code{e$_1$} in the backtrack set of state $\State{1}$ as a new backtracking point
for the conflict between the events \code{e$_1$} and \code{e$_4$}.
Then, the algorithm saves the visited states in this execution into \PrevStateTable
and terminates the execution.

Now the algorithm explores the executions in the backtrack sets.
This time, line~\ref{line:while} in Algorithm~\ref{alg:sdpor} chooses 
the events \code{e$_1$} and \code{e$_3$} in the backtrack set of $\State{1}$, and 
event \code{e$_2$} in the backtrack set of $\State{2}$.
These executions terminate quickly as they immediately revisit states
$\State{1}$ and $\State{2}$ that have been stored in \PrevStateTable
at the end of the first execution---line~\ref{line:termination-condition} evaluates to \emph{true}.
These executions are shown in 
Figures~\ref{fig:example-missing-execution}-b, c, and d.

Next, let us suppose that the algorithm explores the execution from event \code{e$_2$}
in the backtrack set of $\State{0}$ shown in Figure~\ref{fig:example-missing-execution}-e.
At first, line~\ref{line:while} in Algorithm~\ref{alg:sdpor} chooses event \code{e$_2$} in the backtrack set of $\State{0}$ that produces a new state $\State{3}$.
When line~\ref{line:call-update-backtrack-set} calls the \MethodUpdateBacktrackSet{} 
procedure, this invocation allows the algorithm to
perform the DFS backward to find conflicts between event \code{e$_2$} and the transitions in
previous executions.
Thus, it finds a conflict between the current event \code{e$_2$} and event \code{e$_3$} from the previous execution
shown in Figure~\ref{fig:example-missing-execution}-b---both access the shared
variable \code{y}. It then puts event \code{e$_2$} into the backtrack set of state $\State{1}$.
Since lines~\ref{line:termination-condition} and~\ref{line:cycle-current-exec-1}
evaluate to \emph{false}, the algorithm calls the \MethodExplore{} procedure recursively.
In the second recursion, let us suppose that event \code{e$_3$} is chosen.
Line~\ref{line:call-update-backtrack-set} calls the \MethodUpdateBacktrackSet{} procedure:
it finds a conflict between events \code{e$_3$} and \code{e$_2$} and puts event \code{e$_3$}
in the backtrack set of $\State{0}$.
Since event \code{e$_3$} directs the execution to revisit $\State{0}$ that is already in \PrevStateTable,
line~\ref{line:termination-condition} evaluates to \emph{true} and
the algorithm invokes the \MethodUpdateBacktrackSetsFromGraph{} procedure.
This procedure finds a conflict between event \code{$e_2$} of the current execution as a 
reachable transition from $\State{0}$ and event \code{$e_3$} from $\State{3}$ that was just executed; it next puts event \code{e$_2$} into the backtrack set of
$\State{3}$ as a new backtracking point.  At this point, this execution terminates.

The algorithm explores executions from the remaining backtracking points,
namely, executions from event \code{e$_2$} in the backtrack set of $\State{3}$, and 
\code{e$_3$} and \code{e$_4$} in the backtrack set of $\State{0}$.
These executions are shown in Figures~\ref{fig:example-missing-execution}-f, g, and h.  In the execution of event \code{e$_4$} from state $\State{0}$, the algorithm detects a conflict with the prior execution of event \code{e$_3$} from state $\State{3}$; it puts event \code{e$_4$} into the backtrack set of $\State{3}$.

The algorithm then explores the execution shown in Figure~\ref{fig:example-missing-execution}-i, which contains the assertion. The model checker stops its exploration and returns an error.

\mysection{Optimizing Traversals\label{sec:optimizing-traversals}}

The algorithm as presented performs graph traversals to identify conflicts.  Our implementation implements several optimizations to reduce traversal overheads.

The procedure \MethodUpdateBacktrackSetsFromGraph{} traverses the graph after a state match to discover potential conflicting transactions from previous executions.  We eliminate many of these graph traversals by caching the results of the graph search the first time \MethodUpdateBacktrackSetsFromGraph is called.  The results can be cached for each state as a summary of the potentially conflicting transitions that are reachable from the given state.
A potentially conflicting transition, namely $\Trans{\text{conf}}$, is cached in the form of a tuple that contains the event \GetEvent{}(\Trans{\text{conf}}) and its corresponding accesses \Access{}.
%: $\langle \GetEvent{}(\Trans{\text{conf}}), \Access{} \rangle$.

This cached summary is updated by \MethodUpdateBacktrackSet{} during its backwards graph traversal.
With this summary, the algorithm can efficiently re-explore the previously explored transitions: (1) when a previously discovered state \State{} is reached, the summary contains all potentially conflicting events reachable from \State{}, and (2) when performing backwards depth first search traversal, the algorithm can stop the traversal at any state \State{b} whenever the algorithm finds that the current \GetEvent{}(\Trans{\text{conf}}) and \Access{} are already cached in the summary for \State{b}.

%\mysection{Experimental Results}
%\vspace{-2em}
\begin{table*}[!htb]
  \centering
    \caption{Model-checked pairs that finished with DPOR but unfinished (\ie "Unf") without DPOR. \textbf{Evt.} is number of events and \textbf{Time} is in seconds.
    \label{tab:results-dnt-no-dpor-complete}}
  { \footnotesize
  \begin{tabular}{| r | p{42mm} | r | c | c | c | r | r | r |}
    \hline
    \textbf{No.} & \textbf{App} & \textbf{Evt.} & \multicolumn{3}{ c |}{\textbf{Without DPOR}} & \multicolumn{3}{ c |}{\textbf{With DPOR}}\\
    \cline{4-9}
        & & & \textbf{States} & \textbf{Trans.} & \textbf{Time} & \textbf{States} & \textbf{Trans.} & \textbf{Time}\\
        %& & & & & \textbf{(s)} & & & \textbf{(s)}\\
    \hline
        1  & initial-state-event-streamer--thermostat-auto-off                 & 78 & Unf & Unf & Unf & 7,146 & 25,850 & 3,285 \\ \hline
        2  & unbuffered-event-sender--hvac-auto-off.smartapp                   & 78 & Unf & Unf & Unf & 7,123 & 26,016 & 3,432 \\ \hline
        3  & initial-state-event-sender--hvac-auto-off.smartapp                & 78 & Unf & Unf & Unf & 7,007 & 25,220 & 3,215 \\ \hline
        4  & initial-state-event-streamer--hvac-auto-off.smartapp              & 78 & Unf & Unf & Unf & 7,007 & 25,220 & 3,230 \\ \hline
        5  & initialstate-smart-app-v1.2.0--hvac-auto-off.smartapp             & 78 & Unf & Unf & Unf & 7,007 & 25,220 & 3,290 \\ \hline
        6  & lighting-director--circadian-daylight                             & 19 & Unf & Unf & Unf & 6,553 & 33,045 & 6,604 \\ \hline
        7  & initial-state-event-streamer--thermostat                          & 81 & Unf & Unf & Unf & 5,646 & 26,620 & 2,965 \\ \hline
        8  & forgiving-security--unbuffered-event-sender                       & 80 & Unf & Unf & Unf & 5,019 & 45,208 & 6,259 \\ \hline
        9  & forgiving-security--initial-state-event-streamer                  & 80 & Unf & Unf & Unf & 4,902 & 44,230 & 5,697 \\ \hline
        10 & forgiving-security--initialstate-smart-app-v1.2.0                 & 80 & Unf & Unf & Unf & 4,902 & 44,230 & 5,702 \\ \hline
        11 & forgiving-security--initial-state-event-sender                    & 80 & Unf & Unf & Unf & 4,902 & 44,230 & 5,716 \\ \hline
        12 & unbuffered-event-sender--thermostat-window-check                  & 79 & Unf & Unf & Unf & 4,546 & 17,411 & 2,069 \\ \hline
        13 & hue-mood-lighting--Hue-Party-Mode                                 & 18 & Unf & Unf & Unf & 3,457 & 49,138 & 6,132 \\ \hline
        14 & thermostat--initial-state-event-sender                            & 82 & Unf & Unf & Unf & 2,530 & 13,105 & 1,621 \\ \hline
        15 & thermostat--initialstate-smart-app-v1.2.0                         & 82 & Unf & Unf & Unf & 2,530 & 13,105 & 1,606 \\ \hline
        16 & thermostat--unbuffered-event-sender                               & 82 & Unf & Unf & Unf & 2,517 & 12,880 & 1,626 \\ \hline
        17 & initial-state-event-streamer--unlock-it-when-i-arrive             & 79 & Unf & Unf & Unf & 2,459 & 8,825  & 998   \\ \hline
        18 & lights-off-with-no-motion-and-presence--smart-security            & 11 & Unf & Unf & Unf & 2,299 & 11,261 & 1,475 \\ \hline
        19 & initial-state-event-streamer--lock-it-when-i-leave                & 78 & Unf & Unf & Unf & 1,750 & 5,387  & 595   \\ \hline
        20 & smart-nightlight--step-notifier                                   & 13 & Unf & Unf & Unf & 1,568 & 6,460  & 644   \\ \hline
        21 & initial-state-event-sender--NotifyIfLeftUnlocked                  & 78 & Unf & Unf & Unf & 1,482 & 3,830  & 428   \\ \hline
        22 & initial-state-event-streamer--NotifyIfLeftUnlocked                & 78 & Unf & Unf & Unf & 1,482 & 3,830  & 439   \\ \hline
  \end{tabular}
  }
  %}
\end{table*}

\begin{table*}[t]
  \centering
  { \footnotesize
  \begin{tabular}{| r | p{42mm} | r | c | c | c | r | r | r |}
    \hline
    \textbf{No.} & \textbf{App} & \textbf{Evt.} & \multicolumn{3}{ c |}{\textbf{Without DPOR}} & \multicolumn{3}{ c |}{\textbf{With DPOR}}\\
    \cline{4-9}
        & & & \textbf{States} & \textbf{Trans.} & \textbf{Time} & \textbf{States} & \textbf{Trans.} & \textbf{Time}\\
        %& & & & & \textbf{(s)} & & & \textbf{(s)}\\
    \hline
        23 & initialstate-smart-app-v1.2.0--NotifyIfLeftUnlocked               & 78 & Unf & Unf & Unf & 1,482 & 3,830  & 437   \\ \hline
        24 & initial-state-event-streamer--auto-lock-door.smartapp             & 80 & Unf & Unf & Unf & 1,272 & 3,743  & 422   \\ \hline
        25 & lock-it-when-i-leave--unbuffered-event-sender                     & 78 & Unf & Unf & Unf & 1,234 & 3,482  & 438   \\ \hline
        26 & lock-it-when-i-leave--initial-state-event-sender                  & 78 & Unf & Unf & Unf & 1,192 & 3,369  & 424   \\ \hline
        27 & lock-it-when-i-leave--initialstate-smart-app-v1.2.0               & 78 & Unf & Unf & Unf & 1,192 & 3,369  & 428   \\ \hline
        28 & medicine-management-temp-motion--initial-state-event-sender       & 79 & Unf & Unf & Unf & 939   & 2,277  & 275   \\ \hline
        29 & medicine-management-temp-motion--initialstate-smart-app-v1.2.0    & 79 & Unf & Unf & Unf & 939   & 2,277  & 275   \\ \hline
        30 & medicine-management-temp-motion--unbuffered-event-sender          & 79 & Unf & Unf & Unf & 933   & 2,277  & 283   \\ \hline
        31 & initial-state-event-streamer--DeviceTamperAlarm                   & 80 & Unf & Unf & Unf & 860   & 2,439  & 277   \\ \hline
        32 & NotifyIfLeftUnlocked--unbuffered-event-sender                     & 78 & Unf & Unf & Unf & 811   & 2,181  & 270   \\ \hline
        33 & initial-state-event-streamer--smart-auto-lock-unlock              & 80 & Unf & Unf & Unf & 780   & 1,986  & 228   \\ \hline
        34 & initial-state-event-streamer--medicine-management-contact-sensor  & 78 & Unf & Unf & Unf & 738   & 1,657  & 197   \\ \hline
        35 & unlock-it-when-i-arrive--initial-state-event-sender               & 77 & Unf & Unf & Unf & 622   & 2,402  & 273   \\ \hline
        36 & unlock-it-when-i-arrive--initialstate-smart-app-v1.2.0            & 77 & Unf & Unf & Unf & 622   & 2,402  & 274   \\ \hline
        37 & initial-state-event-streamer--lock-it-at-a-specific-time          & 79 & Unf & Unf & Unf & 621   & 1,451  & 175   \\ \hline
        38 & unlock-it-when-i-arrive--unbuffered-event-sender                  & 77 & Unf & Unf & Unf & 618   & 2,405  & 277   \\ \hline
        39 & medicine-management-contact-sensor--initial-state-event-sender    & 78 & Unf & Unf & Unf & 605   & 1,241  & 156   \\ \hline
        40 & medicine-management-contact-sensor--initialstate-smart-app-v1.2.0 & 78 & Unf & Unf & Unf & 605   & 1,241  & 163   \\ \hline
        41 & DeviceTamperAlarm--initial-state-event-sender                     & 80 & Unf & Unf & Unf & 602   & 1,540  & 206   \\ \hline
        42 & DeviceTamperAlarm--initialstate-smart-app-v1.2.0                  & 80 & Unf & Unf & Unf & 602   & 1,540  & 209   \\ \hline
  \end{tabular}
  }
  %}
%\vspace{-3em}
\end{table*}

\begin{table*}[t]
  \centering
  { \footnotesize
  \begin{tabular}{| r | p{42mm} | r | c | c | c | r | r | r |}
    \hline
    \textbf{No.} & \textbf{App} & \textbf{Evt.} & \multicolumn{3}{ c |}{\textbf{Without DPOR}} & \multicolumn{3}{ c |}{\textbf{With DPOR}}\\
    \cline{4-9}
        & & & \textbf{States} & \textbf{Trans.} & \textbf{Time} & \textbf{States} & \textbf{Trans.} & \textbf{Time}\\
        %& & & & & \textbf{(s)} & & & \textbf{(s)}\\
    \hline
        43 & medicine-management-contact-sensor--unbuffered-event-sender       & 78 & Unf & Unf & Unf & 602   & 1,240  & 168   \\ \hline
        44 & DeviceTamperAlarm--unbuffered-event-sender                        & 80 & Unf & Unf & Unf & 600   & 1,534  & 217   \\ \hline
        45 & close-the-valve--initial-state-event-sender                       & 78 & Unf & Unf & Unf & 584   & 1,261  & 162   \\ \hline
        46 & close-the-valve--initial-state-event-streamer                     & 78 & Unf & Unf & Unf & 584   & 1,261  & 164   \\ \hline
        47 & close-the-valve--initialstate-smart-app-v1.2.0                    & 78 & Unf & Unf & Unf & 584   & 1,261  & 162   \\ \hline
        48 & close-the-valve--unbuffered-event-sender                          & 78 & Unf & Unf & Unf & 581   & 1,259  & 172   \\ \hline
        49 & initial-state-event-streamer--medicine-management-temp-motion     & 79 & Unf & Unf & Unf & 549   & 1,298  & 172   \\ \hline
        50 & lock-it-at-a-specific-time--initial-state-event-sender            & 79 & Unf & Unf & Unf & 502   & 1,080  & 141   \\ \hline
        51 & lock-it-at-a-specific-time--initialstate-smart-app-v1.2.0         & 79 & Unf & Unf & Unf & 502   & 1,080  & 140   \\ \hline
        52 & lock-it-at-a-specific-time--unbuffered-event-sender               & 79 & Unf & Unf & Unf & 500   & 1,079  & 146   \\ \hline
        53 & auto-lock-door.smartapp--initial-state-event-sender               & 80 & Unf & Unf & Unf & 498   & 1,617  & 196   \\ \hline
        54 & auto-lock-door.smartapp--initialstate-smart-app-v1.2.0            & 80 & Unf & Unf & Unf & 498   & 1,617  & 194   \\ \hline
        55 & auto-lock-door.smartapp--unbuffered-event-sender                  & 80 & Unf & Unf & Unf & 495   & 1,617  & 202   \\ \hline
        56 & initial-state-event-sender--initialstate-smart-app-v1.2.0         & 78 & Unf & Unf & Unf & 473   & 1,054  & 143   \\ \hline
        57 & initial-state-event-streamer--initial-state-event-sender          & 78 & Unf & Unf & Unf & 473   & 1,054  & 143   \\ \hline
        58 & initial-state-event-streamer--initialstate-smart-app-v1.2.0       & 78 & Unf & Unf & Unf & 473   & 1,054  & 143   \\ \hline
        59 & initial-state-event-sender--unbuffered-event-sender               & 78 & Unf & Unf & Unf & 471   & 1,053  & 147   \\ \hline
        60 & initial-state-event-streamer--unbuffered-event-sender             & 78 & Unf & Unf & Unf & 471   & 1,053  & 147   \\ \hline
        61 & initialstate-smart-app-v1.2.0--unbuffered-event-sender            & 78 & Unf & Unf & Unf & 471   & 1,053  & 145   \\ \hline
        62 & enhanced-auto-lock-door--initial-state-event-sender               & 80 & Unf & Unf & Unf & 309   & 1,139  & 158   \\ \hline
        63 & enhanced-auto-lock-door--initial-state-event-streamer             & 80 & Unf & Unf & Unf & 309   & 1,139  & 158   \\ \hline
        64 & enhanced-auto-lock-door--initialstate-smart-app-v1.2.0            & 80 & Unf & Unf & Unf & 309   & 1,139  & 162   \\ \hline
  \end{tabular}
  }
  %}
%\vspace{-3em}
\end{table*}

\begin{table*}[t]
  \centering
  { \footnotesize
  \begin{tabular}{| r | p{42mm} | r | c | c | c | r | r | r |}
    \hline
    \textbf{No.} & \textbf{App} & \textbf{Evt.} & \multicolumn{3}{ c |}{\textbf{Without DPOR}} & \multicolumn{3}{ c |}{\textbf{With DPOR}}\\
    \cline{4-9}
        & & & \textbf{States} & \textbf{Trans.} & \textbf{Time} & \textbf{States} & \textbf{Trans.} & \textbf{Time}\\
        %& & & & & \textbf{(s)} & & & \textbf{(s)}\\
    \hline
        65 & smart-auto-lock-unlock--initial-state-event-sender                & 80 & Unf & Unf & Unf & 309   & 1,140  & 159   \\ \hline
        66 & smart-auto-lock-unlock--initialstate-smart-app-v1.2.0             & 80 & Unf & Unf & Unf & 309   & 1,140  & 161   \\ \hline
        67 & enhanced-auto-lock-door--unbuffered-event-sender                  & 80 & Unf & Unf & Unf & 307   & 1,139  & 164   \\ \hline
        68 & smart-auto-lock-unlock--unbuffered-event-sender                   & 80 & Unf & Unf & Unf & 307   & 1,140  & 165   \\ \hline
        69 & lighting-director--turn-on-before-sunset                          & 11 & Unf & Unf & Unf & 257   & 960    & 77    \\
    \hline
  \end{tabular}
  }
  %}
%\vspace{-3em}
\end{table*}

\begin{table*}[!htb]
  \centering
    \caption{Model-checked pairs that finished with or without DPOR. \textbf{Evt.} is number of events and \textbf{Time} is in seconds.
    \label{tab:results-both-terminated-complete}}
  { \footnotesize
  \begin{tabular}{| r | p{42mm} | r | r | r | r | r | r | r |}
    \hline
    \textbf{No.} & \textbf{App} & \textbf{Evt.} & \multicolumn{3}{ c |}{\textbf{Without DPOR}} & \multicolumn{3}{ c |}{\textbf{With DPOR}}\\
    \cline{4-9}
        & & & \textbf{States} & \textbf{Trans.} & \textbf{Time} & \textbf{States} & \textbf{Trans.} & \textbf{Time}\\
        %& & & & & \textbf{(s)} & & & \textbf{(s)}\\
    \hline
        1  & smart-nightlight--ecobeeAwayFromHome                              & 14 & 16,441 & 76,720  & 5,059 & 11,743 & 46,196 & 5,498 \\ \hline
        2  & step-notifier--ecobeeAwayFromHome                                 & 11 & 14,401 & 52,800  & 4,885 & 11,490 & 35,142 & 5,079 \\ \hline
        3  & smart-security--ecobeeAwayFromHome                                & 11 & 14,301 & 47,608  & 4,385 & 8,187  & 21,269 & 2,980 \\ \hline
        4  & keep-me-cozy--whole-house-fan                                     & 17 & 8,793  & 149,464 & 4,736 & 8,776  & 95,084 & 6,043 \\ \hline
        5  & keep-me-cozy-ii--thermostat-window-check                          & 13 & 8,764  & 113,919 & 4,070 & 7,884  & 59,342 & 4,515 \\ \hline
        6  & step-notifier--mini-hue-controller                                & 6  & 7,967  & 47,796  & 2,063 & 7,907  & 40,045 & 3,582 \\ \hline
        7  & keep-me-cozy--thermostat-mode-director                            & 12 & 7,633  & 91,584  & 3,259 & 6,913  & 49,850 & 3,652 \\ \hline
        8  & lighting-director--step-notifier                                  & 14 & 7,611  & 106,540 & 5,278 & 2,723  & 25,295 & 2,552 \\ \hline
        9  & smart-alarm--DeviceTamperAlarm                                    & 15 & 5,665  & 84,960  & 3,559 & 3,437  & 40,906 & 4,441 \\ \hline
        10 & forgiving-security--smart-alarm                                   & 13 & 5,545  & 72,072  & 3,134 & 4,903  & 52,205 & 5,728 \\ \hline
        11 & smart-light-timer-x-minutes-unless-already-on--ecobeeAwayFromHome & 9  & 3,775  & 11,160  & 992   & 2,460  & 5,418  & 645   \\ \hline
        12 & smart-security--vacation-lighting-director                        & 12 & 3,759  & 33,264  & 1,641 & 2,849  & 14,108 & 1,604 \\ \hline
        13 & smart-security--turn-on-only-if-i-arrive-after-sunset             & 11 & 3,437  & 28,028  & 1,471 & 2,553  & 11,624 & 1,396 \\ \hline
        14 & smart-security--turn-it-on-when-im-here                           & 11 & 3,437  & 28,028  & 1,470 & 2,549  & 11,878 & 1,401 \\ \hline
        15 & vacation-lighting-director--ecobeeAwayFromHome                    & 10 & 3,313  & 11,040  & 856   & 2,158  & 5,779  & 652   \\ \hline
        16 & smart-light-timer-x-minutes-unless-already-on--step-notifier      & 10 & 3,213  & 32,120  & 2,402 & 2,343  & 11,149 & 1,122 \\ \hline
        17 & thermostat--thermostat-mode-director                              & 12 & 3,169  & 38,016  & 1,454 & 3,157  & 28,437 & 2,470 \\ \hline
        18 & ecobeeAwayFromHome--NotifyIfLeftUnlocked                          & 9  & 2,329  & 6,984   & 623   & 1,714  & 377    & 454   \\ \hline
        19 & keep-me-cozy--thermostat-window-check                             & 14 & 2,185  & 30,576  & 1,148 & 2,176  & 20,458 & 1,576 \\ \hline
        20 & smart-security--turn-on-before-sunset                             & 10 & 2,175  & 16,240  & 855   & 1,588  & 6,119  & 782   \\ \hline
        21 & smart-security--turn-on-at-sunset                                 & 10 & 2,175  & 16,240  & 909   & 1,542  & 5,599  & 783   \\ \hline
        22 & keep-me-cozy--thermostat                                          & 12 & 2,017  & 519,960 & 801   & 2,017  & 15,593 & 1,193 \\ \hline
  \end{tabular}
  }
  %}
%\vspace{-3em}
\end{table*}

\begin{table*}[!htb]
  \centering
  { \footnotesize
  \begin{tabular}{| r | p{42mm} | r | r | r | r | r | r | r |}
    \hline
    \textbf{No.} & \textbf{App} & \textbf{Evt.} & \multicolumn{3}{ c |}{\textbf{Without DPOR}} & \multicolumn{3}{ c |}{\textbf{With DPOR}}\\
    \cline{4-9}
        & & & \textbf{States} & \textbf{Trans.} & \textbf{Time} & \textbf{States} & \textbf{Trans.} & \textbf{Time}\\
        %& & & & & \textbf{(s)} & & & \textbf{(s)}\\
    \hline
        23 & smart-security--turn-it-on-when-it-opens                          & 10 & 1,763  & 12,760  & 700   & 1,340  & 5,913  & 737   \\ \hline
        24 & smart-security--undead-early-warning                              & 10 & 1,763  & 12,760  & 688   & 1,340  & 5,913  & 732   \\ \hline
        25 & photo-burst-when--ecobeeAwayFromHome                              & 13 & 1,409  & 4,576   & 599   & 1,109  & 2,695  & 627   \\ \hline
        26 & auto-lock-door.smartapp--ecobeeAwayFromHome                       & 11 & 1,381  & 4,048   & 373   & 927    & 2,357  & 294   \\ \hline
        27 & lighting-director--vacation-lighting-director                     & 13 & 1,373  & 17,836  & 623   & 782    & 4,943  & 395   \\ \hline
        28 & let-there-be-dark--smart-security                                 & 9  & 1,279  & 8,460   & 494   & 881    & 3,531  & 413   \\ \hline
        29 & thermostat--whole-house-fan                                       & 13 & 1,273  & 16,536  & 690   & 747    & 8,237  & 790   \\ \hline
        30 & let-there-be-dark--ecobeeAwayFromHome                             & 7  & 1,240  & 3,052   & 289   & 924    & 1,895  & 263   \\ \hline
        31 & forgiving-security--smart-security                                & 11 & 1,165  & 7,524   & 620   & 696    & 2,838  & 547   \\ \hline
        32 & lighting-director--smart-light-timer-x-minutes-unless-already-on  & 12 & 1,068  & 12,804  & 465   & 424    & 3,104  & 271   \\ \hline
        33 & smart-security--turn-off-with-motion                              & 9  & 989    & 6,732   & 371   & 732    & 3,050  & 394   \\ \hline
        34 & whole-house-fan--hvac-auto-off.smartapp                           & 9  & 958    & 8,613   & 296   & 935    & 6,852  & 722   \\ \hline
        35 & thermostat-auto-off--whole-house-fan                              & 9  & 958    & 8,613   & 286   & 660    & 5,377  & 424   \\ \hline
        36 & smart-security--turn-on-by-zip-code                               & 9  & 941    & 6,300   & 363   & 705    & 2,855  & 386   \\ \hline
        37 & laundry-monitor--step-notifier                                    & 6  & 873    & 3,954   & 263   & 170    & 327    & 54    \\ \hline
        38 & smart-security--DeviceTamperAlarm                                 & 11 & 872    & 7,546   & 371   & 593    & 2,939  & 376   \\ \hline
        39 & make-it-so--whole-house-fan                                       & 13 & 841    & 10,920  & 409   & 751    & 7,436  & 599   \\ \hline
        40 & step-notifier--BetterLaundryMonitor                               & 6  & 810    & 4,854   & 260   & 453    & 1,283  & 144   \\ \hline
        41 & thermostat-auto-off--thermostat-mode-director                     & 10 & 763    & 7,620   & 281   & 714    & 5,805  & 490   \\ \hline
        42 & medicine-management-temp-motion--circadian-daylight               & 14 & 756    & 10,500  & 431   & 350    & 3,133  & 293   \\ \hline
        43 & smart-nightlight--turn-it-on-when-im-here                         & 11 & 736    & 8,085   & 304   & 238    & 1,085  & 110   \\ \hline
        44 & smart-nightlight--turn-on-only-if-i-arrive-after-sunset           & 11 & 736    & 8,085   & 309   & 238    & 1,085  & 109   \\ \hline
        45 & lighting-director--turn-on-at-sunset                              & 11 & 734    & 8,063   & 294   & 241    & 927    & 80    \\ \hline
        46 & unlock-it-when-i-arrive--ecobeeAwayFromHome                       & 8  & 733    & 2,048   & 192   & 543    & 1,179  & 170   \\ \hline
  \end{tabular}
  }
  %}
%\vspace{-3em}
\end{table*}

\begin{table*}[!htb]
  \centering
  { \footnotesize
  \begin{tabular}{| r | p{42mm} | r | r | r | r | r | r | r |}
    \hline
    \textbf{No.} & \textbf{App} & \textbf{Evt.} & \multicolumn{3}{ c |}{\textbf{Without DPOR}} & \multicolumn{3}{ c |}{\textbf{With DPOR}}\\
    \cline{4-9}
        & & & \textbf{States} & \textbf{Trans.} & \textbf{Time} & \textbf{States} & \textbf{Trans.} & \textbf{Time}\\
        %& & & & & \textbf{(s)} & & & \textbf{(s)}\\
    \hline
        47 & hall-light-welcome-home--lighting-director                        & 12 & 675    & 8,088   & 335   & 225    & 1,857  & 172   \\ \hline
        48 & lights-off-with-no-motion-and-presence--ecobeeAwayFromHome        & 7  & 605    & 1,400   & 178   & 400    & 827    & 140   \\ \hline
        49 & good-night-house--ecobeeAwayFromHome                              & 7  & 599    & 1,288   & 170   & 505    & 1,000  & 198   \\ \hline
        50 & good-night-house--ecobeeAwayFromHome                              & 7  & 599    & 1,288   & 179   & 505    & 1,000  & 198   \\ \hline
        51 & good-night-house--ecobeeAwayFromHome                              & 7  & 599    & 1,288   & 170   & 505    & 999    & 200   \\ \hline
        52 & good-night-house--ecobeeAwayFromHome                              & 7  & 599    & 1,288   & 179   & 505    & 999    & 200   \\ \hline
        53 & keep-me-cozy--WindowOrDoorOpen                                    & 10 & 589    & 5,880   & 238   & 589    & 4,338  & 365   \\ \hline
        54 & keep-me-cozy--hvac-auto-off.smartapp                              & 10 & 589    & 5,880   & 190   & 589    & 4,066  & 290   \\ \hline
        55 & keep-me-cozy--thermostat-auto-off                                 & 10 & 589    & 5,880   & 188   & 589    & 4,126  & 283   \\ \hline
        56 & lock-it-at-a-specific-time--ecobeeAwayFromHome                    & 8  & 553    & 1,472   & 142   & 473    & 1,018  & 152   \\ \hline
        57 & darken-behind-me--ecobeeAwayFromHome                              & 8  & 553    & 1,472   & 154   & 437    & 1,019  & 155   \\ \hline
        58 & turn-off-with-motion--ecobeeAwayFromHome                          & 7  & 553    & 1,288   & 144   & 410    & 848    & 133   \\ \hline
        59 & lights-off-with-no-motion-and-presence--step-notifier             & 8  & 506    & 4,040   & 224   & 309    & 1,309  & 142   \\ \hline
        60 & medicine-management-contact-sensor--circadian-daylight            & 13 & 492    & 6,318   & 369   & 310    & 2,729  & 225   \\ \hline
        61  & make-it-so--single-button-controller                                 & 6  & 454 & 906   & 148 & 369 & 712   & 165 \\ \hline
        62  & lighting-director--turn-it-on-when-it-opens                          & 11 & 423 & 4,642 & 199 & 223 & 1,614 & 135 \\ \hline
        63  & lighting-director--undead-early-warning                              & 11 & 423 & 4,642 & 199 & 223 & 1,614 & 133 \\ \hline
        64  & thermostat-window-check--whole-house-fan                             & 10 & 388 & 3,870 & 178 & 388 & 3,409 & 318 \\ \hline
        65  & good-night--BetterLaundryMonitor                                     & 8  & 385 & 3,072 & 125 & 294 & 1,245 & 116 \\ \hline
        66  & hue-minimote--smart-light-timer-x-minutes-unless-already-on          & 8  & 373 & 2,976 & 183 & 312 & 2,201 & 475 \\ \hline
        67  & gentle-wake-up--BetterLaundryMonitor                                 & 7  & 361 & 2,520 & 114 & 361 & 2,309 & 226 \\ \hline
        68  & make-it-so--thermostat-mode-director                                 & 8  & 361 & 2,880 & 137 & 345 & 2,335 & 251 \\ \hline
  \end{tabular}
  }
  %}
%\vspace{-3em}
\end{table*}

\begin{table*}[!htb]
  \centering
  { \footnotesize
  \begin{tabular}{| r | p{42mm} | r | r | r | r | r | r | r |}
    \hline
    \textbf{No.} & \textbf{App} & \textbf{Evt.} & \multicolumn{3}{ c |}{\textbf{Without DPOR}} & \multicolumn{3}{ c |}{\textbf{With DPOR}}\\
    \cline{4-9}
        & & & \textbf{States} & \textbf{Trans.} & \textbf{Time} & \textbf{States} & \textbf{Trans.} & \textbf{Time}\\
        %& & & & & \textbf{(s)} & & & \textbf{(s)}\\
    \hline
        69  & smart-nightlight--turn-on-at-sunset                                  & 10 & 356 & 3,550 & 141 & 333 & 2,861 & 239 \\ \hline
        70  & smart-nightlight--turn-it-on-when-it-opens                           & 10 & 331 & 3,300 & 139 & 101 & 369   & 50  \\ \hline
        71  & smart-nightlight--undead-early-warning                               & 10 & 331 & 3,300 & 140 & 101 & 369   & 50  \\ \hline
        72  & gentle-wake-up--good-night                                           & 8  & 321 & 2,560 & 125 & 321 & 2,376 & 239 \\ \hline
        73  & thermostat--thermostat-window-check                                  & 10 & 313 & 3,120 & 147 & 193 & 1,383 & 159 \\ \hline
        74  & smart-nightlight--vacation-lighting-director                         & 12 & 294 & 3,516 & 133 & 168 & 1,649 & 122 \\ \hline
        75  & lighting-director--turn-on-by-zip-code                               & 10 & 278 & 2,770 & 124 & 115 & 791   & 77  \\ \hline
        76  & smart-nightlight--turn-on-before-sunset                              & 9  & 276 & 2,475 & 92  & 257 & 2,013 & 146 \\ \hline
        77  & good-night--humidity-alert                                           & 8  & 257 & 1,024 & 65  & 243 & 911   & 98  \\ \hline
        78  & thermostat-auto-off--thermostat-window-check                         & 8  & 257 & 2,048 & 86  & 153 & 1,025 & 97  \\ \hline
        79  & double-tap--gentle-wake-up                                           & 7  & 216 & 560   & 58  & 194 & 445   & 67  \\ \hline
        80  & good-night--nfc-tag-toggle                                           & 8  & 215 & 1,712 & 94  & 215 & 1,433 & 149 \\ \hline
        81  & step-notifier--Hue-Party-Mode                                        & 6  & 215 & 1,284 & 82  & 215 & 1,061 & 126 \\ \hline
        82  & step-notifier--turn-on-only-if-i-arrive-after-sunset                 & 6  & 214 & 1,278 & 81  & 104 & 305   & 46  \\ \hline
        83  & Hue-Party-Mode--mini-hue-controller                                  & 4  & 212 & 844   & 42  & 201 & 669   & 86  \\ \hline
        84  & brighten-dark-places--lighting-director                              & 11 & 206 & 2,255 & 111 & 114 & 918   & 89  \\ \hline
        85  & whole-house-fan--WindowOrDoorOpen                                    & 11 & 196 & 2,145 & 120 & 196 & 1,691 & 232 \\ \hline
        86  & smart-nightlight--turn-off-with-motion                               & 9  & 196 & 1,755 & 76  & 87  & 357   & 41  \\ \hline
        87  & smart-nightlight--turn-on-by-zip-code                                & 9  & 196 & 1,755 & 79  & 43  & 131   & 17  \\ \hline
        88  & good-night--the-big-switch                                           & 8  & 191 & 1,520 & 64  & 191 & 1,168 & 103 \\ \hline
        89  & hall-light-welcome-home--step-notifier                               & 6  & 178 & 1,062 & 83  & 40  & 125   & 26  \\ \hline
        90  & smart-light-timer-x-minutes-unless-already-on--turn-on-before-sunset & 7  & 173 & 1,204 & 46  & 163 & 798   & 66  \\ \hline
        91  & smart-light-timer-x-minutes-unless-already-on--turn-on-at-sunset     & 7  & 173 & 1,204 & 50  & 157 & 770   & 68  \\ \hline
        92  & double-tap--good-night                                               & 8  & 161 & 512   & 49  & 149 & 456   & 64  \\ \hline
        93  & make-it-so--thermostat-window-check                                  & 10 & 157 & 1,560 & 80  & 157 & 1,231 & 123 \\ \hline
  \end{tabular}
  }
  %}
%\vspace{-3em}
\end{table*}

\begin{table*}[!htb]
  \centering
  { \footnotesize
  \begin{tabular}{| r | p{42mm} | r | r | r | r | r | r | r |}
    \hline
    \textbf{No.} & \textbf{App} & \textbf{Evt.} & \multicolumn{3}{ c |}{\textbf{Without DPOR}} & \multicolumn{3}{ c |}{\textbf{With DPOR}}\\
    \cline{4-9}
        & & & \textbf{States} & \textbf{Trans.} & \textbf{Time} & \textbf{States} & \textbf{Trans.} & \textbf{Time}\\
        %& & & & & \textbf{(s)} & & & \textbf{(s)}\\
    \hline
        94  & make-it-so--BetterLaundryMonitor                                     & 6  & 145 & 864   & 43  & 145 & 735   & 79  \\ \hline
        95  & make-it-so--thermostat                                               & 8  & 145 & 1,152 & 56  & 145 & 905   & 97  \\ \hline
        96  & rise-and-shine--BetterLaundryMonitor                                 & 6  & 136 & 810   & 43  & 103 & 371   & 74  \\ \hline
        97  & nfc-tag-toggle--single-button-controller                             & 4  & 121 & 160   & 51  & 105 & 128   & 66  \\ \hline
        98  & step-notifier--turn-on-at-sunset                                     & 5  & 117 & 580   & 45  & 68  & 177   & 31  \\ \hline
        99  & step-notifier--turn-on-before-sunset                                 & 5  & 117 & 580   & 44  & 68  & 177   & 30  \\ \hline
        100 & energy-saver--good-night                                             & 8  & 113 & 896   & 45  & 68  & 271   & 29  \\ \hline
        101 & big-turn-on--BetterLaundryMonitor                                    & 6  & 112 & 666   & 35  & 112 & 580   & 60  \\ \hline
        102 & the-big-switch--BetterLaundryMonitor                                 & 5  & 109 & 540   & 32  & 109 & 486   & 63  \\ \hline
        103 & gentle-wake-up--rise-and-shine                                       & 6  & 103 & 612   & 37  & 103 & 511   & 62  \\ \hline
        104 & BetterLaundryMonitor--Hue-Party-Mode                                 & 4  & 103 & 408   & 27  & 73  & 219   & 36  \\ \hline
        105 & 01-control-lights-and-locks-with-contact-sensor--good-night          & 8  & 97  & 768   & 69  & 65  & 399   & 49  \\ \hline
        106 & control-switch-with-contact-sensor--good-night                       & 8  & 97  & 768   & 34  & 65  & 399   & 48  \\ \hline
        107 & big-turn-on--good-night                                              & 7  & 95  & 658   & 33  & 95  & 637   & 64  \\ \hline
        108 & forgiving-security--smart-light-timer-x-minutes-unless-already-on    & 8  & 95  & 752   & 40  & 40  & 212   & 42  \\ \hline
        109 & smart-light-timer-x-minutes-unless-already-on--turn-on-by-zip-code   & 6  & 87  & 516   & 24  & 74  & 349   & 35  \\ \hline
        110 & step-notifier--turn-it-on-when-it-opens                              & 5  & 87  & 430   & 37  & 43  & 111   & 25  \\ \hline
        111 & step-notifier--undead-early-warning                                  & 5  & 87  & 430   & 37  & 43  & 111   & 25  \\ \hline
        112 & brighten-my-path--step-notifier                                      & 5  & 87  & 430   & 37  & 42  & 111   & 26  \\ \hline
        113 & thermostat--WindowOrDoorOpen                                         & 6  & 85  & 504   & 37  & 22  & 44    & 14  \\ \hline
        114 & thermostat--hvac-auto-off.smartapp                                   & 6  & 85  & 504   & 32  & 20  & 47    & 14  \\ \hline
        115 & thermostat--thermostat-auto-off                                      & 6  & 85  & 504   & 30  & 20  & 47    & 12  \\ \hline
        116 & BetterLaundryMonitor--WindowOrDoorOpen                               & 4  & 82  & 324   & 28  & 62  & 192   & 37  \\ \hline
  \end{tabular}
  }
  %}
%\vspace{-3em}
\end{table*}

\begin{table*}[!htb]
  \centering
  { \footnotesize
  \begin{tabular}{| r | p{42mm} | r | r | r | r | r | r | r |}
    \hline
    \textbf{No.} & \textbf{App} & \textbf{Evt.} & \multicolumn{3}{ c |}{\textbf{Without DPOR}} & \multicolumn{3}{ c |}{\textbf{With DPOR}}\\
    \cline{4-9}
        & & & \textbf{States} & \textbf{Trans.} & \textbf{Time} & \textbf{States} & \textbf{Trans.} & \textbf{Time}\\
        %& & & & & \textbf{(s)} & & & \textbf{(s)}\\
    \hline
        117 & good-night--make-it-so                                               & 7  & 81  & 560   & 33  & 81  & 448   & 55  \\ \hline
        118 & good-night--rise-and-shine                                           & 6  & 80  & 474   & 25  & 65  & 294   & 47  \\ \hline
        119 & thermostat-window-check--WindowOrDoorOpen                            & 7  & 76  & 525   & 37  & 75  & 455   & 62  \\ \hline
        120 & make-it-so--auto-lock-door.smartapp                                  & 8  & 73  & 576   & 29  & 73  & 509   & 56  \\ \hline
        121 & good-night--once-a-day                                       & 8 & 73 & 576 & 27 & 57 & 402 & 39 \\ \hline
        122 & brighten-dark-places--step-notifier                          & 5 & 71 & 350 & 35 & 26 & 72  & 17 \\ \hline
        123 & gentle-wake-up--monitor-on-sense                             & 6 & 69 & 408 & 27 & 69 & 348 & 44 \\ \hline
        124 & big-turn-off--good-night                                     & 7 & 67 & 462 & 49 & 59 & 290 & 42 \\ \hline
        125 & good-night--monitor-on-sense                                 & 7 & 65 & 448 & 25 & 40 & 132 & 23 \\ \hline
        126 & hue-minimote--mini-hue-controller                            & 2 & 60 & 118 & 15 & 60 & 118 & 32 \\ \hline
        127 & double-tap--nfc-tag-toggle                                   & 4 & 57 & 64  & 32 & 55 & 57  & 35 \\ \hline
        128 & make-it-so--WindowOrDoorOpen                                 & 6 & 55 & 324 & 25 & 54 & 288 & 44 \\ \hline
        129 & good-night--power-allowance                                  & 7 & 49 & 336 & 25 & 49 & 297 & 43 \\ \hline
        130 & nfc-tag-toggle--the-big-switch                               & 5 & 49 & 240 & 21 & 43 & 125 & 37 \\ \hline
        131 & good-night--turn-it-on-for-5-minutes                         & 7 & 47 & 322 & 24 & 41 & 193 & 31 \\ \hline
        132 & step-notifier--turn-on-by-zip-code                           & 4 & 46 & 180 & 21 & 27 & 62  & 17 \\ \hline
        133 & make-it-so--NotifyIfLeftUnlocked                             & 6 & 43 & 252 & 19 & 43 & 232 & 39 \\ \hline
        134 & make-it-so--hvac-auto-off.smartapp                           & 6 & 43 & 252 & 19 & 43 & 234 & 37 \\ \hline
        135 & make-it-so--thermostat-auto-off                              & 6 & 43 & 252 & 18 & 43 & 234 & 35 \\ \hline
        136 & gentle-wake-up--make-it-so                                   & 5 & 41 & 200 & 19 & 41 & 179 & 32 \\ \hline
        137 & big-turn-on--gentle-wake-up                                  & 5 & 41 & 200 & 17 & 33 & 121 & 24 \\ \hline
        138 & good-night--sunrise-sunset                                   & 8 & 41 & 320 & 36 & 33 & 205 & 43 \\ \hline
        139 & good-night-house--lights-off-with-no-motion-and-presence     & 5 & 40 & 195 & 19 & 40 & 178 & 44 \\ \hline
        140 & single-button-controller--NotifyIfLeftUnlocked               & 4 & 40 & 52  & 35 & 31 & 34  & 31 \\ \hline
        141 & big-turn-off--energy-saver                                   & 6 & 39 & 228 & 34 & 39 & 204 & 31 \\ \hline
        142 & good-night-house--single-button-controller                   & 3 & 37 & 36  & 26 & 37 & 36  & 32 \\ \hline
        143 & double-tap--humidity-alert                                   & 4 & 37 & 48  & 20 & 31 & 34  & 27 \\ \hline
        144 & make-it-so--rise-and-shine                                   & 5 & 36 & 175 & 16 & 36 & 149 & 29 \\ \hline
        145 & good-night-house--make-it-so                                 & 4 & 35 & 136 & 13 & 35 & 128 & 30 \\ \hline
        146 & big-turn-on--rise-and-shine                                  & 5 & 31 & 150 & 14 & 31 & 129 & 23 \\ \hline
        147 & hue-minimote--Hue-Party-Mode                                 & 4 & 31 & 120 & 19 & 31 & 93  & 38 \\ \hline
  \end{tabular}
  }
  %}
%\vspace{-3em}
\end{table*}

\begin{table*}[!htb]
  \centering
  { \footnotesize
  \begin{tabular}{| r | p{42mm} | r | r | r | r | r | r | r |}
    \hline
    \textbf{No.} & \textbf{App} & \textbf{Evt.} & \multicolumn{3}{ c |}{\textbf{Without DPOR}} & \multicolumn{3}{ c |}{\textbf{With DPOR}}\\
    \cline{4-9}
        & & & \textbf{States} & \textbf{Trans.} & \textbf{Time} & \textbf{States} & \textbf{Trans.} & \textbf{Time}\\
        %& & & & & \textbf{(s)} & & & \textbf{(s)}\\
    \hline
        148 & let-there-be-dark--vacation-lighting-director                & 5 & 31 & 150 & 14 & 31 & 150 & 23 \\ \hline
        149 & door-state-to-color-light-hue-bulb--Hue-Party-Mode           & 4 & 30 & 116 & 15 & 27 & 82  & 21 \\ \hline
        150 & monitor-on-sense--BetterLaundryMonitor                       & 3 & 28 & 81  & 12 & 28 & 72  & 20 \\ \hline
        151 & monitor-on-sense--rise-and-shine                             & 5 & 27 & 130 & 14 & 24 & 85  & 22 \\ \hline
        152 & turn-off-with-motion--vacation-lighting-director             & 5 & 27 & 130 & 13 & 24 & 101 & 19 \\ \hline
        153 & lights-off-with-no-motion-and-presence--turn-off-with-motion & 4 & 27 & 104 & 14 & 23 & 82  & 19 \\ \hline
        154 & good-night--smart-turn-it-on                                 & 6 & 25 & 144 & 13 & 25 & 119 & 18 \\ \hline
        155 & make-it-so--monitor-on-sense                                 & 5 & 25 & 120 & 14 & 25 & 102 & 23 \\ \hline
        156 & make-it-so--unlock-it-when-i-arrive                          & 5 & 25 & 120 & 14 & 25 & 102 & 22 \\ \hline
        157 & nfc-tag-toggle--sunrise-sunset                               & 4 & 25 & 96  & 17 & 25 & 86  & 37 \\ \hline
        158 & good-night-house--vacation-lighting-director                 & 4 & 23 & 88  & 12 & 23 & 88  & 23 \\ \hline
        159 & good-night-house--hue-minimote                               & 3 & 22 & 63  & 16 & 22 & 61  & 29 \\ \hline
        160 & darken-behind-me--lights-off-with-no-motion-and-presence     & 5 & 22 & 105 & 19 & 20 & 81  & 30 \\ \hline
        161 & big-turn-on--monitor-on-sense                                & 5 & 20 & 95  & 13 & 20 & 82  & 24 \\ \hline
        162 & auto-lock-door.smartapp--NotifyIfLeftUnlocked                & 4 & 19 & 72  & 16 & 19 & 68  & 20 \\ \hline
        163 & good-night-house--turn-off-with-motion                       & 3 & 18 & 51  & 10 & 18 & 49  & 21 \\ \hline
        164 & nfc-tag-toggle--NotifyIfLeftUnlocked                         & 4 & 18 & 68  & 15 & 16 & 43  & 26 \\ \hline
        165 & good-night-house--auto-lock-door.smartapp                    & 5 & 17 & 80  & 12 & 15 & 60  & 24 \\ \hline
        166 & big-turn-off--power-allowance                                & 5 & 14 & 65  & 11 & 14 & 60  & 19 \\ \hline
        167 & big-turn-on--make-it-so                                      & 4 & 14 & 52  & 10 & 14 & 52  & 17 \\ \hline
        168 & forgiving-security--turn-on-at-sunset                        & 3 & 14 & 39  & 10 & 14 & 28  & 15 \\ \hline
        169 & forgiving-security--turn-on-before-sunset                    & 3 & 14 & 39  & 9  & 14 & 28  & 11 \\ \hline
        170 & big-turn-off--smart-turn-it-on                               & 4 & 13 & 48  & 8  & 13 & 48  & 12 \\ \hline
        171 & its-too-hot--its-too-cold                                    & 2 & 13 & 12  & 10 & 13 & 12  & 17 \\ \hline
        172 & lock-it-at-a-specific-time--make-it-so                       & 5 & 13 & 60  & 10 & 13 & 54  & 15 \\ \hline
        173 & unlock-it-when-i-arrive--auto-lock-door.smartapp             & 5 & 13 & 60  & 12 & 13 & 56  & 21 \\ \hline
        174 & good-night-house--nfc-tag-toggle                             & 2 & 13 & 24  & 9  & 11 & 16  & 18 \\ \hline
  \end{tabular}
  }
  %}
%\vspace{-3em}
\end{table*}

\begin{table*}[!htb]
  \centering
  { \footnotesize
  \begin{tabular}{| r | p{42mm} | r | r | r | r | r | r | r |}
    \hline
    \textbf{No.} & \textbf{App} & \textbf{Evt.} & \multicolumn{3}{ c |}{\textbf{Without DPOR}} & \multicolumn{3}{ c |}{\textbf{With DPOR}}\\
    \cline{4-9}
        & & & \textbf{States} & \textbf{Trans.} & \textbf{Time} & \textbf{States} & \textbf{Trans.} & \textbf{Time}\\
        %& & & & & \textbf{(s)} & & & \textbf{(s)}\\
    \hline
        175 & hall-light-welcome-home--turn-on-at-sunset                   & 3 & 13 & 36  & 10 & 8  & 14  & 8  \\ \hline
        176 & hall-light-welcome-home--turn-on-before-sunset               & 3 & 13 & 36  & 8  & 8  & 14  & 8  \\ \hline
        177 & double-tap--smart-turn-it-on                                 & 2 & 12 & 8   & 11 & 12 & 8   & 10 \\ \hline
        178 & forgiving-security--hall-light-welcome-home                  & 3 & 12 & 33  & 10 & 12 & 24  & 15 \\ \hline
        179 & brighten-my-path--hall-light-welcome-home                    & 3 & 12 & 33  & 9  & 10 & 18  & 9  \\ \hline
        180 & good-night-house--NotifyIfLeftUnlocked                       & 3 & 12 & 33  & 10 & 9  & 20  & 14 \\ \hline
        181 & lock-it-when-i-leave--NotifyIfLeftUnlocked          & 4 & 11 & 40 & 11 & 9  & 26 & 17 \\ \hline
        182 & enhanced-auto-lock-door--NotifyIfLeftUnlocked       & 4 & 10 & 36 & 10 & 10 & 30 & 16 \\ \hline
        183 & forgiving-security--DeviceTamperAlarm               & 4 & 10 & 36 & 10 & 10 & 29 & 14 \\ \hline
        184 & hue-minimote--turn-on-by-zip-code                   & 2 & 10 & 18 & 10 & 10 & 18 & 17 \\ \hline
        185 & darken-behind-me--turn-off-with-motion              & 3 & 9  & 24 & 8  & 9  & 18 & 8  \\ \hline
        186 & enhanced-auto-lock-door--good-night-house           & 5 & 9  & 40 & 13 & 9  & 38 & 28 \\ \hline
        187 & hvac-auto-off.smartapp--WindowOrDoorOpen            & 4 & 9  & 32 & 12 & 9  & 2  & 13 \\ \hline
        188 & lock-it-at-a-specific-time--auto-lock-door.smartapp & 5 & 9  & 40 & 9  & 9  & 36 & 15 \\ \hline
        189 & nfc-tag-toggle--smart-turn-it-on                    & 2 & 9  & 16 & 8  & 9  & 16 & 14 \\ \hline
        190 & thermostat-auto-off--WindowOrDoorOpen               & 4 & 9  & 32 & 10 & 9  & 23 & 13 \\ \hline
        191 & energy-saver--power-allowance                       & 3 & 8  & 21 & 8  & 8  & 21 & 14 \\ \hline
        192 & thermostat-auto-off--hvac-auto-off.smartapp         & 2 & 8  & 14 & 8  & 8  & 14 & 10 \\ \hline
        193 & unlock-it-when-i-arrive--NotifyIfLeftUnlocked       & 3 & 8  & 21 & 9  & 8  & 17 & 13 \\ \hline
        194 & brighten-my-path--turn-on-at-sunset                 & 2 & 7  & 12 & 7  & 7  & 10 & 12 \\ \hline
        195 & brighten-my-path--turn-on-before-sunset             & 2 & 7  & 12 & 7  & 7  & 10 & 8  \\ \hline
        196 & darken-behind-me--good-night-house                  & 2 & 7  & 12 & 8  & 7  & 11 & 11 \\ \hline
        197 & energy-saver--smart-turn-it-on                      & 2 & 7  & 12 & 7  & 7  & 12 & 8  \\ \hline
        198 & turn-on-at-sunset--turn-it-on-when-it-opens         & 2 & 7  & 12 & 7  & 7  & 9  & 11 \\ \hline
  \end{tabular}
  }
  %}
%\vspace{-3em}
\end{table*}

\begin{table*}[!htb]
  \centering
  { \footnotesize
  \begin{tabular}{| r | p{42mm} | r | r | r | r | r | r | r |}
    \hline
    \textbf{No.} & \textbf{App} & \textbf{Evt.} & \multicolumn{3}{ c |}{\textbf{Without DPOR}} & \multicolumn{3}{ c |}{\textbf{With DPOR}}\\
    \cline{4-9}
        & & & \textbf{States} & \textbf{Trans.} & \textbf{Time} & \textbf{States} & \textbf{Trans.} & \textbf{Time}\\
        %& & & & & \textbf{(s)} & & & \textbf{(s)}\\
    \hline
        199 & turn-on-at-sunset--turn-on-before-sunset            & 2 & 7  & 12 & 8  & 7  & 12 & 10 \\ \hline
        200 & turn-on-at-sunset--undead-early-warning             & 2 & 7  & 12 & 7  & 7  & 9  & 11 \\ \hline
        201 & turn-on-before-sunset--turn-it-on-when-it-opens     & 2 & 7  & 12 & 7  & 7  & 9  & 8  \\ \hline
        202 & turn-on-before-sunset--undead-early-warning         & 2 & 7  & 12 & 7  & 7  & 9  & 8  \\ \hline
        203 & brighten-dark-places--turn-on-at-sunset             & 2 & 7  & 12 & 7  & 6  & 8  & 8  \\ \hline
        204 & brighten-dark-places--turn-on-before-sunset         & 2 & 7  & 12 & 7  & 6  & 8  & 8  \\ \hline
        205 & brighten-dark-places--hall-light-welcome-home       & 2 & 7  & 12 & 7  & 5  & 6  & 7  \\ \hline
        206 & hall-light-welcome-home--turn-it-on-when-it-opens   & 2 & 7  & 12 & 7  & 5  & 6  & 7  \\ \hline
        207 & hall-light-welcome-home--turn-on-by-zip-code        & 2 & 7  & 12 & 7  & 5  & 6  & 7  \\ \hline
        208 & hall-light-welcome-home--undead-early-warning       & 2 & 7  & 12 & 7  & 5  & 6  & 7  \\ \hline
        209 & brighten-dark-places--forgiving-security            & 2 & 6  & 10 & 7  & 6  & 8  & 11 \\ \hline
        210 & brighten-my-path--forgiving-security                & 2 & 6  & 10 & 8  & 6  & 8  & 12 \\ \hline
        211 & brighten-my-path--turn-it-on-when-it-opens          & 2 & 6  & 10 & 7  & 6  & 8  & 11 \\ \hline
        212 & brighten-my-path--undead-early-warning              & 2 & 6  & 10 & 7  & 6  & 8  & 10 \\ \hline
        213 & forgiving-security--turn-it-on-when-it-opens        & 2 & 6  & 10 & 7  & 6  & 8  & 11 \\ \hline
        214 & forgiving-security--turn-on-by-zip-code             & 2 & 6  & 10 & 7  & 6  & 8  & 8  \\ \hline
        215 & forgiving-security--undead-early-warning            & 2 & 6  & 10 & 7  & 6  & 8  & 11 \\ \hline
        216 & lock-it-at-a-specific-time--NotifyIfLeftUnlocked    & 3 & 6  & 15 & 8  & 6  & 14 & 11 \\ \hline
        217 & brighten-dark-places--brighten-my-path              & 2 & 6  & 10 & 8  & 5  & 6  & 9  \\ \hline
        218 & enhanced-auto-lock-door--auto-lock-door.smartapp    & 4 & 5  & 16 & 9  & 5  & 15 & 18 \\ \hline
        219 & good-night-house--lock-it-at-a-specific-time        & 2 & 5  & 8  & 7  & 5  & 6  & 8  \\ \hline
        220 & good-night-house--turn-on-by-zip-code               & 2 & 5  & 8  & 7  & 5  & 6  & 8  \\ \hline
        221 & turn-on-at-sunset--turn-on-by-zip-code              & 2 & 5  & 8  & 6  & 5  & 6  & 7  \\ \hline
  \end{tabular}
  }
  %}
%\vspace{-3em}
\end{table*}

\begin{table*}[!htb]
  \centering
  { \footnotesize
  \begin{tabular}{| r | p{42mm} | r | r | r | r | r | r | r |}
    \hline
    \textbf{No.} & \textbf{App} & \textbf{Evt.} & \multicolumn{3}{ c |}{\textbf{Without DPOR}} & \multicolumn{3}{ c |}{\textbf{With DPOR}}\\
    \cline{4-9}
        & & & \textbf{States} & \textbf{Trans.} & \textbf{Time} & \textbf{States} & \textbf{Trans.} & \textbf{Time}\\
        %& & & & & \textbf{(s)} & & & \textbf{(s)}\\
    \hline
        222 & turn-on-before-sunset--turn-on-by-zip-code          & 2 & 5  & 8  & 6  & 5  & 6  & 7  \\ \hline
        223 & brighten-dark-places--turn-it-on-when-it-opens      & 2 & 4  & 6  & 6  & 4  & 5  & 7  \\ \hline
        224 & brighten-dark-places--turn-on-by-zip-code           & 2 & 4  & 6  & 6  & 4  & 5  & 7  \\ \hline
        225 & brighten-dark-places--undead-early-warning          & 2 & 4  & 6  & 6  & 4  & 5  & 7  \\ \hline
        226 & brighten-my-path--turn-on-by-zip-code               & 2 & 4  & 6  & 6  & 4  & 5  & 6  \\ \hline
        227 & turn-it-on-when-it-opens--undead-early-warning      & 2 & 4  & 6  & 6  & 4  & 5  & 7  \\ \hline
        228 & turn-on-by-zip-code--turn-it-on-when-it-opens       & 2 & 4  & 6  & 6  & 4  & 5  & 7  \\ \hline
        229 & turn-on-by-zip-code--undead-early-warning           & 2 & 4  & 6  & 6  & 4  & 5  & 7  \\ \hline
  \end{tabular}
  }
  %}
%\vspace{-3em}
\end{table*}
}

%\newpage
%\appendix

\end{document}